\documentclass[a4paper,pra,aps,floatfix,superscriptaddress,onecolumn,longbibliography,notitlepage,footinbib,unpublished,11pt
]{quantumarticle}
\pdfoutput=1
\usepackage[utf8x]{inputenc}
\usepackage[english]{babel}
\usepackage[T1]{fontenc}
\usepackage{amsmath,amsthm,amsfonts,amssymb,xcolor}
\usepackage{graphicx}
\usepackage{floatrow}
\usepackage{subfig}
\usepackage{microtype}
\usepackage{braket}
\usepackage{physics}
\usepackage{bm}
\usepackage{scalefnt}
\usepackage{asymptote}
\usepackage{asypictureB}
\usepackage{float}
\usepackage{makecell}
\usepackage[numbers,sort&compress,square]{natbib}
\usepackage{multicol}
\usepackage{booktabs}
\usepackage[all,cmtip]{xy}
\usepackage{thmtools,thm-restate}
\usepackage[shortlabels]{enumitem}

\newcommand{\drawgenerator}[8]{%
\xymatrix@!0{%
& #8 \ar@{-}[ld]\ar@{.}[dd] \ar@{-}[rr] & & #7 \ar@{-}[ld]  \\%
#1 \ar@{-}[rr] \ar@{-}[dd] &  & #2 \ar@{-}[dd] &            \\%
& #6 \ar@{.}[ld] &  & #5 \ar@{-}[uu] \ar@{.}[ll]       \\%
#3 \ar@{-}[rr] &  & #4 \ar@{-}[ru]                       %
}%
}

\makeatletter
\def\l@subsubsection#1#2{}
\makeatother
\usepackage{lipsum}
\usepackage[splitrule]{footmisc}
\usepackage{xcolor}
\definecolor{darkblue}{RGB}{0,0,128}
\usepackage[breaklinks, colorlinks, linkcolor=blue, citecolor=red, filecolor=blue, urlcolor=darkblue,pdftitle={Test}]{hyperref}
\usepackage[capitalize,nameinlink]{cleveref}

\newtheorem{lemma}{Lemma}
\newtheorem{theorem}[lemma]{Theorem}

\theoremstyle{definition}

\newcommand{\lcm}{\text{lcm}}

\newcommand{\XY}{XY }
\newcommand{\XZZX}{XZZX }
\newcommand{\defcubic}{Clifford-deformed surface }

\newenvironment{aligns}{\subequations \align} {\endalign \endsubequations}

\definecolor{arpit}{RGB}{127,0,0}

\newcommand\blfootnote[1]{%
\begingroup
\renewcommand\thefootnote{}\footnote{#1}%
\addtocounter{footnote}{-1}%
\endgroup
}

\begin{document}
\title{Tailoring three-dimensional topological codes for biased noise}

\author{$\text{Eric Huang}^{\star,}$}
\affiliation{Perimeter Institute for Theoretical Physics, Waterloo, ON N2L 2Y5, Canada}
\affiliation{University of Maryland, College Park, MD 20742 USA}

\author{$\text{Arthur Pesah}^{\star,}$}
\affiliation{Department of Physics \& Astronomy, University College London, WC1E 6BT London, United Kingdom
}

\author{Christopher T. Chubb}
\affiliation{Institut quantique \& D\'epartement de physique, Universit\'e de Sherbrooke, Sherbrooke, QC J1K 2R1, Canada}
\affiliation{Institute for Theoretical Physics, ETH Z\"urich, 8093 Z\"urich, Switzerland}

\author{Michael Vasmer}
\affiliation{Perimeter Institute for Theoretical Physics, Waterloo, ON N2L 2Y5, Canada}
\affiliation{Institute for Quantum Computing, University of Waterloo, Waterloo, ON N2L 3G1, Canada}

\author{Arpit Dua}
\affiliation{Department of Physics and Institute for Quantum Information and Matter, California Institute of Technology, Pasadena, CA 91125, USA}
\blfootnote{*Equal contribution, alphabetical ordering.}

\begin{abstract}
Tailored topological stabilizer codes in two dimensions have been shown to exhibit high storage threshold error rates and 
improved subthreshold performance under biased Pauli noise. 
Three-dimensional (3D) topological codes can allow for several advantages including a transversal implementation of non-Clifford logical gates,
single-shot decoding strategies, parallelized decoding in the case of fracton codes as well as construction of fractal lattice codes.
Motivated by this, we tailor 3D topological codes for enhanced storage performance under biased Pauli noise. 
We present Clifford deformations of various 3D topological codes, such that they exhibit a threshold error rate of $50\%$ under infinitely biased Pauli noise. 
Our examples include the 3D surface code on the cubic lattice, the 3D surface code on a checkerboard lattice that lends itself
to a subsystem code with a single-shot decoder, the 3D color code, as well as fracton models such as the X-cube model, 
the Sierpiński model and the Haah code. 
We use the belief propagation with ordered statistics decoder (BP-OSD) to study threshold error rates at finite bias. 
We also present a rotated layout for the 3D surface code, which uses roughly half the number of physical qubits for the same code distance
under appropriate boundary conditions. 
Imposing coprime periodic dimensions on this rotated layout leads to logical operators of weight $O(n)$ at infinite bias and
a corresponding $\exp[-O(n)]$ subthreshold scaling of the logical failure rate, where $n$ is the number of physical qubits in the code. 
Even though this scaling is unstable due to the existence of logical representations with $O(1)$ low-rate Pauli errors, 
the number of such representations scales only polynomially for the Clifford-deformed code, leading to an enhanced effective distance.  
\end{abstract}

\maketitle

\section{Introduction}
Fault-tolerant quantum computation is a crucial ingredient for building a scalable quantum computer. 
Topological stabilizer codes are a highly-prized family of low-density parity-check (LDPC) codes due to their geometrically local parity checks, 
high storage threshold error rates and low-overhead fault-tolerant logical gate implementations. 
It has been found that topological codes can be tailored to noise to achieve higher success rates and
threshold error rates~\cite{AliferisPreskill_bias2008,Ultrahigh2018}. 
Generally, for evaluating the performance of a Pauli stabilizer code, a Pauli noise model is considered due to its efficient simulability. 
It has recently been suggested that Pauli noise can be biased towards dephasing in certain realistic laboratory qubits, 
or can be engineered to be so~\cite{aliferis2009fault, Nigg_2014,burkard2021,puri2020bias}. 
For biased Pauli noise, Clifford-deformed surface codes such as the XZZX, XY and (XYZ)$^2$ surface codes, the XYZ color code, and
families of randomly Clifford-deformed surface codes,
have been shown to exhibit high threshold error rates and enhanced subthreshold 
performance~\cite{XZZX2021, Ultrahigh2018,dua2022clifforddeformed,tiurevCorrectingNonindependentNonidentically2022, srivastava2021xyz, miguel2022cellular}.

Three-dimensional (3D) topological codes offer several advantages over all the known two-dimensional topological codes. 
Unlike 2D stabilizer codes, 3D stabilizer codes such as the 3D surface code allow for a transversal implementation of a non-Clifford gate and overall,
a fault-tolerant universal gate set~\cite{bombinTopologicalComputationBraiding2007,bombinGaugeColorCodes2015,kubicaUniversalTransversalGates2015,kubica2015unfolding,Vasmer2019}. 
Using 3D codes for storing logical information can also be advantageous in terms of decoding. 
For instance, the loop-like syndromes associated with 3D stabilizer codes such as the surface code and color code can be decoded using 
single-shot decoding strategies~\cite{breuckmann2017a,duivenvoordenRenormalizationGroupDecoder2019,kubica2019cellular,vasmer2021sweep,quintavalle2021single,higgottImprovedSingleshotDecoding2022}. 
Furthermore, 3D subsystem codes such as the gauge color code~\cite{bombinGaugeColorCodes2015,bombin2015a,brownFaulttolerantErrorCorrection2016} 
and the 3D subsystem surface code~\cite{kubica2021single} allow for a single-shot decoding strategy for general Pauli noise, 
where the noisy error syndrome need only be measured once
\footnote{
    This should be contrasted with 2D topological codes, where, in order to ensure fault-tolerance, 
    the noisy error syndrome must be measured $O(d)$ times for a distance $d$ code.
}. 
Such single-shot decoding strategies not only reduce the time overhead but also are more resilient to 
time-correlated noise~\cite{bombinResilienceTimeCorrelatedNoise2016}.
3D fracton topological codes such as the X-cube model (partially) allow for parallelized decoding in submanifolds of the lattice 
due to the mobility of syndromes being restricted to these submanifolds \cite{brown2020parallelized}.
Another recently discovered advantage of 3D codes such as the surface code is that they can be used to construct fractal lattice codes
by punching holes with appropriate boundary conditions~\cite{Zhu_fractal_2022,Dua_fractal_2022}.
Such fractal lattice surface codes allow for fault-tolerant universal quantum computation with a reduced space overhead 
and a single-shot decoding strategy that can be used for loop-like syndromes on the fractal geometry~\cite{Dua_fractal_2022}.
Even though designing a qubit architecture with a 3D connectivity is a serious experimental challenge in many quantum computing platforms, 
several recent advances have made the prospect of a 3D architecture more amenable to near-term experiments.
For instance, qubit shuttling has recently been shown to enable 3D connectivity on a 2D layout~\cite{cai2022looped}, 
and could for instance be implemented using silicon qubits~\cite{Buonacorsi_2019}, ion traps~\cite{akhtar2022high} or neutral atoms~\cite{lukingroupcoldatoms}.
Other platforms, such as 3D integrated superconducting qubits~\cite{mallek2021fabrication,2017_3DISC,IBM3D} 
and photonic qubits~\cite{bartolucci2021,bombin2021interleaving,Bourassa2021blueprintscalable,Tzitrin2021,bombin2DQuantumComputation2018} 
could also allow the realization of 3D codes.

Motivated by the advantages of 3D topological codes, we tailor them for enhanced storage performance in the presence of biased Pauli noise. 
We construct Clifford-deformed codes from the 3D surface code (cubic lattice)~\cite{Dennis2002}, 
the 3D color code~\cite{bombinExactTopologicalQuantum2007,bombinTopologicalComputationBraiding2007}, 
the X-cube fracton model~\cite{Vijay_2016}, the Sierpiński fracton model~\cite{Yoshida_2013} 
and the Haah code~\cite{Haah_2011}. 
We also propose a Clifford deformation of the 3D surface code on the checkerboard lattice, which lends itself 
to a subsystem code with a single-shot decoding strategy.
All of our Clifford-deformed codes allow for decoding strategies with a threshold error rate of $50\%$ at infinite dephasing bias. 
These Clifford-deformed codes are constructed to have materialized linear symmetries that allow for the first step of decoding to be a 
minimum-weight perfect matching (MWPM) decoder~\cite{Dennis2002,brown2022conservation} in submanifolds supporting those symmetries. 
The resulting models after this first step have further linear symmetries that allow for another round of the matching decoder. 
The combination of these steps gives the full matching decoder with a threshold error rate of $50\%$. 

For a subset of the codes mentioned above, we use the belief propagation with ordered statistics 
decoder~\cite{panteleev2019degenerate,roffe2020decoding} (BP-OSD) to numerically find the threshold error rates at finite bias. 
We discover that the threshold error rate increases with the bias for both original and Clifford-deformed codes. 
Beyond a critical value of bias, the threshold error rate of the Clifford-deformed codes exceeds that of the original codes. However, we also find some limitations of BP-OSD: for the X-cube model and the 3D surface code on the checkerboard lattice, 
the apparent threshold error rate tends to recede when increasing the system size, showing either a lower infinite-size threshold error rate
or no threshold at all. This effect, previously observed on 2D surface codes and color codes \cite{higgottImprovedSingleshotDecoding2022}, 
is analyzed in more detail in the context of 3D topological codes.
We also compare BP-OSD with the sweep-matching decoder---a combination of minimum-weight perfect matching
and the sweep decoder~\cite{kubica2019cellular,vasmer2021sweep}---for the 3D surface code (cubic lattice) 
and find better threshold error rates with BP-OSD.

Lastly,
we define a rotated layout for the 3D CSS surface code, which offers the same distances $(d_X,d_Z)$ 
for both Pauli $X$ and Pauli $Z$ logical operators as in the standard layout, while using roughly half the number of physical qubits,
in analogy with the 2D case~\cite{Tailoring2019}.
Using this rotated layout for the Clifford-deformed surface code and imposing coprime dimensions with appropriate boundary conditions 
leads to weight ${O}(n)$ logical operators at infinite dephasing bias. 
As a result, the subthreshold performance of the Clifford-deformed surface code is enhanced, with a logical error rate scaling 
as $e^{-{O}(n)}$ with the number of qubits $n$,
in comparison to $e^{-O(n^{1/3})}$ for the original code.

The paper is structured as follows: In \cref{sec:background}, we review
biased noise models and Clifford-deformed codes. We discuss the notions of materialized symmetries in the context 
of the 2D XZZX and XY surface codes, showing that both codes have a 50\% threshold error rate at infinite bias.
For the XY surface code, we introduce a technique called the weight reduction technique, 
that we also use for some of the 3D topological codes studied in this paper.
In \cref{sec:def3dcodes}, we present CSS and Clifford-deformed 3D topological codes, including the surface code on a cubic lattice,
the surface code on a checkerboard lattice, the color code, and fracton models (X-cube model, Sierpiński fracton model and Haah code). 
We describe the symmetries and prove the 50\% threshold error rates at infinite bias for each of our Clifford-deformed codes.
In \cref{sec:finbiasthresholds}, we present numerical threshold error rate estimates for some CSS and Clifford-deformed 3D codes
under finite bias, using BP-OSD and sweep-matching decoders. We also demonstrate some of the limitations of BP-OSD in the context
of 3D codes.
In \cref{sec:rotsts}, we explore further optimizations including a rotated
layout that reduces the number of physical qubits $n$,
and coprime lattice sizes that achieve logical error rates scaling as
$e^{-O(n)}$ in the subthreshold regime.
Finally in \cref{sec:discussion}, we discuss the implications of our findings. 

In \cref{sec:rhombic-thresholdproof-appendix},
we prove that the Clifford-deformed surface code on a
checkerboard lattice has a 50\% infinite-bias threshold error rate. 
In \cref{sec:decoders-appendix},
we discuss the decoders used in the numerical simulations for the threshold error rates at finite bias.
In \cref{sec:numerics-appendix},
we discuss the numerical methods used to estimate the threshold error rates at finite bias. We list the table of contents below. 

\tableofcontents

\section{Background}\label{sec:background}
\subsection{Biased Pauli noise}

In order to evaluate the performance of a code, it is conventional to choose a noise model
where each qubit is independently subjected to quantum noise. 
Upon the Pauli stabilizer measurements, such quantum noise is digitized to Pauli errors up to coherent rotation~\cite{bealeQuantumErrorCorrection2018,FernGeneralizedPerformance2006,GreenbaumCoherent2017,HuangCoherent2019}. 
To incorporate the coherent rotation in general is hard. 
However, the Pauli twirling approximation has been found to yield estimates of threshold error rates that are close to the ones 
found by including an efficiently tractable choice of coherent rotation~\cite{Bravyi_2018}. 
For our purposes, we consider a general single-qubit Pauli noise channel of the form

\begin{align}
    \rho \to (1 - p) \rho + p\left(r_X X\rho X + r_Y Y\rho Y + r_Z Z\rho Z\right)
\end{align}
where $p\in [0, 1]$ is the single-qubit \emph{physical error rate}
and the weights $r_X, r_Y, r_Z$ describe the relative probability of Pauli
errors $X$, $Y$, $Z$ respectively, such that $r_X, r_Y, r_Z\geq 0$ and $r_X + r_Y + r_Z = 1$. For
instance, a depolarising channel, with $r_X=r_Y=r_Z=1/3$, causes Pauli
errors $X$, $Y$, $Z$ with equal probabilities, and thus describes \emph{unbiased Pauli noise}.
On the other hand, a dephasing channel, with $r_X=r_Y=0$ and $r_Z=1$,
results in $Z$ errors alone and describes \emph{infinitely biased Pauli noise}.
The \emph{bias} $\eta_Z$ for dephasing is formally
defined~\cite{Tailoring2019} as the ratio
\begin{align}
    \eta_Z := \frac{r_Z}{r_X + r_Y},
\end{align}
such that (unbiased) depolarizing noise corresponds to $\eta_Z = 0.5$ while infinitely biased noise corresponds to $\eta_Z = \infty$.
In this work, we restrict our attention to noise channels with symmetric dephasing bias such that $r_X=r_Y$ and hence, 
the noise is characterized by physical error rate $p$ and dephasing bias $\eta_Z$.

\subsection{Clifford-deformed codes}
We consider a Pauli stabilizer code $\mathcal{C}$ with stabilizer group generated by stabilizers $\{S_i\}$ acting on the physical qubits $Q$. We can construct a new code $\mathcal{\widetilde{C}}$
by applying a Clifford circuit $U_C$ consisting of single-qubit Clifford operations on the physical qubits $Q$. We refer to this new code $\mathcal{\widetilde{C}}$ as a Clifford-deformed code. Under such a Clifford circuit, the generators $S_i$ are modified to
\begin{align}
    \widetilde{S_i} = U_C^\dag S_i U_C
\end{align}
which also form a set of commuting Pauli operators. For biased Pauli noise models, it can be advantageous to use a Clifford-deformed code.
Intuitively, if more stabilizer generators anticommute with the more common errors,
the resulting increase in nontrivial syndrome bits provides more information to the decoder to better estimate the correction operators.
Clifford deformations also have the effect of drastically reducing the number of $Z$-only or mostly-$Z$,
logical operators. As a result, this reduces the degeneracy of possible errors causing a given syndrome,
which can be exploited by decoders for superior performance.

Pioneering studies on Clifford-deformed codes for biased noise considered the XY surface code \cite{Ultrahigh2018} and the XZZX surface code
\cite{XZZX2021} in two spatial dimensions. The former is obtained from the CSS
surface code in 2D by replacing all Pauli $Z$s by $Y$s in the stabilizer
generators while the latter is obtained (from the CSS surface code) by applying
a Hadamard operation on half of the qubits such that all stabilizer generators become $X\otimes Z \otimes Z\otimes X$; see \cref{fig:xzzx-code}. 
Both these Clifford-deformed codes have threshold error rates that track the hashing bounds at finite dephasing bias and have a threshold error rate of 50\% at infinite dephasing bias. 
However, unlike the XY code, the XZZX code threshold error rates track the hashing bound for noise biased towards $X$ and $Y$ Pauli errors as well.
Recently, the performance of randomly Clifford deformations of the 2D surface code subjected to noise biased towards dephasing, has also been investigated~\cite{dua2022clifforddeformed}.
A phase of 50\% infinite-bias threshold error rates was found in the parameter space of
$(\Pi_{XZ},\Pi_{YZ})$, where $\Pi_{XZ} (\Pi_{YZ})$ is the probability of a Clifford operation that implements the permutation $X\leftrightarrow Z$ ($Y\leftrightarrow Z$).
This phase can be explained intuitively via a mapping to percolation problems. 
Moreover, certain randomly Clifford-deformed surface codes on $odd \times odd$ dimensions, outperformed the XZZX and XY codes with the same dimensions in the scaling of the subthreshold logical failure rate.

A suitable choice of Clifford deformation can result in a higher threshold error rate and can also improve subthreshold failure rates for finite system sizes. 
But subthreshold failure rates for finite system sizes are also sensitive to the dimensions and boundary conditions 
(note the mention of $odd \times odd$ dimensions in context of subthreshold failure rates in the previous paragraph). 
For instance, choosing coprime periodic dimensions for the XZZX code results in a subthreshold failure rate scaling of $e^{-{O}(n)}$ 
in comparison to that of $e^{-{O}(\sqrt{n})}$ for the CSS surface code.

For translation-invariant deformations such as the XZZX and XY surface code, the $50\%$ threshold error rates at infinite bias can be understood
in terms of the symmetries of the stabilizer group that appear in this noise regime.
We review this for the XZZX code and the XY code below.

\begin{figure}[t]
    \centering
    \subfloat[]{%
        \includegraphics[scale=1.1]{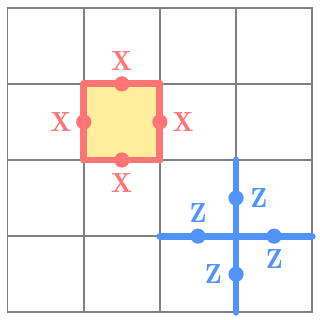}%
        \label{fig:xzzx-lattice-original}
    }
    \subfloat[]{%
        \includegraphics[scale=1.1]{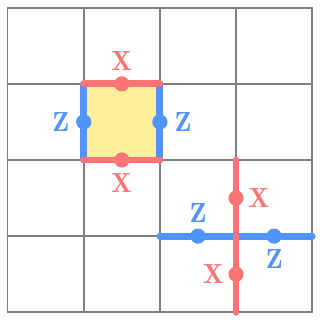}%
        \label{fig:xzzx-lattice-deformed}
    }
    \subfloat[]{%
        \includegraphics[scale=1.1]{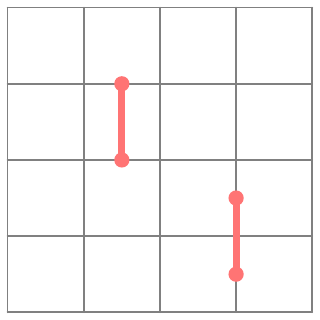}%
        \label{fig:xzzx-lattice-classical}
    }
    \subfloat[]{%
        \includegraphics[scale=1.1]{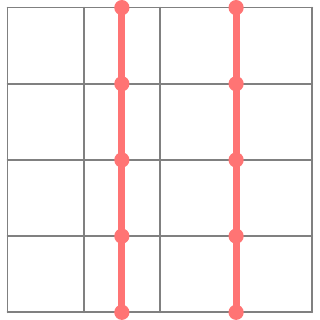}%
        \label{fig:xzzx-symmetries}
    }
    \caption{Illustration of the XZZX surface code.
    \textbf{(a)} Original surface code stabilizers, on a square lattice with periodic boundary conditions. 
    Qubits are on edges and stabilizers on faces and vertices.
    \textbf{(b)} XZZX surface code stabilizers, obtained by applying a Hadamard operation on all the vertical qubits.
    \textbf{(c)} At infinite $Z$ bias, we can ignore the $Z$ part of the stabilizers to form a classical code, whose
    parity-check operators (red edges) are supported on two qubits (red dots). 
    \textbf{(d)} The product of weight-2 parity checks along a vertical line is equal to the identity operator.
    This means that each vertical line contains an even number of excitations, which can be independently decoded by matching.
    We call this type of relation a materialized linear symmetry.
    }%
    \label{fig:xzzx-code}
\end{figure}

\begin{figure}
    \centering
    \subfloat[]{%
        \includegraphics[scale=1.2]{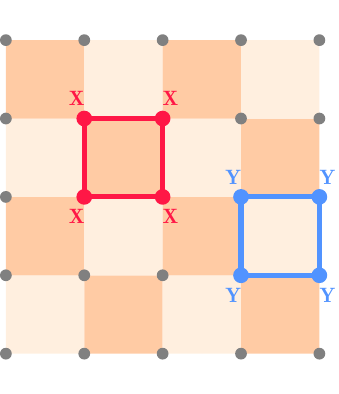}%
        \label{fig:xy-lattice}
    }\hspace{4mm}
    \subfloat[]{%
        \includegraphics[scale=1.2]{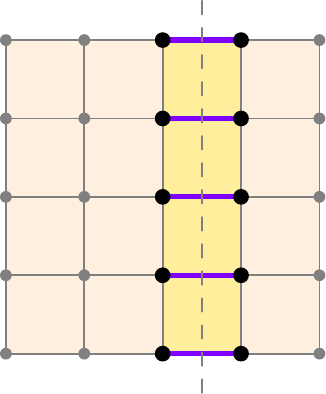}%
        \label{fig:xy-symmetry-1}
    }
    \hspace{4mm}
    \subfloat[]{%
        \includegraphics[scale=1.2]{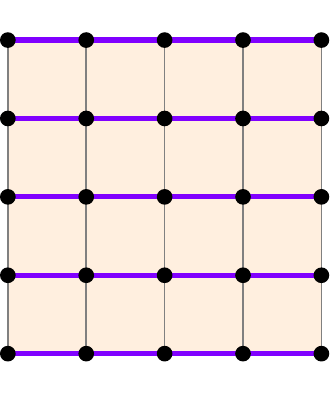}%
        \label{fig:xy-symmetry-2}
    }\\
    \vspace{5mm}
    \subfloat[]{%
        \includegraphics[scale=1.2]{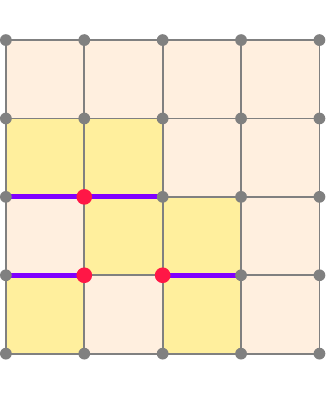}%
        \label{fig:xy-example-step-1}
    }\hspace{4mm}
    \subfloat[]{%
        \includegraphics[scale=1.2]{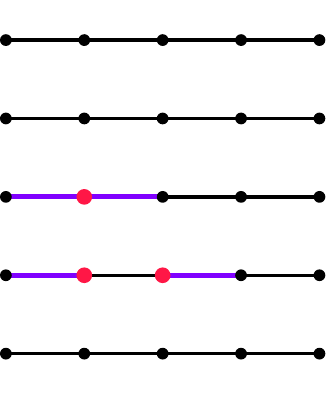}%
        \label{fig:xy-example-step-2}
    }\\
        \vspace{5mm}
    \caption{Illustration of the weight-reduction technique on the XY code.
    \textbf{(a)} XY code lattice. Qubits are on the vertices. Dark (resp.
    light) plaquettes are stabilizer generators made of $X$ (resp. $Y$) operators.
    This code can be obtained by applying a phase gate and a Hadamard gate to every qubit of the CSS surface code.
    \textbf{(b)} Linear symmetry of the code at infinite $Z$ bias, assuming periodic boundary conditions.
    In this regime, the XY code becomes a classical code with the same four-body checks on every square.
    We can interpret each column as a repetition code (see dashed line for an example), where the purple edges are the variables and the yellow squares the checks.
    Indeed, the product of the yellow checks along the dashed line is effectively identity for pure $Z$ noise.
    Decoding this repetition code gives us the parity of the two qubits on each purple edge.
    \textbf{(c)} Second level of repetition codes. Once matching has been performed on all the columns, we obtain new two-body checks on all
    the horizontal edges (purple). They themselves form linear symmetries along each row.
    Matching along these symmetries allows one to decode errors on all the qubits. 
    \textbf{(d)} Example of decoding using the weight-reduction strategy. Errors correspond to red vertices and face excitations to yellow squares.
    In the first decoding step, we use minimum-weight perfect matching between squares along each vertical line. 
    The resulting ``corrections'' are the purple edges. This means that one of the two qubits of each purple edge
    is predicted to have an error.
    \textbf{(e)} In the second decoding step, the purple edges are reinterpreted as excitations of some horizontal repetition codes.
    Using matching on each repetition code gives the desired correction.
    }%
    \label{fig:xy-code}
\end{figure}

\subsection{XZZX code: materialized symmetries and conserved quantities under biased noise}

The XZZX surface code can be obtained by applying a Hadamard operator on every vertical qubit of the CSS surface code.
The resulting stabilizers are shown in \cref{fig:xzzx-lattice-deformed}.
At infinite $Z$ bias, the action of noise on the qubits where a stabilizer acts as $Z$ results in no syndrome.
Hence, the stabilizer effectively acts as identity on these qubits for infinite-bias noise.
As a result, the code becomes equivalent to a classical code made of two-body parity-check operators $B_f$ and $A_v$, 
as illustrated in \cref{fig:xzzx-lattice-classical}.

This resulting effective model has the following relations on each vertical line $\ell$ of the lattice:
\begin{aligns}
    \label{eq:xzzx-bf-relation}
    \prod_{f \in \ell} B_f = I, \\
    \label{eq:xzzx-av-relation}
    \prod_{v \in \ell} A_v = I.
\end{aligns}
These relations, represented in \cref{fig:xzzx-symmetries}, are referred to as \emph{materialized subsystem symmetries} \cite{brown2020parallelized, brown2022conservation}.
Because of these symmetries, the syndrome values $b_f$ and $a_v$ of all the stabilizers under infinite-bias noise obey
\begin{aligns}
    \prod_{f \in \ell} b_f = 1, \\
    \prod_{v \in \ell} a_v = 1,
\end{aligns}
which are the conservation laws associated with the materialized subsystem symmetries.
This implies that at infinite $Z$ bias, the number of excitations of the XZZX stabilizer generators along any vertical line is even.
In other words, single stabilizer generator syndromes can only ``move'' along
vertical lines at infinite bias under application of $Z$ errors.
This leads to a simple decoding strategy for an XZZX code subjected to $Z$ errors: we match the syndromes along each line independently.
This is equivalent to decoding $L$ independent classical repetition codes, each of size $L$.
Since the repetition code has a threshold error rate of $50\%$, the success rate of this infinite-bias decoder is lower-bounded by:
\begin{align}
    p_{\text{success}} \geq \left(1- A e^{-\alpha L}\right)^L \approx 1 - A L e^{-\alpha L} \xrightarrow[L \to \infty]{} 1
\end{align}
for any fixed physical error rate below $50\%$, where $A$ and $\alpha$ are two positive constants. The reason this expression is a lower bound and not an equality is that failing to decode an even number of repetition codes also results in successful decoding.
Thus this strategy results in a threshold error rate of $50\%$~\cite{XZZX2021}.

\subsection{XY code: weight-reduction technique}
\label{sec:xy-code}

The XY surface code \cite{Ultrahigh2018} is another example of a Clifford-deformed surface code that has a threshold error rate of $50\%$ at infinite bias.
It is formed by applying an $S$ gate and a Hadamard gate on every qubit, having
the effect of turning all $Z$ stabilizers into $Y$ stabilizers,
as shown in \cref{fig:xy-lattice}.
As a result, at infinite $Z$ bias, the code becomes a classical code where the parity checks are four-body terms supported on every square.

To prove that this code has a $50\%$ infinite-bias threshold error rate, we need to generalize the technique developed for the XZZX surface code.
As can be seen in \cref{fig:xy-symmetry-1}, the checks form linear symmetries along both the vertical and horizontal directions.
However, in contrast to the XZZX code, these symmetries involve weight-4 checks, making the underlying decomposition into repetition codes
less obvious to see.

To decompose the code into repetition codes, we use a two-step decoding strategy that we call the \textit{weight-reduction technique},
which we use extensively in our study of 3D codes.

In the first step, we start by writing each square check as the parity of its two incident horizontal edge checks, 
where an edge check is defined as the parity of its two incident vertices.
In \cref{fig:xy-symmetry-1}, these horizontal edge checks, surrounding the top and bottom part of each square, are represented in purple.
Note that this is a purely formal manipulation, as the edge checks are not part of the syndrome at the moment.

We now use the linear symmetries consisting of the product of square checks along any vertical line.
An example of such a symmetry is highlighted in \cref{fig:xy-symmetry-1}.
These linear symmetries give rise to repetition codes, where the data bits are the horizontal edge checks and the parity checks are
the squares acting on two neighboring edges.
Successful decoding of these repetition codes allows us to obtain the value of all the horizontal edge checks, 
which therefore become part of the syndrome.
In other words, assuming a successful matching, we turn the weight-4 checks into weight-2 checks 
supported on the horizontal edges of the code.
This is the core of the weight-reduction technique.

The second decoding step starts by noticing that the new horizontal edge checks form a linear symmetry
on each horizontal line: the product of edges along any horizontal line is equal to the identity.
Each horizontal line can therefore be interpreted as a repetition code, and decoding all the repetition codes
allows us to correct errors on all the qubits. 
Those linear symmetries are illustrated in \cref{fig:xy-symmetry-2}.
An example of decoding with this two-step strategy is shown in \cref{fig:xy-example-step-1} and \ref{fig:xy-example-step-2}.

We now show that this decoding strategy leads to a threshold error rate of $50\%$.
For a given physical error rate below $50\%$, the probability that this decoder succeeds is lower-bounded by the probability that
both the $L$ 1D matchings of the first step and the $L$ 1D matchings of the second step are successful. 
More precisely, for a fixed physical error rate below $50\%$, we have
\begin{align}
    p_{\text{success}} \geq \left(1- A e^{-\alpha L}\right)^{2L} \xrightarrow[L \to \infty]{} 1,
\end{align}
where $A$ and $\alpha$ are positive constants. 
This shows that the threshold error rate of the $XY$ code under infinite $Z$ bias is $50\%$.

\section{3D Clifford-deformed topological codes}\label{sec:def3dcodes}

\begin{figure}[t]
    \centering
    \subfloat[]{
        \includegraphics[scale=1.3]{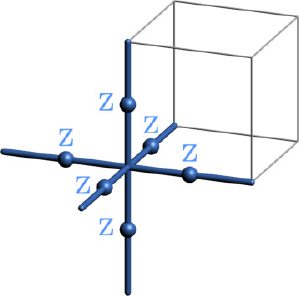}
        \includegraphics[scale=1.3]{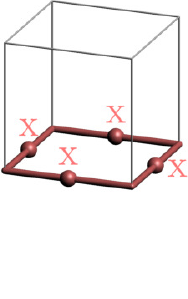}
        \includegraphics[scale=1.3]{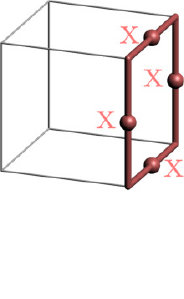}
        \includegraphics[scale=1.3]{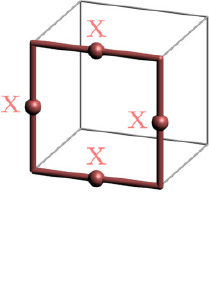}
        \label{fig:toric-3d-stabilizers-original}
    }
    \hfill
    \subfloat[]{
        \includegraphics[scale=1.3]{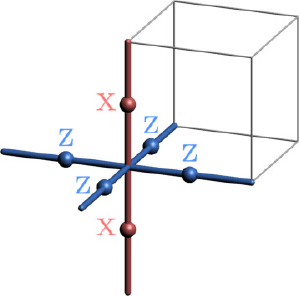}
        \includegraphics[scale=1.3]{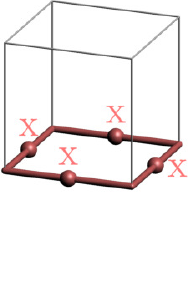}
        \includegraphics[scale=1.3]{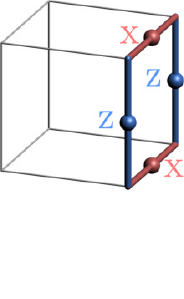}
        \includegraphics[scale=1.3]{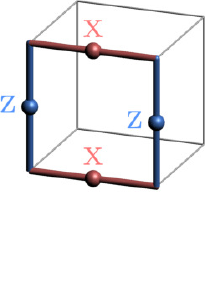}
        \label{fig:toric-3d-stabilizers-deformed}
    }
    \caption{%
    \textbf{(a)} Stabilizer generators in the CSS 3D surface code are
    vertex operators on each vertex and face operators for each face,
    for which there are three orientations.
    \textbf{(b)} Stabilizer generators in the Clifford-deformed surface code obtained via Hadamard operations on qubits along vertical edges.}
    \label{fig:toric-3d-stabilizers}
\end{figure}

In this section, we present Clifford deformations of 3D codes with a macroscopic number of materialized symmetries at infinite bias.
Having a macroscopic number of such symmetries leads to a macroscopic number of conservation laws obeyed by the syndromes, and this can lead 
to a decoder with a high threshold error rate.

Our general strategy for showing that a Clifford-deformed code has a threshold error rate of $50\%$ at infinite $Z$ bias is the following.
We start by identifying the linear symmetries of the code when the $Z$ part of the stabilizers is ignored.
Each linear symmetry gives rise to a repetition code decoding problem.
We can therefore construct a decoder that starts by a round of minimum-weight perfect matching (MWPM) decoding~\cite{Dennis2002,brown2022conservation} 
on the one-dimensional submanifolds supporting the symmetries.
This results in a new model with parity-check operators of reduced weight, as demonstrated for the XY code in \cref{sec:xy-code}.
We then identify the linear symmetries of this new model with reduced-weight stabilizers, 
decode the corresponding repetition codes using MWPM, and repeat the process until a correction operator has been
assigned to all the qubits.
Due to the fact that each step of this decoding strategy consists of decoding repetition codes,
the existence of such a decoder for any Clifford-deformed code shows that the overall threshold error rate of the code is $50\%$.

We now present our Clifford-deformed codes and their explicit decoders one by one below. 

\subsection{3D surface code}

The conventional form of the 3D surface code is defined on a cubic lattice with qubits sitting on edges.
The stabilizer generators consist of the vertex operators
$A_v=\prod_{e\in v} Z_e$, made of Pauli $Z$ operators on each of the six edges $e$ adjacent to a vertex $v$,
and the face operators $B_f=\prod_{e\in f}X_e$, made of Pauli $X$ operators on each of the four edges adjacent to a face $f$,
as illustrated in \cref{fig:toric-3d-stabilizers-original}.
The syndromes associated with violations of vertex stabilizer generators are point-like
and created in pairs at the boundaries of strings of $X$ errors.
The syndromes associated with violations of face stabilizers are loop-like syndromes and created at the boundaries of membranes of $Z$ errors,
as shown in \cref{fig:toric-3d-errors}.
The logical $X$ operators are topologically nontrivial string operators and the logical $Z$ operators are topologically nontrivial membranes.
In particular, on a 3D torus, there are three pairs of inequivalent minimum-weight logical operators $\overline{X}_u, \overline{Z}_u$,
one for each axis $u \in \{ \hat{x},\hat{y},\hat{z} \}$, where $\overline{X}_u$ is a string of $X$ operators oriented along $u$,
and $\overline{Z}_u$ is a membrane orthogonal to $u$. Examples are shown in \cref{fig:toric-3d-logicals}.
Thus, the 3D surface code encodes three logical qubits and has code distance $\min(L_x, L_y, L_z)$, specified by the minimum weight of the string operators.
One can define a open boundary version of the code on a cubic lattice with rough boundaries on a pair of opposite faces and smooth boundaries
on remaining four sides. The logical string operator connects the rough boundaries while the logical membrane operator connects the smooth boundaries.
Hence, the code encodes one logical qubit.

\begin{figure}[t]
    \centering
    \subfloat[]{
        \includegraphics[scale=1.2]{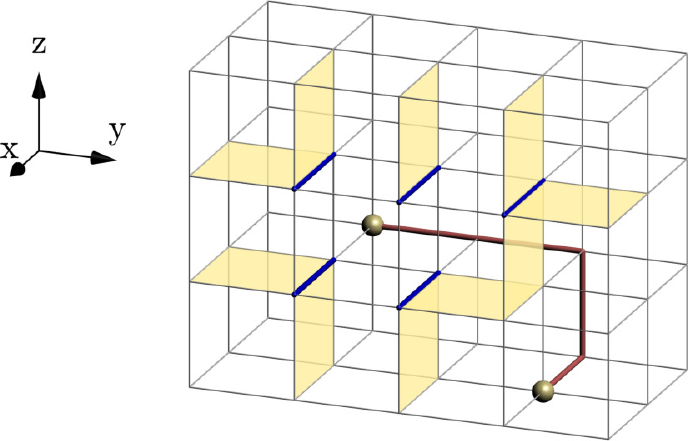}
        \label{fig:toric-3d-errors}
    }
    \subfloat[]{
        \includegraphics[scale=1.2]{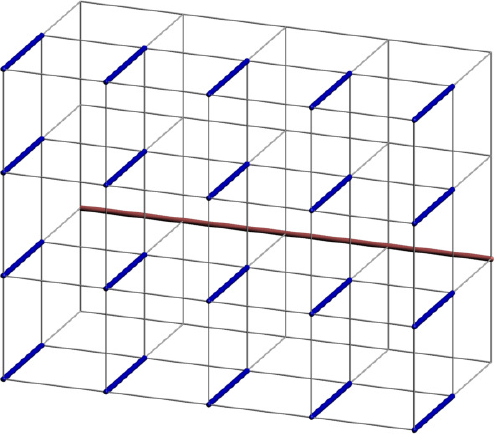}
        \label{fig:toric-3d-logicals}
    }
    \caption{Errors and logical operators of the CSS 3D surface code.
    \textbf{(a)} Errors in the 3D surface code create two types of syndromes:
    point-like syndromes at the boundary of chains of $X$ errors, and loop-like syndromes around membranes of $Z$ errors.
    \textbf{(b)} Examples of logical $X$ and $Z$ operators of the 3D surface code.
    } %
    \label{fig:toric-3d-errors-logicals}
\end{figure}
\begin{figure}[t]
    \centering
    \subfloat[]{
        \includegraphics[scale=1.2]{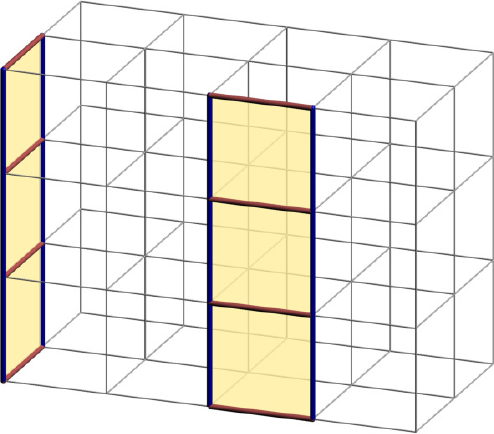}
        \label{fig:toric-symmetries-a}
    }
    \subfloat[]{
        \includegraphics[scale=1.2]{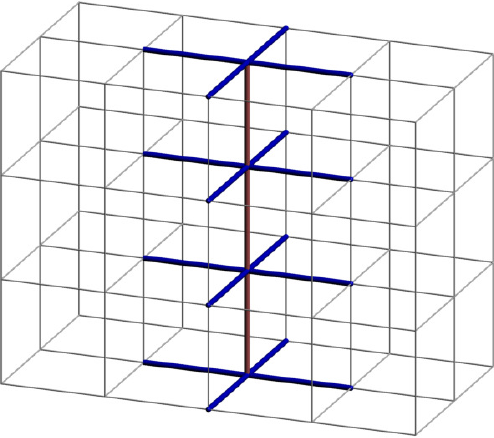}
        \label{fig:toric-symmetries-b}
    }
    \caption{Materialized linear symmetries of the Clifford-deformed 3D surface code.
    Products of plaquette/vertex operators along vertical lines are made only of $Z$ operators,
    and are therefore effectively equivalent to the identity in the purely $Z$
    infinite-bias noise regime.
    \textbf{(a)} Product of vertical plaquette operators (yellow) along two vertical lines.
    \textbf{(b)} Product of vertex operators along a vertical line
    } %
    \label{fig:toric-symmetries}
\end{figure}
\begin{figure}[t]
    \centering
    \subfloat[]{
        \includegraphics[scale=0.94]{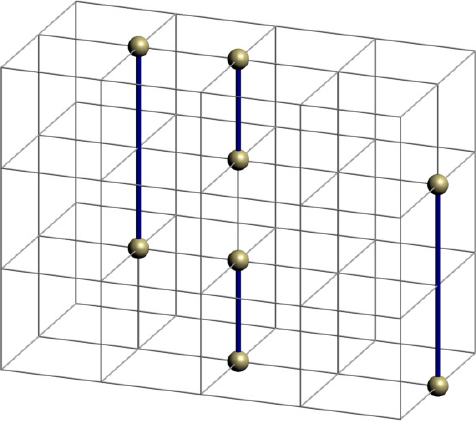}
        \label{fig:toric-3d-decoder-a}
    }
    \subfloat[]{
        \includegraphics[scale=0.94]{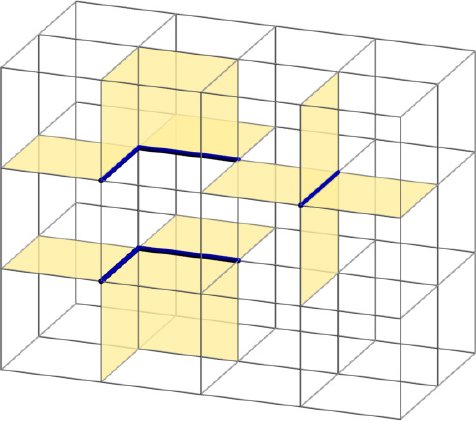}
        \label{fig:toric-3d-decoder-b}
    }
    \subfloat[]{
        \includegraphics[scale=0.94]{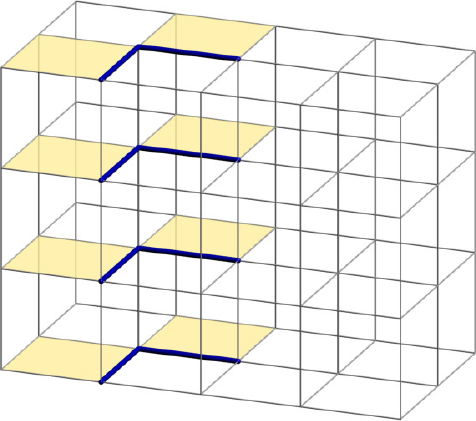}
        \label{fig:toric-3d-decoder-c}
    }
    \caption{Decoder used to prove that the Clifford-deformed 3D surface code has a threshold error rate of $50\%$ at infinite $Z$ bias.
    \textbf{(a)} Decoding of syndromes on vertices. Under local Pauli errors, vertex syndromes (yellow spheres) can only move in the vertical direction due to
    the linear symmetries of \cref{fig:toric-symmetries-b}.
    We can therefore decode errors on vertical qubits by matching vertex syndromes along this axis.
    \textbf{(b)} Decoding of syndromes on faces. Under local Pauli errors, $xz$ and $yz$ face syndromes (shown as yellow squares) can only move in the vertical direction due to
    the linear symmetries of \cref{fig:toric-symmetries-a}. We can therefore decode errors on the horizontal qubits by matching
    the face syndromes along the vertical axis.
    \textbf{(c)} Decoding of residual face syndromes. The decoder can be further improved by decoding residual errors coming from failed matchings in (b).
    A failed matching results in an identical $xy$ face syndrome on every $xy$ plane.
    These can be decoded using a 2D minimum-weight perfect matching algorithm.}%
    \label{fig:toric-3d-decoder}
\end{figure}

We consider a Clifford deformation of the 3D surface code where we apply a Hadamard operator on all the qubits on edges oriented along the $z$ direction,
which we call \emph{vertical qubits}.
On the contrary,
\emph{horizontal qubits} which reside on edges oriented along the $x$ or $y$
directions remain untouched by the Clifford deformation.
The resulting stabilizer generators are shown in \cref{fig:toric-3d-stabilizers-deformed}. The code has certain linear materialized symmetries for infinitely biased noise. Following the techniques that utilized materialized symmetries at biased noise to define decoding strategies for two-dimensional topological codes~\cite{Tailoring2019,XZZX2021,brown2022conservation,brown2020parallelized}, we prove the following theorem. 

\begin{theorem}
The Clifford-deformed 3D surface code has a threshold error rate of $50\%$ under pure $Z$ noise.
\end{theorem}

\begin{proof}
This code has linear materialized symmetries as shown in \cref{fig:toric-symmetries}.
The first set of symmetry operators consists of products of vertex stabilizers along a one-dimensional closed cycle
in the $z$ direction.
These products consist solely of Pauli $Z$ operators and hence effectively act as identity at infinite $Z$ bias.
The other symmetry consists of products of XZZX face stabilizers in the $xz$ and $yz$ planes along vertical lines. 
Due to the conservation laws obeyed by the syndrome along these symmetry lines,
we can independently match excitations along these lines as explained below. 

At infinite bias, we have only Pauli $Z$ errors, which anticommute with only $X$ Pauli operators in the stabilizer generators. Hence one can ignore the $Z$ terms and consider only anticommutation between the
$Z$ errors and the $X$ stabilizer generators.
This is equivalent to considering a classical parity-check code where $Z$ errors
are detected by a parity-check matrix that denotes the location of $X$ terms in
the stabilizer generators.
Thus the Clifford-deformed 3D surface code becomes a classical code, with weight-2 checks on vertices, $xz$ faces, and $yz$ faces.
These checks form the linear symmetries discussed earlier and illustrated in \cref{fig:toric-symmetries}.
Errors on the qubits oriented in the $z$ direction can be decoded by performing matching on the vertices along the corresponding symmetry lines.
An example of decoding of the errors that create syndromes of vertex operators is illustrated in \cref{fig:toric-3d-decoder-a}.
Errors on qubits oriented in the $x$ and $y$ directions can be decoded by performing
matching on the $xz$ and $yz$ faces along each vertical line respectively.
An example of decoding of face syndromes is illustrated in \cref{fig:toric-3d-decoder-b}.
If all these $3L^2$ matchings succeed, by correctly identifying the position of the errors, we have succeeded in decoding all the qubits.
Therefore, the probability of success is lower-bounded by the probability of succeeding in all these $3L^2$ matchings.
Hence, for a fixed physical error rate below $50\%$, we have
\begin{align}
    p_{\text{success}} \geq \left(1- A e^{-\alpha L}\right)^{3L^2} \xrightarrow[L \to \infty]{} 1
\end{align}
where $\alpha$ and $A$ are positive constants. Therefore, the code has a threshold error rate of $50\%$ at infinite $Z$ bias.
\end{proof}

While the decoder described in the proof is sufficient to get a $50\%$ threshold error rate, note that it can be further improved by taking into account the $xy$ face syndrome information as well. Let us consider a failed matching during the face decoding step.
By definition, it results in a line of errors on the horizontal qubits along the vertical direction, as illustrated in \cref{fig:toric-3d-decoder-c}.
It means that all the $xy$ face syndromes are identical on every $xy$ plane. Since we have a 2D materialized symmetry on $xy$ planes, we can decode them using a 2D minimum-weight perfect matching algorithm to return to the codespace.

\subsubsection{3D surface code on the checkerboard lattice}
Topological codes can be defined on various lattices or triangulations of a manifold for the same topological order.
The 3D surface code on a checkerboard lattice is shown in \cref{fig:rhombic-stabilizers-original}.
The code is defined using $Z$ cube stabilizers on one sublattice and $X$ triangle stabilizers on the other sublattice.

\begin{figure}[t]

    \subfloat[]{
        \includegraphics[scale=1.15]{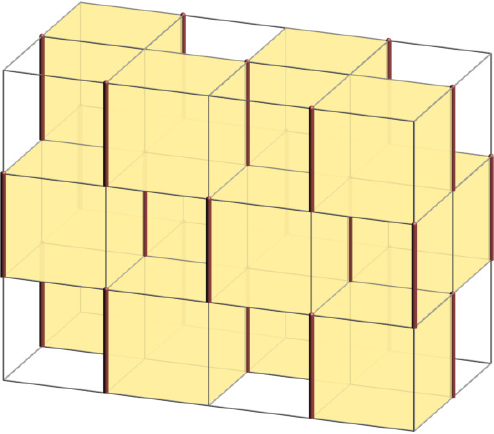}
        \label{fig:rhombic-checkerboard-deformation}
    }\\
    \subfloat[]{
        \includegraphics{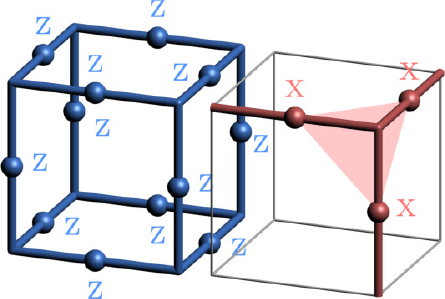}
        \label{fig:rhombic-stabilizers-original}
    }
    \subfloat[]{
        \includegraphics{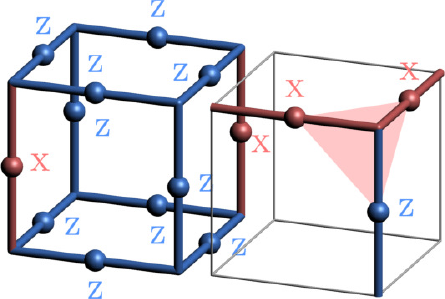}
        \label{fig:rhombic-stabilizers-def}
    }
    \caption{3D surface code on the checkerboard lattice.
    \textbf{(a)} Checkerboard lattice. 
    The yellow cubes represent the cube stabilizers and the empty cubes represent where triangle stabilizers would go.
    They have been omitted for clarity, but note that there are 8 of them in each empty cube.
    \textbf{(b)} Original CSS stabilizers.
    \textbf{(c)} Stabilizers of the Clifford-deformed code. The red edges highlighted in \textbf{(a)} are the edges where the $X\leftrightarrow Z$ Clifford deformation is applied.
    }
    \label{fig:rhombic-stabilizers}
\end{figure}

\begin{figure}[t]
    \centering
    \subfloat[]{
        \includegraphics[scale=1.15]{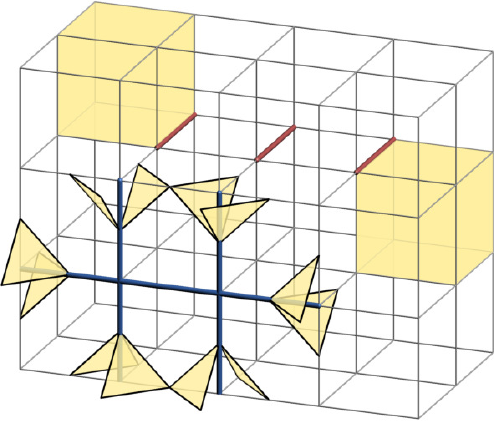}
        \label{fig:rhombic-errors}
    }
    \hspace{1cm}
    \subfloat[]{
        \includegraphics[scale=1.15]{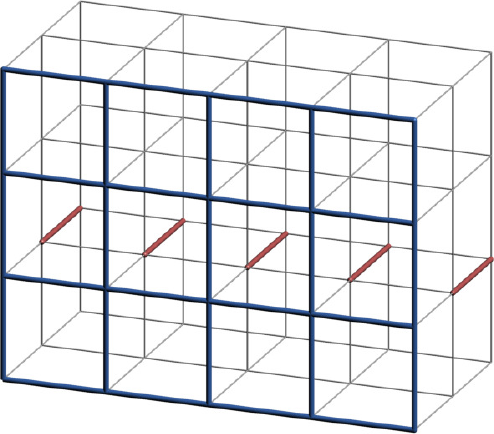}
        \label{fig:rhombic-logicals}
    }
    \caption{Errors and logical operators of the surface code on a
    checkerboard lattice.
    \textbf{(a)} Errors in the surface code on a checkerboard lattice with
    non-trivial syndromes (yellow).
    \textbf{(b)} Example of $X$ and $Z$ logical operators of the surface code
    on a checkerboard lattice.
    }
\end{figure}

Such a variant of the 3D surface code was used in the construction of 3D surface codes with a transversal CCZ gate~\cite{Vasmer2019}.
Moreover, these checkerboard lattice codes are instrumental in the construction of the CCZ gate for the 2D surface code~\cite{BrownCCZ},
and the 3D subsystem surface code~\cite{kubica2021single}.
This surface code variant is defined on a cubic lattice of even dimensions with qubits sitting on edges.
The cubic cells of the checkerboard lattice come in two colors.
Half of these cells (i.e.\ of one color) have a 12-body $Z$ cube stabilizer supported on them i.e.\ $A_c = \prod_{e \in c} Z_e$ is product of $Z$ operators over the twelve edges of the cube.
The other half of the cubic cells each have eight triangle-shaped stabilizer operators, associated with the eight vertices of the cell.
A triangle stabilizer on a vertex $v$ of a cubic cell $c$ is defined as the product of three $X$ operators adjacent to $v$ and contained in $c$, $B_{c,v}= \prod_{e \in c \cap v} X_e$.
The stabilizer generators are illustrated in \cref{fig:rhombic-stabilizers}.
Since the topological order is independent of the lattice details, the syndromes of this code also come in point-like and loop-like flavors.
The syndromes of the cube stabilizers are point-like and created at the ends of a string of Pauli $X$ errors.
The syndromes associated with the triangle stabilizers form a loop around membranes of $Z$ errors, as
shown in \cref{fig:rhombic-errors}.

The checkerboard lattice surface code
also encodes three logical qubits on an $\emph{even}\times \emph{even} \times \emph{even}$ torus with the logical operator pairs
consisting of nontrivial $X$ strings and $Z$ membranes along and orthogonal to three lattice directions respectively (see \cref{fig:rhombic-logicals}).
We consider a Clifford deformation of the checkerboard lattice surface code which consists of applying a Hadamard operation on half of the vertical qubits, 
in a three-dimensional checkerboard manner (see \cref{fig:rhombic-checkerboard-deformation,fig:rhombic-stabilizers-def}). This Clifford-deformed checkerboard lattice surface code has a $50\%$ threshold error rate under pure $Z$ noise. The proof is presented in \cref{sec:rhombic-thresholdproof-appendix}.

\subsection{3D color code}

The 3D color code can be defined on any 4-valent 3D lattice whose cells are 4-colorable, i.e.\ one should be able to assign one of four colors to each of the cells such that any two cells sharing a face have different colors.
Here, we study the 3D color code defined on the truncated octahedral lattice with periodic boundary conditions,
as shown in \cref{fig:3d-color-code-lattice}. In this lattice, each cell is a truncated octahedron, 
made of 24 vertices, 6 square faces and 8 hexagonal faces. $X$ stabilizer generators are defined on every 
cell, and $Z$ stabilizers on every face, as shown in \cref{fig:3d-color-code-stabilizers}. 
Coloring the cells using yellow, red, blue and green, we can describe the lattice as the interlacing of
a red-yellow and a blue-green sublattice, where cells of each sublattice are connected via square faces.
Cells belonging to different sublattices are connected via hexagonal faces.
Here, we use the convention of describing faces by the two colors of their adjacent cells. For instance,
a face at the intersection of a yellow cell and a red cell is called a yellow-red face.

\begin{figure}
    \centering
    \subfloat[]{
        \includegraphics{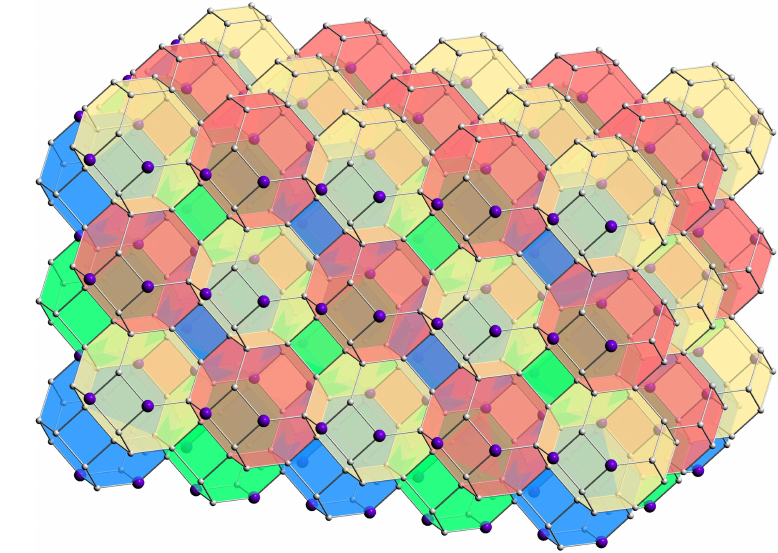}
        \label{fig:3d-color-code-lattice}
    }
    \subfloat[]{
        \includegraphics{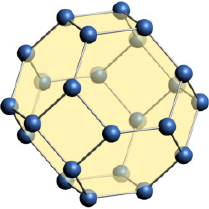}
        \includegraphics{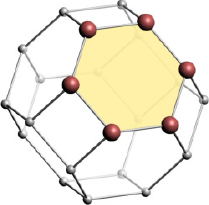}
        \includegraphics{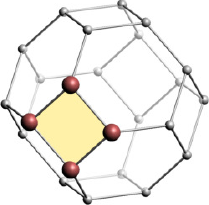}
        \label{fig:3d-color-code-stabilizers}
    }
    \caption{3D color code on a truncated octahedral lattice with periodic boundary conditions.
    \textbf{(a)} Truncated octahedral lattice, where qubits live on vertices and stabilizers live on cells
    and faces. The purple vertices correspond to qubits that we choose to Clifford-deform with a Hadamard operation.
    \textbf{(b)} Original stabilizers. $Z$ stabilizer generators are supported on the 24 qubits (blue vertices) of every cell.
    $X$ stabilizer generators are supported on every face, both hexagonal and square.
    }
    \label{fig:3d-color-code-lattice-stabilizers}
\end{figure}
\begin{figure}
    \centering
    \subfloat[]{
        \includegraphics{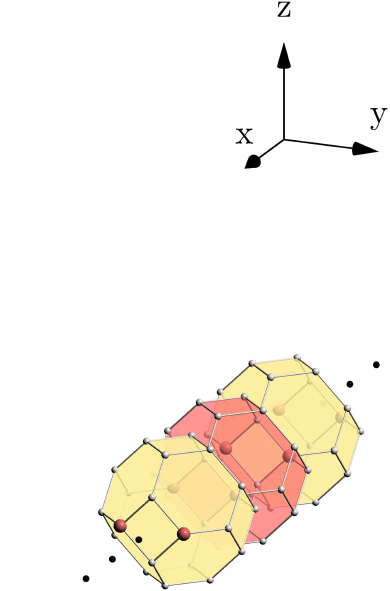}
        \label{fig:3d-color-code-cell-symmetries-a}
    }
    \subfloat[]{
        \includegraphics{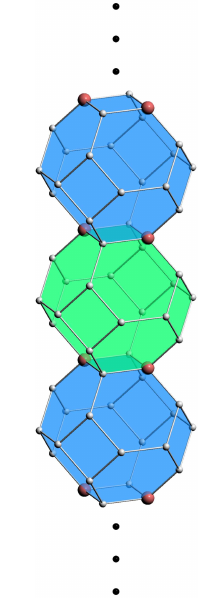}
        \label{fig:3d-color-code-cell-symmetries-b}
    }
    \subfloat[]{
        \includegraphics{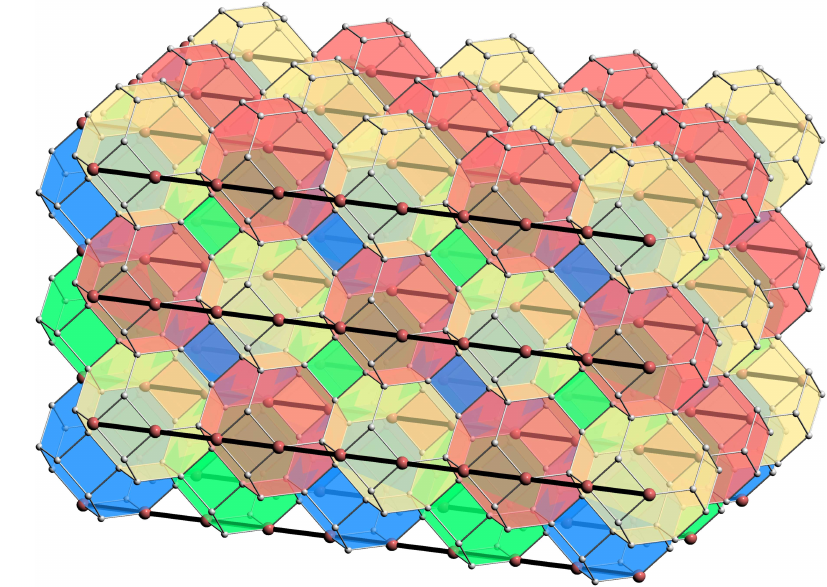}
        \label{fig:3d-color-code-cell-symmetries-c}
    }
    \caption{Linear symmetries of the cells.
    Red qubits are the Clifford-deformed ones, which effectively become the support of the cell stabilizers in the infinite $Z$ bias regime.
    \textbf{(a)} Symmetry of the yellow and red cells along the $x$ axis.
    \textbf{(b)} Symmetry of the blue and green cells along the $z$ axis.
    \textbf{(c)} Matching along the symmetries (a) and (b) gives rise
    to 2-body checks that form new linear symmetries along the $y$ axis (black lines)
    }
    \label{fig:3d-color-code-cell-symmetries}
\end{figure}
There exists a mapping between color codes and surface codes in any spatial dimension.  A string of $X$ errors also produces a pair of point-like syndromes on 3-cells,
and a membrane of $Z$ errors also produces a loop-like syndromes on 2-cells (faces). If we impose periodic boundary conditions on the lattice defined above, the code encodes 9 logical qubits, with three $X$ string logical operators on each direction, and three $Z$ membrane logical operators on each plane. The similarities between the two codes can be understood using a folding procedure, which maps three copies
of the 3D surface code to the 3D color code \cite{kubica2015unfolding,vasmerMorphingQuantumCodes2022}. 
However, the 3D color code has some unique properties, such its transversal $T$ gate and its flexible subsystem variant, making it a competitive candidate for a practical 3D code.

To tailor the 3D color code to biased noise, we select all yellow-red squares normal to the $x$ direction,
and apply a Hadamard to diagonally opposite qubits of each square, as illustrated in \cref{fig:3d-color-code-lattice}.
We now show that the resulting code has a $50\%$ threshold error rate at infinite $Z$ bias.

\begin{theorem}
    The Clifford-deformed 3D color code has a threshold error rate of $50\%$ under pure $Z$ noise
\end{theorem}

\begin{proof}
We start by decoding the syndromes on the 3-cells, effectively supported on four qubits at infinite $Z$ bias (the purple qubits in \cref{fig:3d-color-code-lattice}).
For this, we exploit two materialized linear symmetries, represented in \cref{fig:3d-color-code-cell-symmetries-a,fig:3d-color-code-cell-symmetries-b}, along the $x$ and $z$ directions respectively.
By the weight-reduction technique, matching cell syndromes along these two directions results in new weight-2 checks, which form linear symmetries along the $y$ axis,
as shown in \cref{fig:3d-color-code-cell-symmetries-c}.
Matching along these resulting linear symmetries completes the decoding of the syndromes on the 3-cells. Since all steps consist of decoding repetition codes, we have $50\%$ threshold error rate on the cell sector.

\begin{figure}
    \centering
    \subfloat[]{
        \includegraphics{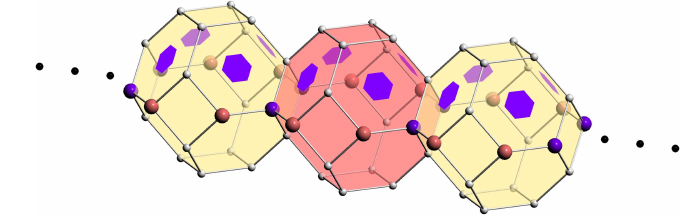}
        \label{fig:3d-color-code-face-symmetries-1a}
    }
    \subfloat[]{
        \includegraphics{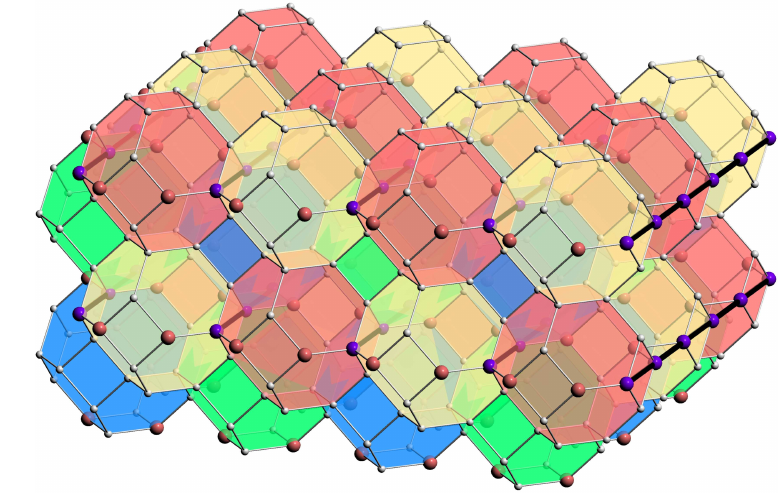}
        \label{fig:3d-color-code-face-symmetries-1b}
    }
    \caption{Step 1 of decoding of face syndromes.
    \textbf{(a)} Linear symmetry of the hexagonal faces: multiplying four hexagonal faces (marked by purple hexagons) on
    any red or yellow cell gives an eight-body check (red and purple vertices).
    Since the red vertices represent Clifford-deformed qubits that can effectively be removed from the stabilizers at infinite $Z$ bias,
    these eight-body checks effectively become four-body (purple vertices).
    The product of these checks along the $y$ axis forms a linear symmetry.
    \textbf{(b)} Successful matching along the linear symmetry in (a) gives rise to two-body checks
    supported on the yellow-red squares of the $xz$ plane.
    Combining them with the checks lying on the green-blue squares of the $xy$ plane,
    which are effectively two-body due to the Clifford deformation, we get a linear symmetry along the $x$ direction,
    represented by the black lines. Matching along these black lines allows us to decode errors on all purple qubits.
    }
    \label{fig:3d-color-code-face-symmetries-1}
\end{figure}

\begin{figure}
    \centering
    \begin{tabular}{cc}
        \subfloat[]{
            \makecell{
                \includegraphics[width=0.38\textwidth]{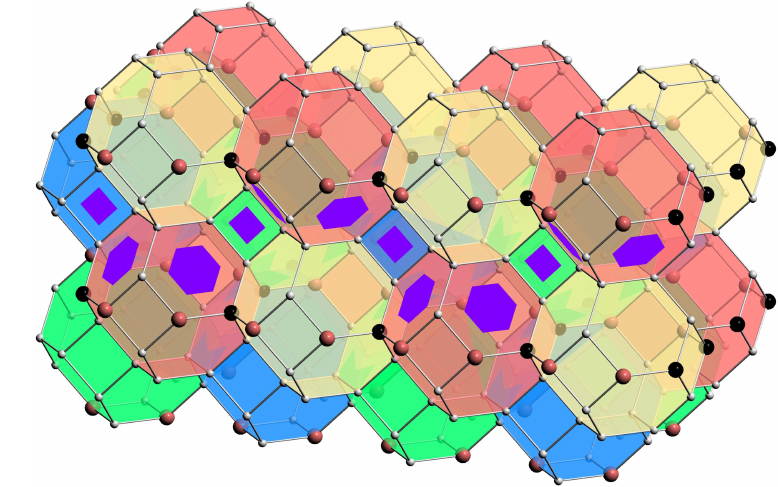}
                \includegraphics[width=0.38\textwidth]{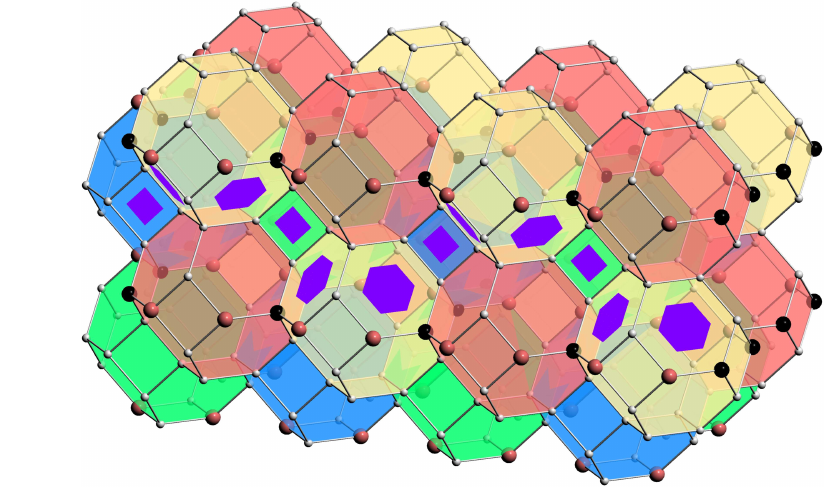}
                \label{fig:3d-color-code-face-symmetries-2a}
            }
        } &
        \subfloat[]{
            \makecell{
                \includegraphics{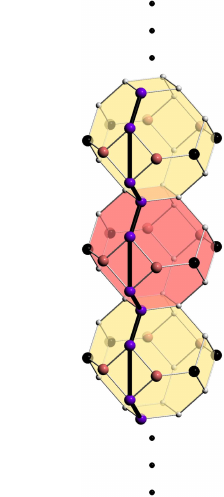}
                \label{fig:3d-color-code-face-symmetries-2b}
            }
        }
    \end{tabular}
    \caption{Step 2 of decoding of face syndromes. Red vertices represent the Clifford-deformed qubits (effectively excluded from the face stabilizers at infinite $Z$ bias),
    while black vertices represent the qubits on which errors have already been decoded in the previous step.
    \textbf{(a)} Linear symmetries in the $y$ direction, involving weight-4 checks sitting on both squares and hexagons.
    \textbf{(b)} Matching along the symmetries represented in (a) results in new
    weight-2 checks, supported either on ends of hexagon-square edges or on ends
    of hexagon-hexagon edges.
    Combining the new checks supported on ends of hexagon-hexagon edges with the checks supported on the red-yellow squares of the $yz$ plane
    (which are also weight-2 due to the Clifford deformation), we get a linear symmetry along the $z$ axis, represented by the black line.
    The remaining new weight-2 checks on ends of hexagon-square edges on the blue-green squares of the $yz$ plane are used in the next step.
    Note also that the yellow-red squares of the $xy$ plane effectively
    become weight-2 checks after this step.
    }
    \label{fig:3d-color-code-face-symmetries-2}
\end{figure}

Decoding of the face syndromes is done in three steps, each of which decodes errors on a different subset of the qubits.
In the first step, illustrated in \cref{fig:3d-color-code-face-symmetries-1}, we notice the existence of a linear symmetry
along the $y$ direction. The main unit of this symmetry is a four-body check constructed by taking the product of four adjacent hexagons
in a red or yellow cell, as shown in
\cref{fig:3d-color-code-face-symmetries-1a}. Matching along this symmetry
results in new weight-2 checks on the yellow-red faces of the $yz$ plane, which
can be combined with effectively weight-2 checks on the blue-green faces of the
$xy$ plane in an alternating fashion
to form a linear symmetry along the $x$ axis, as shown in \cref{fig:3d-color-code-face-symmetries-1b}.
Matching along such symmetries completes the decoding of errors on the purple qubits in \cref{fig:3d-color-code-face-symmetries-1b}.

\begin{figure}
    \centering
    \includegraphics[scale=0.99]{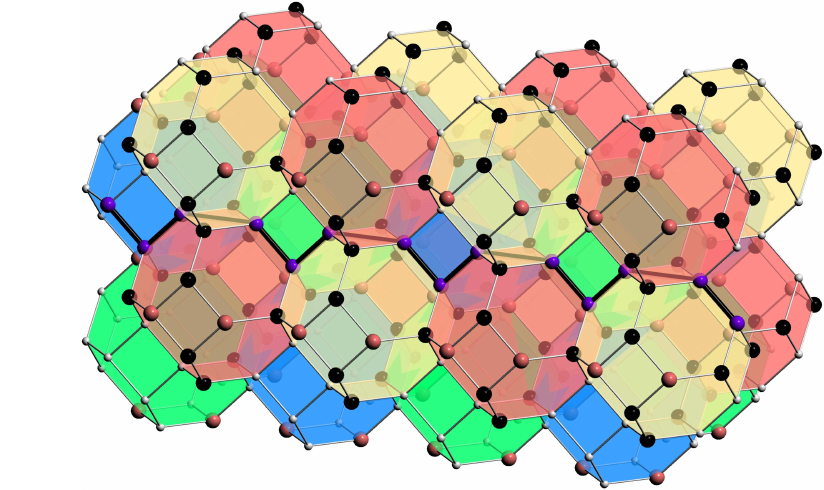}
    \includegraphics[scale=0.99]{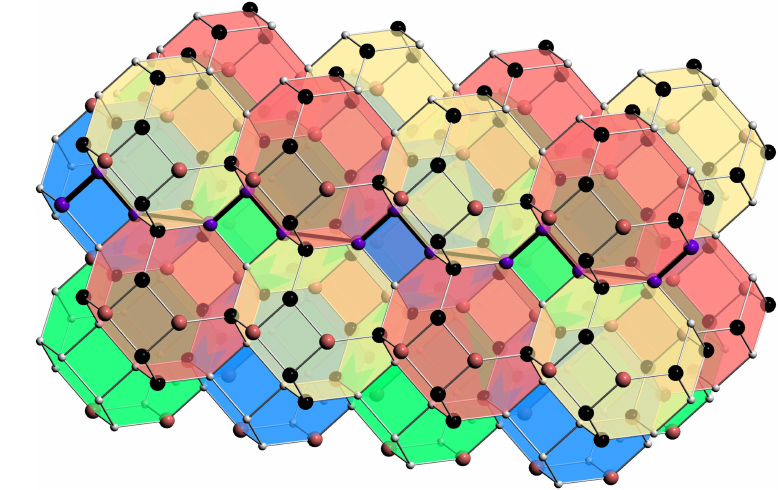}
    \caption{Step 3 of decoding of face syndromes.
    Using both the weight-2 checks supported on the blue-green squares of the $yz$ plane, obtained in the second decoding step,
    and the checks sitting on the yellow-red squares of the $xy$ plane (which are weight-2 due to the second step as well),
    we get new linear symmetries along the $y$ axis, shown with black lines.
    Matching along these black lines allows us to decode errors on all the
    remaining qubits shown in purple.
    }
    \label{fig:3d-color-code-face-symmetries-3}
\end{figure}
In the second step, illustrated in \cref{fig:3d-color-code-face-symmetries-2}, we notice that the hexagons are now effectively weight-4,
and form linear symmetries in the $y$ direction when combined with square faces
in a square-hexagon-hexagon repeating manner.
Matching along these symmetries, as shown in \cref{fig:3d-color-code-face-symmetries-2a}, gives us new weight-2 checks,
supported on the ends of either a hexagon-hexagon intersection edge, or a square-hexagon edge.
Combining the weight-2 checks on hexagon-hexagon edges with the weight-2 checks
supported on the red-yellow squares of the $yz$ plane, we get new linear symmetries in the $z$ direction, as shown in \cref{fig:3d-color-code-face-symmetries-2b}.
Matching along these completes the decoding of errors on the purple qubits of \cref{fig:3d-color-code-face-symmetries-2b}.

In the final step, illustrated in \cref{fig:3d-color-code-face-symmetries-3},
we observe that the remaining weight-2 checks on square-hexagon edges obtained
in the second step combine with effectively weight-2 checks on yellow-red
squares of the $xy$ plane to form new linear symmetries along the $y$ axis.
Matching along these symmetries decodes the errors on the remaining qubits.

Since errors on all the qubits have been decoded by performing matching on a polynomial number of repetition codes, our decoder has a threshold error rate of $50\%$ for the  Clifford-deformed 3D color code.
\end{proof}

\subsection{Fracton codes}
Fracton models offer an interesting set of models to study under biased noise because the models have intrinsically rigid logical operators. 
This means that under multiplication by stabilizer generators, the logical operators do not deform in a topological sense. 
For instance, under stabilizer multiplication, a rigid string-like logical operator may not deform into a string-like logical operator of the same width.

By choosing an appropriate Clifford deformation of the stabilizers, fracton models can have materialized subsystem symmetries 
with respect to infinite-bias noise. 
The combination of the intrinsic conservation laws in addition to the conservation laws associated with the materialized subsystem symmetries 
with respect to the noise, can lead to decoders with high threshold error rates.
Below, we discuss a few canonical examples of fracton models, the X-cube model (type-I fracton model), 
the Sierpiński fracton model (fractal type-I) and the Haah code (type-II) along with their Clifford-deformed codes~\footnote{Type-I fracton
models have string logical operators while type-II do not. Fractal type-I fracton models have fractal-shaped rigid logical operators.}.

\begin{figure}[t]
    \centering
    \subfloat[]{
        \includegraphics[scale=0.8]{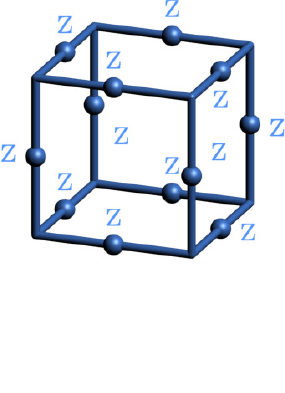}
        \includegraphics[scale=0.8]{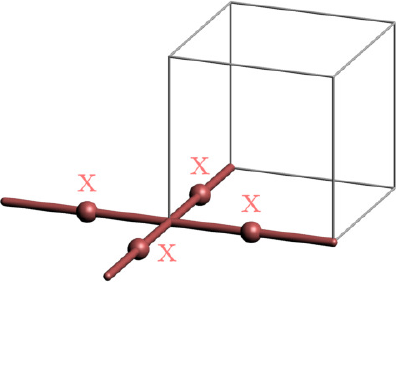}
        \includegraphics[scale=0.8]{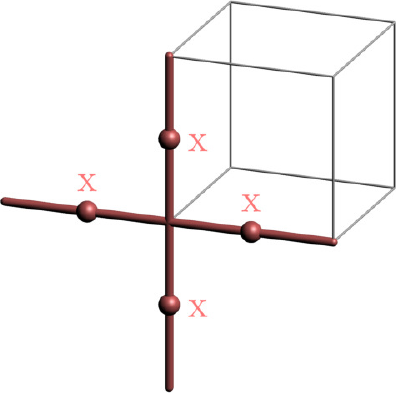}
        \includegraphics[scale=0.8]{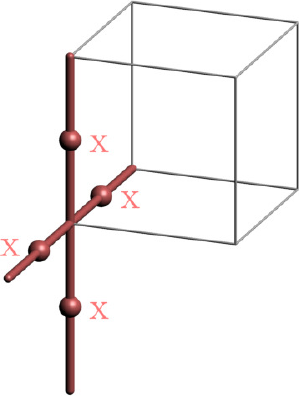}
        \label{fig:xcube-stabilizers-original}
    }
    \hfill
    \subfloat[]{
        \includegraphics[scale=0.8]{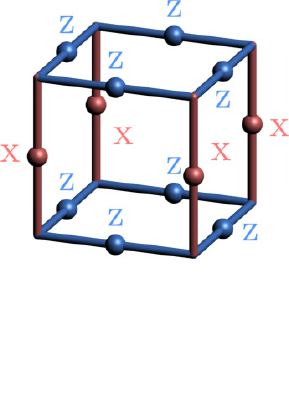}
        \includegraphics[scale=0.8]{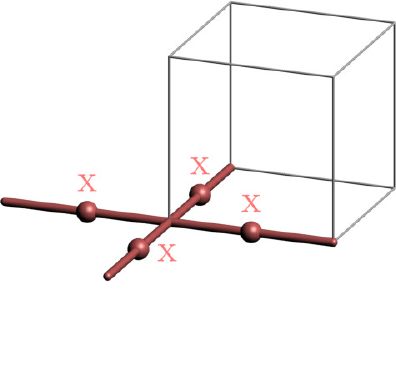}
        \includegraphics[scale=0.8]{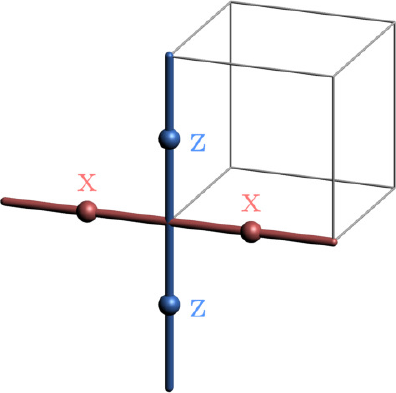}
        \includegraphics[scale=0.8]{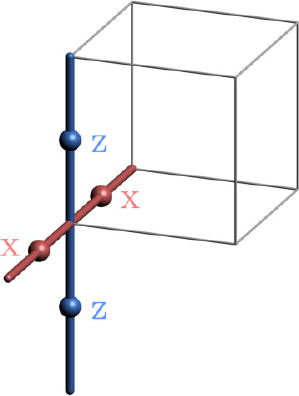}
        \label{fig:xcube-stabilizers-def}
    }
    \caption{X-cube model. \textbf{(a)} Original stabilizers. \textbf{(b)}
    Clifford-deformed stabilizers} %
    \label{fig:xcube-stabilizers}
\end{figure}
\begin{figure}[t]
    \centering
    \subfloat[]{
        \includegraphics[scale=1.2]{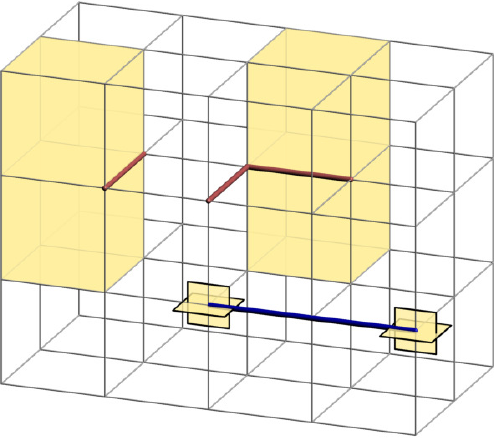}
        \label{fig:xcube-errors}
    }
    \subfloat[]{
        \includegraphics[scale=1.2]{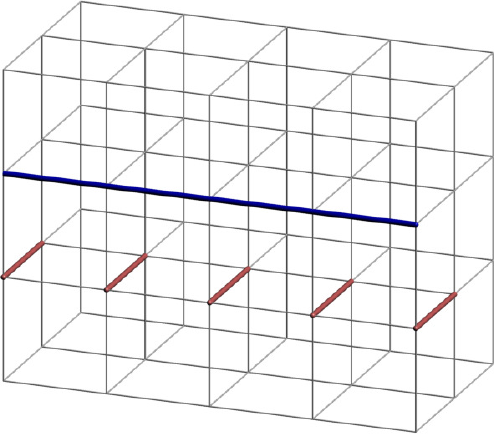}
        \label{fig:xcube-logicals}
    }
    \caption{
        \textbf{(a)} Errors in the X-cube model create two types of syndromes:
        planons (at cells) and lineons (at vertices with orientation).
        \textbf{(b)} Example of $X$ and $Z$ logical operators of the X-cube model in red
        and blue respectively.
    } %
    \label{fig:xcube-errors-logicals}
\end{figure}

\begin{figure}
    \centering
    \subfloat[]{
        \includegraphics{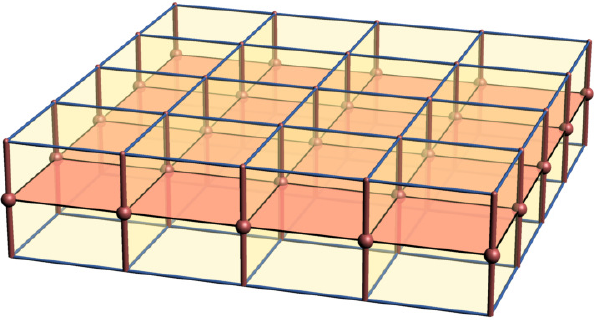}
        \label{fig:xcube-symmetries-a}
    }
    \subfloat[]{
        \includegraphics{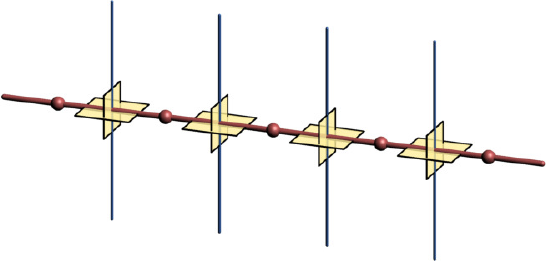}
        \includegraphics{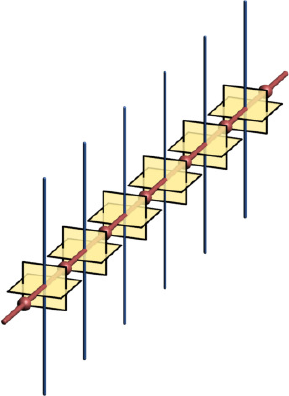}
        \label{fig:xcube-symmetries-b}
    }
    \caption{Symmetries of the Clifford-deformed X-cube model at infinite $Z$ bias.
    \textbf{(a)} Cube stabilizers reduce to four-body checks (red squares) between vertical qubits,
    forming an independent infinite-bias XY code on each layer. They therefore inherit the linear symmetries
    and the $50\%$-threshold error rate of the XY code (see \cref{sec:xy-code}).
    \textbf{(b)} Vertex stabilizers in the $yz$ and $xz$ planes effectively become two-body checks (red edges) between qubits on horizontal edges that form linear symmetries along the $x$ and $y$ axes respectively.
    }
    \label{fig:xcube-symmetries}
\end{figure}

\subsubsection{X-cube model}
\label{sec:xcube-deformation}

The X-cube fracton model is the canonical example of a (foliated) type-I fracton topological order which is defined by the 
presence of topological excitations with restricted mobility. 
It is characterized by a sub-extensive ground space degeneracy and rigid string logical operators. 
A foliated topological stabilizer model is defined by a foliation structure~\cite{shirley2017fracton} 
which implies that the model can be grown by stacking with a 2D topological state and applying a local unitary. 
The X-cube model is 3-foliated, which implies that stacks of surface codes can be extracted under a local unitary along all three lattice directions.

The X-cube fracton model \cite{vijay2016fracton} is defined on a
cubic lattice with qubits on edges. The stabilizer generators come in two types: the \emph{cube stabilizers}, defined on each cubic cell of the
lattice as the product of $Z$ operators over the twelve edges of the cube,
$A_c = \prod_{e \in c} Z_e$, and the \emph{vertex stabilizers}, defined for each
vertex $v$ and orientation $u \in \{ \hat{x}, \hat{y}, \hat{z} \}$ as the
product of the four $X$ operators adjacent to $v$ and orthogonal to $u$,
$B_{v,u} = \prod_{e \in v: e\perp u} X_e$. See \cref{fig:xcube-stabilizers-original}. 
Considering the X-cube model on an $L_x \times L_y \times L_z$ cuboid with periodic boundary conditions, 
the logical operator basis has independent rigid logical string operators that cannot be deformed into each other,
i.e.\ are inequivalent under stabilizer multiplication.
This leads to a macroscopic number of independent logical operator pairs and a linear growth of the number of
encoded qubits.
These logical operators can be expressed as $\bar{X}^{\hat{i}}_{\hat{k},\ell},\bar{Z}^{\hat{j}}_{\hat{k},\ell}$ on pairs of non-contractible loops, where
$\hat{i}\neq \hat{j} \neq \hat{k}$ run over $\{\hat{x},\hat{y},\hat{z}\}$ and $\ell = 0,\dots,L_k-1$. They are defined as
\begin{align}
    \bar{X}^{\hat{x}}_{\hat{z},\ell} = \prod_x X_{x,0,\ell,\hat{z}} \, ,
&&
\bar{Z}^{\hat{y}}_{\hat{z},\ell} = \prod_y Z_{0,y,\ell,\hat{z}}\, ,
\end{align}
and in a similar fashion for other permutations of $x,y,z$,
where $X_{x,y,z,\hat{k}}$ (resp. $Z_{x,y,z,\hat{k}}$) denotes
a Pauli $X$ (resp.  $Z$) operator on the edge adjacent to
the vertex at coordinates $(x,y,z)$ pointing in the $+\hat{k}$ direction for
$\hat{k}\in\left\{ \hat{x},\hat{y},\hat{z} \right\}$.
These string operators are not independent due to the three relations given by $\prod_{\ell} \bar{X}^{\hat{i}}_{\hat{k},\ell} = \prod_{\ell} \bar{X}^{\hat{k}}_{\hat{i},\ell}$ and $\bar{Z}^{\hat{i}}_{\hat{j},0} = \bar{Z}^{\hat{i}}_{\hat{k},0}$. 
Thus, there are overall $2(L_x+L_y+L_z)-3$ logical operator pairs. 
These string operators are rigid in nature as is characteristic of type-I models. 
The rigidity of the string operators directly corresponds to the restricted mobility of excitations. 
For example, particles that are pair-created by a completely rigid undeformable string operator are restricted to move in one dimension 
and are therefore \emph{lineons}. 
Truncations of logical string operators of $X$ errors on a lattice plane create syndromes of cube stabilizers at their end points, 
as shown in \cref{fig:xcube-errors}.
These syndromes cannot freely move to another plane (under arbitrary noise) without creating other syndromes.
Hence, the cube syndromes are referred to as \emph{planons}.
Similarly, the vertex syndromes are created at the ends of rigid strings of $Z$ errors.
Note that two of the vertex stabilizer generators are violated at each end of the string.
This composite syndrome at each end is referred to as the lineon since it cannot move (under arbitrary noise) to another line away from the rigid string,
without creating more syndromes.

We consider a Clifford deformation of the X-cube model where a Hadamard is applied on all vertical edges, similar to the Clifford deformation of the 3D surface code on a cubic lattice.
The Clifford-deformed stabilizer generators are represented in \cref{fig:xcube-stabilizers-def}. 
At infinite $Z$ bias, lineons on the Clifford-deformed X-cube model can only be created by $Z$ errors on $z$ edges, while planons can only be created by $Z$ errors on $x$ and $y$ edges. As a result, we have the following materialized symmetries: the product of vertical planons along a vertical line is effectively the identity, and the product of cubes along a horizontal line is effectively the identity. These symmetries are represented in \cref{fig:xcube-symmetries}. Using the conservation laws associated with these symmetries, we prove that Clifford-deformed X-cube model has a threshold error rate of $50\%$ at infinite $Z$ bias. Note that, using statistical-mechanical simulations, the optimal infinite bias thresholds for the CSS X-cube model with the cube (fracton) term made of Pauli $X$ operators, have been found to be 15.2\% and 7.5\% at infinite $X$ and $Z$ biases respectively~\cite{song_optimal_xcube}.

\begin{theorem}
The Clifford-deformed X-cube model has a $50\%$ threshold error rate under pure $Z$ noise.
\end{theorem}

\begin{proof}

Let us first consider the cube stabilizer generators. As illustrated in \cref{fig:xcube-symmetries-a}, at infinite bias, these stabilizer generators are now effectively
supported on four vertical qubits and form independent sheets of infinite-bias XY surface codes on each layer.
We showed in \cref{sec:xy-code} that this code has linear symmetries on all its rows and columns,
and by using the weight-reduction technique, we proved that it has a threshold error rate of $50\%$. 
Since decoding the cube stabilizers at infinite bias is equivalent to decoding $L$ different XY codes, 
we can infer that the cube sector has a threshold error rate of $50\%$.

Let us now consider the lineon sector. As illustrated in \cref{fig:xcube-symmetries-b}, at infinite bias, vertex stabilizers in the $yz$ and $xz$ planes effectively become two-body checks between qubits on horizontal edges that form linear symmetries along the $x$ and $y$ axes respectively.
The problem of decoding the lineon sector therefore becomes equivalent to decoding $O(L^2)$ repetition codes,
which also has a $50\%$ threshold error rate.
Therefore, the overall code has a threshold error rate of $50\%$.
\end{proof}


\subsubsection{Sierpiński fractal model}
The Sierpiński fractal model, due to Castelnovo, Chamon and Yoshida~\cite{doi:10.1080/14786435.2011.609152,yoshida2013exotic}, is the simplest example of fractal type-I topological order.
Fractcal type-I topological order is defined as type-I fracton topological order that is characterized by the presence of fractal-shaped logical operators and hence does not admit a foliation structure. The model is defined on a cubic lattice, where each vertex has two qubits.
The stabilizer generators are shown in \cref{fig:Sierpiński-stabilizers-original}. This model supports rigid string operators (corresponding to one-dimensional particles or lineons) in the $\hat{z}$ direction and a Sierpinski triangle fractal operator that moves topological excitations apart in 2D. 
Hence this model provides an example with no planons which is consistent with it not having a foliation structure~\cite{Dua_2019,Dua_2020}. 
\begin{figure}
    \centering
    \subfloat[]{
        \includegraphics{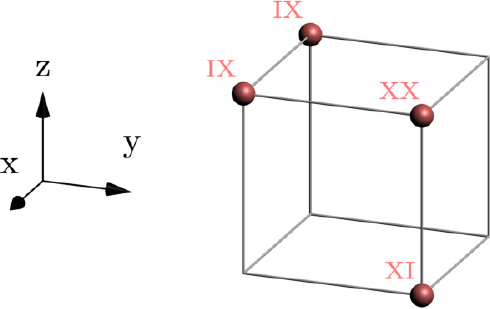}
        \includegraphics{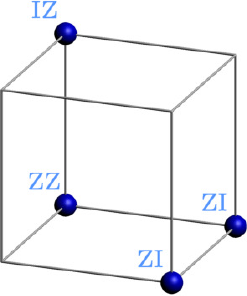}
        \label{fig:Sierpiński-stabilizers-original}
    }
    \hfill
    \subfloat[]{
        \includegraphics{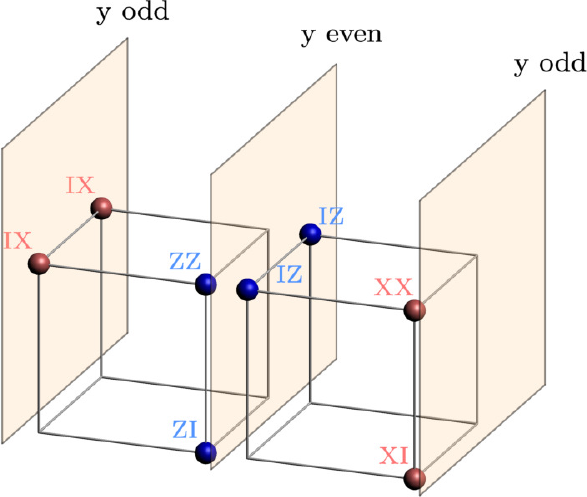}
        \label{fig:Sierpiński-stabilizers-def-1}
    }
    \subfloat[]{
        \includegraphics{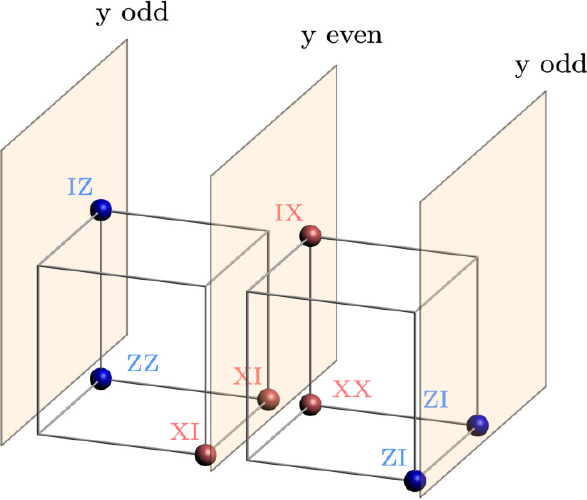}
        \label{fig:Sierpiński-stabilizers-def-2}
    }
    \caption{Sierpiński code, original and Clifford-deformed.
    \textbf{(a)} Original CSS stabilizers, defined on every cube of a cubic lattice with two qubits per vertex.
    There are two types: $X$ stabilizers (left) and $Z$ stabilizers (right).
    \textbf{(b)} Clifford-deformed $X$ stabilizers. 
    The Clifford deformation consists of applying a Hadamard on all qubits of
    vertices on $xz$ planes of even $y$.
    Therefore, stabilizers alternate between the left and the right versions depending on the parity of $y$
    \textbf{(c)} Clifford-deformed $Z$ stabilizers, using the Clifford deformation described in (b).
    }
    \label{fig:Sierpiński-stabilizers}
\end{figure}

\begin{figure}
    \centering
    \subfloat[]{
        \includegraphics{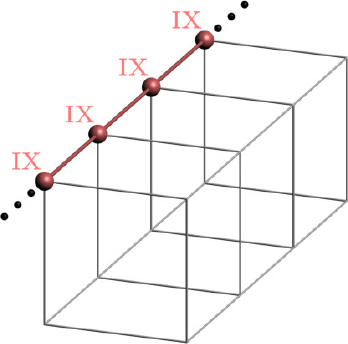}
        \label{fig:Sierpiński-symmetries-a}
    }\quad
    \subfloat[]{
        \includegraphics{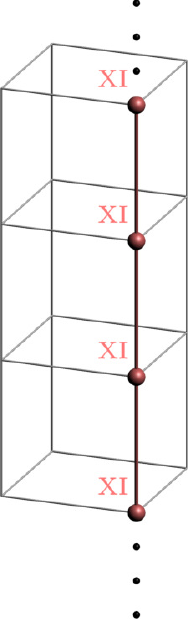}
        \label{fig:Sierpiński-symmetries-b}
    }\quad
    \subfloat[]{
        \includegraphics{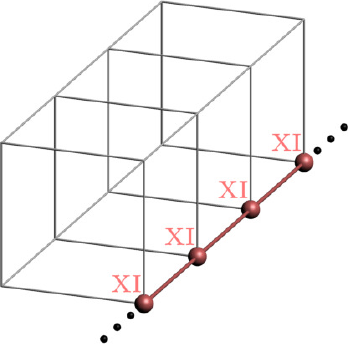}
        \label{fig:Sierpiński-symmetries-c}
    }\quad
    \subfloat[]{
        \includegraphics{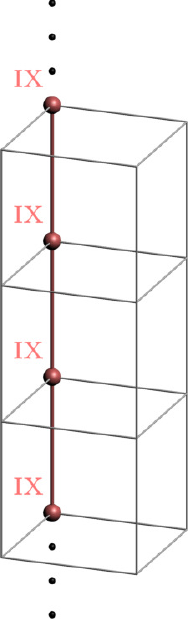}
        \label{fig:Sierpiński-symmetries-d}
    }
    \caption{Effective symmetries of the Clifford-deformed Sierpiński code under
    infinite bias.
    \textbf{(a)} Linear symmetry of Clifford-deformed $X$ stabilizers
    which allows us to decode errors on the second qubits of vertices on planes of odd $y$.
    \textbf{(b)} New linear symmetry of Clifford-deformed $X$ stabilizers obtained after decoding the errors on qubits in (a).
    Since errors on the second qubits of the red vertices have all been decoded,
    the $XX$ terms become $XI$ terms,
    allowing us to decode the errors on the first qubits of vertices on
    planes of odd $y$.
    \textbf{(c)} Linear symmetry of Clifford-deformed $Z$ stabilizers that
    allows us to decode errors on the first qubits of vertices on planes of even
    $y$.
    \textbf{(d)} New linear symmetry of Clifford-deformed $Z$ stabilizers, obtained after decoding the errors on qubits in (c)
    using the same argument as in (b). 
    This allows us to decode errors on the second qubits of vertices on planes of even $y$.
    }
    \label{fig:Sierpiński-symmetries}
\end{figure}

We present the Clifford deformation of this model
where a Hadamard is applied to all qubits of alternating planes,
such as on all $xz$ planes with an \emph{even} $y$ coordinate
as shown in
\cref{fig:Sierpiński-stabilizers-def-1,fig:Sierpiński-stabilizers-def-2}. The model has materialized symmetries which lead to a threshold error rate of $50\%$ at infinite bias as stated below. Note that for the original CSS model, one can use the relation from Ref.~\cite{Takeda_2005} involving the entropy function $h(p)$ for infinite bias threshold error rates $p_X$ ($p_Z$) at infinite $X$ bias ($Z$ bias) respectively as follows, 
\begin{align}
    h(p_X)+h(p_Z)\approx 1,
\end{align} 
where $h(p)=-p\log_2 (p)-(1-p)\log_2(1-p)$ is the binary entropy. Due to the invariance of the model stabilizers under $X\leftrightarrow Z$ permutation, permutation of two qubits on the sites, and inversion, we have 
$h(p_X)=h(p_Z)$. Together we get $h(p_Z)\approx 1/2$ which yields an optimal threshold estimate of $p_z\approx 0.11$ at infinite bias. 

\begin{theorem}
    The Clifford-deformed Sierpiński code has a threshold error rate of $50\%$ under pure $Z$ noise.
\end{theorem}

\begin{proof}
    We first study the Clifford-deformed $X$ stabilizers in
    \cref{fig:Sierpiński-stabilizers-def-1}. 
    At infinite $Z$ bias, these effectively become two-vertex checks supported on planes of odd $y$, oriented either along the $x$ direction or the $z$ direction.
    The checks oriented along the $x$ direction have a term $IX$ on each vertex and form a linear symmetry,
    as shown in \cref{fig:Sierpiński-symmetries-a}.
    Matching along it allows us to decode the errors on the second qubits of the
    vertices living on planes of odd $y$. 
    Once the errors on these qubits have been decoded, we can use them to simplify the checks oriented in the $z$ direction. 
    More precisely, all the terms sitting on the second qubit of a vertex can now be removed from the check,
    turning $XX$ into $XI$. Those updated checks form a new linear symmetry, shown in \cref{fig:Sierpiński-symmetries-b}.
    Matching along this symmetry allows us to decode the first qubit of every vertex living on planes of odd $y$.
    The proof for the Clifford-deformed $Z$ stabilizers follows a similar pattern and is illustrated in
    \cref{fig:Sierpiński-symmetries-c,fig:Sierpiński-symmetries-d}. Overall, this decoding strategy is equivalent to decoding a polynomial number of repetition codes (in the lattice size $L$),
    showing that it has a threshold error rate of $50\%$.
\end{proof}

\subsubsection{The Haah code}

\begin{figure}
    \centering
    \subfloat[]{
        \includegraphics{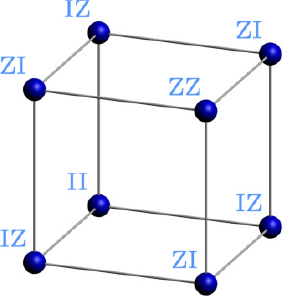}
        \includegraphics{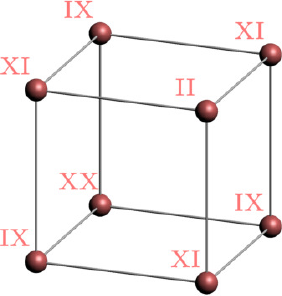}
        \label{fig:haah-stabilizers-original}
    }
    \hfill
    \subfloat[]{
        \includegraphics{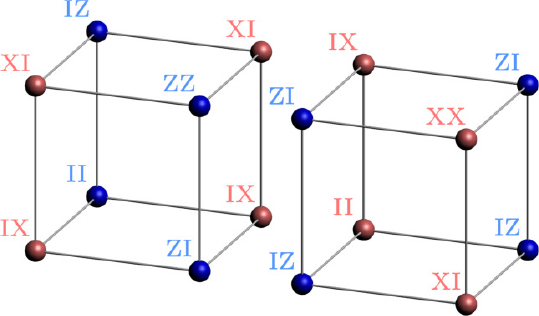}
        \label{fig:haah-stabilizers-def-1}
    }
    \subfloat[]{
        \includegraphics{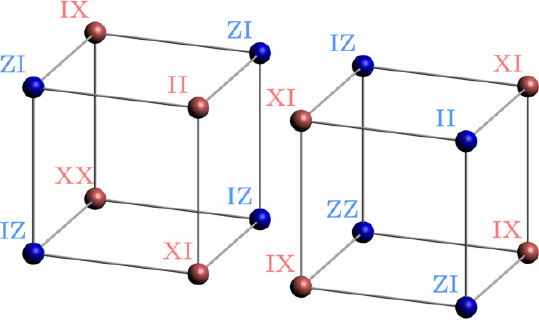}
        \label{fig:haah-stabilizers-def-2}
    }
    \caption{The Haah Code.
    \textbf{(a)} Original stabilizers. The code is defined on a cubic lattice, with two qubits per vertex. 
    Each cube of the lattice contains both an $X$ and a $Z$ stabilizer.
    \textbf{(b)} Clifford-deformed $Z$ stabilizer. A Hadamard is applied on the two qubits of half of the vertices,
    in a checkerboard manner on each layer. As a result, half of the cells contain the stabilizer on the left,
    and half of them contain the stabilizer on the right.
    \textbf{(c)} Clifford-deformed $X$ stabilizer, when applying the transformation described in (b).
    }
    \label{fig:haah-stabilizers}
\end{figure}

The Haah code is the canonical example of type-II fracton topological order, which is characterized by the absence of string logical operators, presence of fractal logical operators, and a sub-extensive ground space degeneracy that can fluctuate with the system size. The stabilizer generators of the Haah code (CSS model) are shown in \cref{fig:haah-stabilizers-original}. 

We present the Clifford-deformed Haah code in
\cref{fig:haah-stabilizers-def-1,fig:haah-stabilizers-def-2} which we prove below to have a 50\% threshold error rate at infinite bias. Note that, similar to the CSS Sierpiński model, for the CSS Haah code, one can also use the relation from Ref.~\cite{Takeda_2005} involving the entropy function $h(p)$ for infinite bias threshold error rates $p_X$ ($p_Z$) at infinite $X$ bias ($Z$ bias) respectively as follows, 
\begin{align}
    h(p_X)+h(p_Z)\approx 1,
\end{align} 
where $h(p)=-p\log (p)$. And again, due to the invariance of the model stabilizers under $X\leftrightarrow Z$ permutation, permutation of two qubits on the sites and inversion, we have 
$h(p_X)=h(p_Z)$. Together we get $h(p_Z)\approx1/2$ which yields an optimal threshold estimate of $p_Z\approx0.11$ at infinite $Z$ bias. 

We now state the theorem about the threshold and its proof. 

\begin{theorem}
    The Clifford-deformed Haah code on a periodic lattice with dimensions $(L_x, L_y, L_z)$, such that $L_z=2^k$, $L_x=L_y$, 
    and $\gcd(L_x,L_z)=2$, has a threshold error rate of $50\%$ under pure $Z$ noise.
\end{theorem}

Such constraints ensure that the horizontal dimensions are even, which is required for our checkerboard-like Clifford deformation
to be well-defined, and that $\gcd(L_x, L_z)$ does not grow with the lattice size, which is needed for the linear symmetries considered in our proof. 
Because the number of encoded qubits in a fractal type-II fracton model like the Haah code can fluctuate wildly with system size, 
it is subtle to extract the threshold error rate from an arbitrary family of increasing system sizes. 
Nevertheless, there are families of codes
satisfying the above constraints whose number of encoded qubits is constant.
We checked this numerically for codes of size
$(2^k+2, 2^k+2, 2^k)$ for $k\geq 2$,
which have 6 encoded qubits,
and $(2\cdot 3^k, 2\cdot 3^k, 2^k)$ for $k\geq 1$,
which have 6 encoded qubits.

\begin{proof}

Let us call the Clifford-deformed qubits \emph{type-A qubits}, and the vertices they live on \emph{type-A vertices}.
We refer to the other half of the qubits (resp. vertices) as \emph{type-B qubits} (resp. \emph{type-B vertices}).
Those two types of qubits alternate in a 2D checkerboard manner on each horizontal layer of the lattice.
In this language, the Clifford deformation of the $Z$-type stabilizers gives stabilizers supported on type-A qubits only.
We call these stabilizers \emph{type-A stabilizers}. Similarly, \emph{type-B stabilizers} are the ones resulting from 
the Clifford deformation of the $X$-type stabilizers, and are supported on type-B qubits only.

\begin{figure}
    \centering
    \subfloat[]{
        \includegraphics{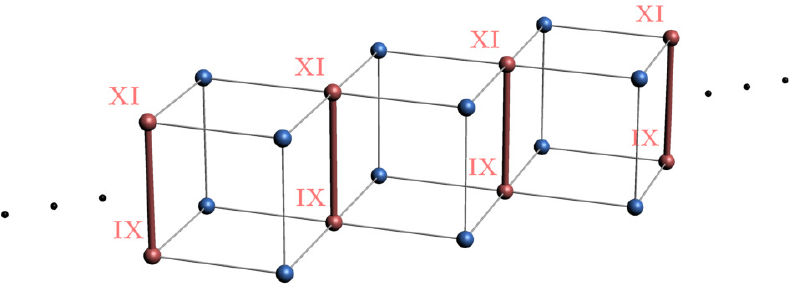}
        \label{fig:haah-linear-symmetry}
    }
    \hfill
    \subfloat[]{
        \includegraphics{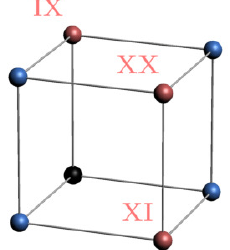}
        \includegraphics{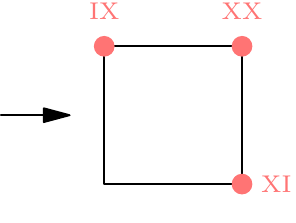}
        \label{fig:haah-projection}
    }
    \subfloat[]{
        \includegraphics{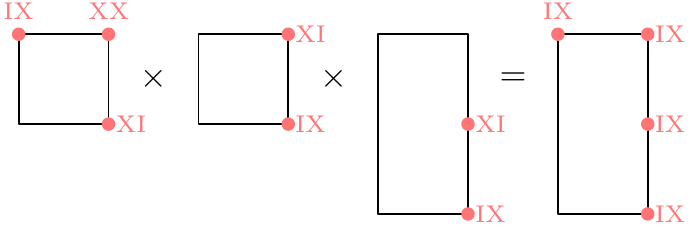}
        \label{fig:haah-new-stabilizer}
    }
    \label{fig:haah-proof-1}
    \caption{Deformed Haah code $50\%$ threshold error rate proof, step 1
    \textbf{(a)} Linear symmetry. Matching along it allows a weight reduction to weight-2 checks (red edges)
    \textbf{(b)} Visualization of the second stabilizer type on a 2D plane
    \textbf{(c)} New L-shaped pure ``IX'' stabilizer, obtained by multiplying the stabilizer in (b) by the weight-2 checks
    obtained in (a). We call this an L-check.
    }
\end{figure}

Let us start by considering only type-A stabilizers, which are shown in \cref{fig:haah-stabilizers-def-1}. 
The decoding strategy is the same when tackling type-B stabilizers.
In the infinite-bias regime, the Clifford-deformed code becomes a classical code, with two types of parity-check operators 
alternating in a checkerboard manner. 
Those two types of checks have weight 4, but one is supported on four vertices and the other on three.
We can observe the presence of a linear symmetry for the ones supported on four vertices, as shown in \cref{fig:haah-linear-symmetry}.
Matching along this symmetry results in the appearance of new weight-2 checks, with the term ``XI'' on one vertex and ``IX'' on the other vertex.
We call them ``XI-IX'' checks.

We then consider the other type of check supported on three vertices. 
Multiplying it by two ``XI-IX'' checks, we get a new L-shaped weight-4 check made only of ``IX'' terms, as represented in \cref{fig:haah-new-stabilizer}.
We call them ``L-checks''.

We now use a technique introduced to study the classical Fibonacci codes \cite{nixon2021correcting} and the XYZ color code \cite{miguel2022cellular}.
We first notice that applying four L-checks as an L results in a new L-check where the spacing between the non-zero qubits has doubled
but the weight is still 4.
Applying the same process recursively results in a fractal of original checks, forming an L-check whose size can be an arbitrary power of two.
This process is shown in \cref{fig:haah-fractal-evolve}.
Note that this family of L-checks always lives on a 2D diagonal slice of our lattice, which has dimensions $(L_x,L_z)$.
This is due to the choice of periodic boundary conditions and equal horizontal dimensions. 

We now use the fact that $L_z=2^k$. 
Applying the fractal process described previously, we can create an L-check where two terms are separated by a distance of exactly $2^k$.
Due to the periodic boundary conditions, these two terms cancel out, leaving only two qubits in the support of the check.
This new weight-2 check is shown in \cref{fig:haah-fractal-cancel-out}.

As $\gcd(L_x, L_z)=2$, we can multiply these weight-2 checks on a line to cover exactly $L_x$ qubits,
before the line comes back to itself, as shown in \cref{fig:haah-wrapping-around}. \begin{figure}
    \centering
    \subfloat[]{
        \includegraphics[scale=1.0]{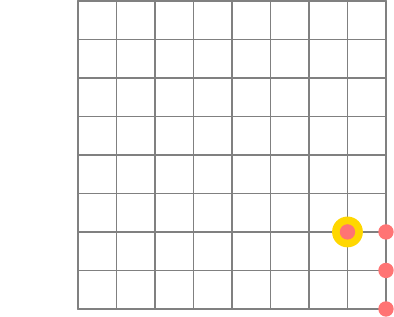}
        \includegraphics[scale=1.0]{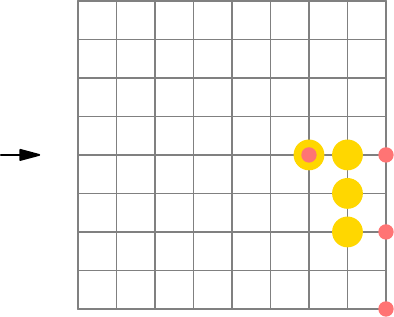}
        \includegraphics[scale=1.0]{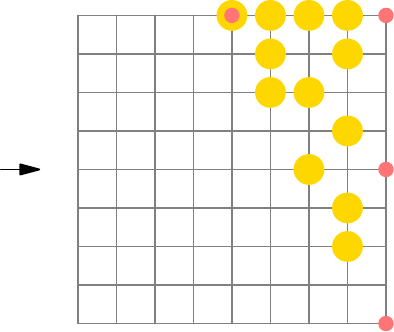}
        \label{fig:haah-fractal-evolve}
    }
    \hfill
    \subfloat[]{
        \includegraphics[scale=1.0]{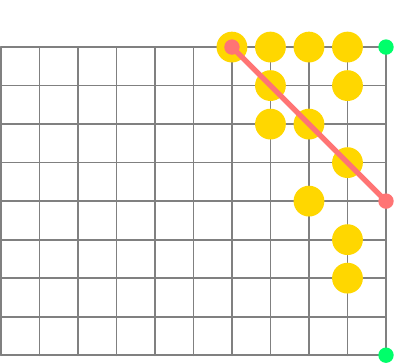}
        \label{fig:haah-fractal-cancel-out}
    }
    \subfloat[]{
        \includegraphics[scale=1.0]{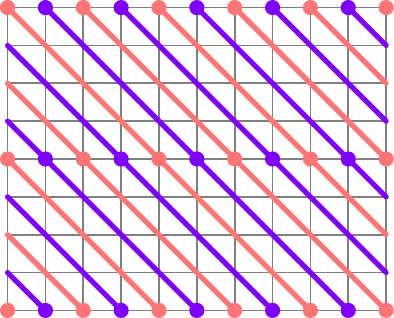}
        \label{fig:haah-wrapping-around}
    }
    \caption{Step 2 in the decoding of Clifford-deformed Haah code at infinite bias. Yellow dots on vertices represent the application of an L-check, 
    as shown on the left of (a), where red dots represent the qubits in its support.
    \textbf{(a)} By multiplying the checks in a fractal manner, we can obtain a similar weight-4 check with
    any power of two spacing between the qubits in its support.
    \textbf{(b)} Taking a periodic lattice with vertical dimension $L_z=2^k$, we can make two of the qubits cancel
    out (green) to obtain a new weight-2 check (red)
    \textbf{(c)} On any periodic lattice with dimensions $(L_x,L_y)$ such that $\gcd(L_x,L_y)=2$, diagonal lines wrapping around the torus
    cover half of the vertices (red or purple). Therefore, the product of these new weight-2 checks on a line 
    forms a linear symmetry (either red or purple depending on the starting point).
    }
    \label{fig:haah-proof-2}
\end{figure}

This is due to the fact that on a periodic lattice with $\gcd(L_x,L_y)=g$, 
there are $g$ distinct diagonal lines covering each a fraction $1/g$ of the vertices \cite{Tailoring2019}.
Since the dimensions of our diagonal slice is $(L_x,L_z)$, it contains $L_xL_z$ vertices, and the diagonal line
formed by the weight-2 checks covers $L_x L_z / 2$ vertices. Dividing this by the size of the weight-2 check,
$L_z/2$, we obtain that the product of weight-2 checks on the line covers exactly $L_x$ qubits.

Therefore, this product is equal to the identity operator and we get a linear symmetry.
By translating the large L-check, we can include any arbitrary qubit of the 2D diagonal slice in the support of this linear symmetry.
Matching along them on all the 2D slices therefore results in decoding the errors on second qubits of all the type-A vertices.
Using the ``XI-IX'' check allows us to decode the errors on the first qubits of these vertices as well.
Finally, applying the same decoder to the type-B stabilizer, we can decode errors on all type-B qubits.

Since this decoder only involved matching on a polynomial number of repetition codes, we can deduce that our Clifford-deformed Haah code 
has a $50\%$ threshold error rate at infinite bias.

\end{proof}

Note that the choice of $L_x = L_y$ and $\gcd(L_x,L_z)=2$, while simplifying the proof, are not strictly necessary 
to obtain the desired linear symmetries and the $50\%$ threshold error rate. 
Relaxing these constraints has the effect of changing the size of the 2D diagonal slice, as it can wrap around the torus several times
when the horizontal dimensions are not equal. The new size of the slice can be shown to be $\ell=L_x L_y / \gcd(L_x,L_y)=\lcm(L_x,L_y)$.
As a result, the condition to obtain a linear symmetry on the weight-2 checks is modified: we want the number of terms 
in the linear symmetry to grow polynomially with the system size. As the number of terms is given by the size of the diagonal line which supports
the symmetry, $\lcm(\ell,L_z)=\lcm(L_x,L_y,L_z)$
\footnote{Here we use the fact that $\lcm(\lcm(a,b),c)=\lcm(a,b,c)$.}
, divided by the size of the weight-2 check, $L_z/2$, the condition can be reformulated as
\begin{align}
    \frac{\lcm(L_x,L_y,L_z)}{L_z} = \Omega(\text{poly}(L_z))
    \label{eq:haah-general-condition}
\end{align}
Imposing \cref{eq:haah-general-condition} with $L_z=2^k$, and $L_x$, $L_y$ even (to guarantee that the Clifford deformation is well-defined)
is enough to have a $50\%$ threshold error rate for the Haah code.

\section{Threshold error rates at finite bias}\label{sec:finbiasthresholds}
Using the belief propagation with ordered statistics decoder
(BP-OSD)~\cite{panteleev2019degenerate,roffe2020decoding,quintavalle2021single},
described in \cref{sec:bposd-appendix},
we evaluate the threshold error rates of both the CSS and Clifford-deformed 3D surface code defined on cubic and checkerboard lattices as well as the X-cube model
at different bias ratios over several orders of magnitude and at infinite bias.
We also estimate threshold error rates for the surface code on the cubic lattice using the
sweep-matching decoder described in \cref{sec:sweepmatch-appendix}.
We plot the threshold error rate estimates for different values of bias in
\cref{fig:thresholdvsbias}.
The numerical values of the threshold error rate estimates and uncertainties are
listed in \cref{tab:toric}
for the surface code on a cubic lattice,
in \cref{tab:rhombic} for the surface code on a checkerboard lattice and in
\cref{tab:xcube} for the X-cube model.
The estimates are best-fit parameters to a finite-size scaling ansatz
and uncertainties are bootstrapped $1\sigma$ credible intervals that account for
the finite number of trials and choice of parameters.
An in-depth overview of how these are obtained is in
\cref{sec:numerics-appendix}.
Unless a lower numerical threshold error rate can be resolved by finite-size
scaling the theoretically proved $50\%$ threshold error rate is tabulated and
plotted for Clifford-deformed codes at infinite bias.

We note that at moderate biases, three-dimensional surface
codes such as the \defcubic code and the CSS surface code on a 3D checkerboard lattice, can have threshold error rates close to the hashing bound and to those of 2D codes like \XZZX and \XY.
This offers a noise regime in which one could consider a dimensional jump for
implementation of non-Clifford gates and be able to maintain at least the code
capacity threshold error rates.
For high bias $\eta_Z\gtrsim 100$, the Clifford-deformed surface code on a cubic
lattice beats the CSS surface code on a cubic lattice in threshold error rate performance,
with threshold error rates of 50\% and 21.37(4)\% respectively at infinite bias.
The Clifford-deformed code on the checkerboard lattice at infinite bias also boasts an
advantage at infinite bias.
Above modest values of bias $\eta_Z\gtrsim 30$,
the Clifford-deformed X-cube model outperforms its CSS counterpart.
This is owed to the rigid noise symmetries
which allow decoding in rigid submanifolds.

\paragraph{Limitations of BP-OSD} As discussed in \cref{sec:bposd-appendix}, the performance of BP-OSD
greatly depends on the characteristics of the code, particularly its \emph{girth} (size of the shortest cycle
in the Tanner graph) and its \emph{split-belief number} 
(weight of the smallest error that produces a degenerate syndrome). 
The latter can be calculated by taking the weight of the smallest even-weight stabilizer and dividing it by two.
We show these two numbers for different 3D codes in \cref{tab:girth-and-split-belief}.
\begin{table}[]
    \begin{tabular}{@{}lllll@{}}
    \toprule
    Code                           & $X$-girth & $Z$-girth & $X$-split-belief & $Z$-split-belief    \\                            &  &  & number & number \\ \midrule
    3D surface code (cubic)        & 8       & 8       & 2                     & 3                     \\
    3D surface code (checkerboard) & 8       & 6       & 2                     & 6                     \\
    X-cube model                   & 8       & 4       & 2                     & 6                     \\
    3D color code                  & 4       & 6       & 2                     & 12                    \\ \bottomrule
    \end{tabular}
    \caption{Girth and split-belief numbers for four of the 3D codes studied in this work.
    The girth is the size of the shortest cycle of the Tanner graph,
    while the split-belief number is the weight of the smallest error that causes a degenerate
    syndrome. We consider here the $X$ and $Z$ parts of the Tanner graph separately, as we are
    using independent BP-OSD decoders for each error type. Here, the $X$-girth (resp. $Z$-girth)
    corresponds to the girth when considering only the $X$ (resp. $Z$) stabilizers in the Tanner graph,
    and similarly so for the $X$ and $Z$-split-belief numbers.
    \label{tab:girth-and-split-belief}
    }
\end{table}
\begin{figure}
    \centering
    \includegraphics[width=0.7\textwidth,page=126]{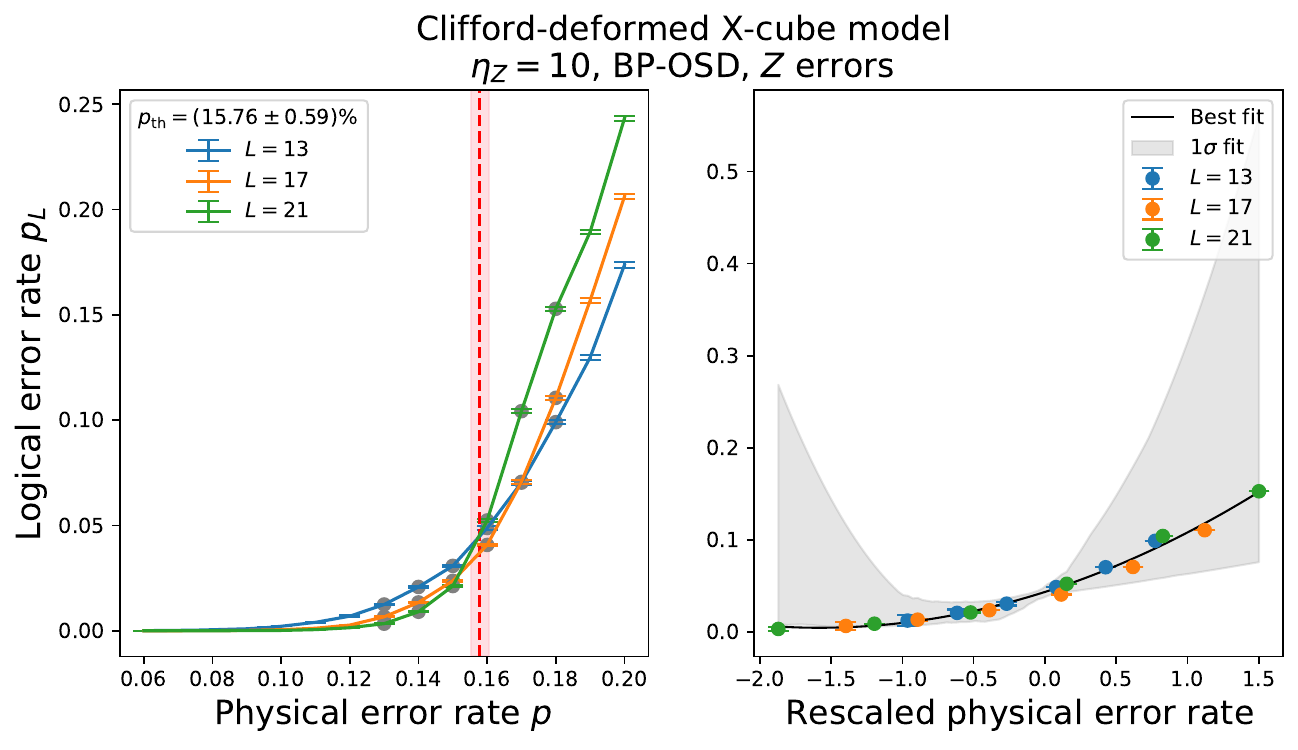}
    \caption{Size-dependent reduction in apparent threshold error rate using a
    BP-OSD decoder for increasing X-cube model lattice sizes under $Z$-biased
    noise with $\eta_Z=10$.
    (left) Rate of logical $Z$ errors vs the physical error rate $p$.
    Note that the apparent intersection of the curves for $L=13$ is somewhat
    different from that of the $L=21$.
    Despite the apparent threshold error rate shift,
    we still provide a threshold error rate estimate by regression with the finite-size
    scaling ansatz.
    (right) Logical $Z$ error rate vs rescaled physical error rate
    $x=\left( p-p_{\textrm{th}} \right)^{1/\nu}$ where 
    the threshold error rate $p_{\textrm{th}}$ and critical exponent $\nu$ have been
    estimated by fitting to a finite-sized scaling ansatz,
    for which the best-fit curve and $1\sigma$ fits are plotted against data
    points colored by code lattice size as detailed in
    \cref{sec:numerics-appendix}.}
    \label{fig:bposdshift}
\end{figure}

A common phenomenon that appears for codes with low girth or low split-belief number is 
the receding threshold error rate problem: the apparent threshold error rate decreases with increasing system sizes \cite{higgottImprovedSingleshotDecoding2022}.
We observed this finite-size effect for the 3D surface code on a checkerboard lattice and for the X-cube model,
while the 3D surface code on a cubic lattice did not have this issue for sizes up to 22.
An illustration of this phenomenon is shown in \cref{fig:bposdshift} for the X-cube model at bias $\eta=10$.
Consequently, on codes where BP-OSD suffers this limitation, there is greater uncertainty on the threshold error rate estimates upon using the
bootstrapped finite-size scaling method described in \cref{sec:numerics-appendix}, as can be seen for instance in \cref{tab:xcube}.

We also perform preliminary experiments on the 3D color code, but we observe a strong receding threshold effect with an apparent threshold error rate orders of magnitude below its optimal value. 
This can be explained by the particularly low girth of the 3D color code, and entices us not to pursue these experiments.

Note that while we have provided decoders with 50\% threshold error rate at infinite bias for all the Clifford-deformed codes studied in this paper,
BP-OSD does not always achieve that threshold. For example in the X-cube model as seen in \cref{fig:thresholdvsbias-c}
and \cref{tab:xcube} where the threshold error rate is only $17.2(3)\%$.
Nevertheless, the Clifford-deformed X-cube model still outperforms the CSS
X-cube model
which has a lower threshold error rate of $9.25(10)\%$ when decoding with
BP-OSD.

\begin{figure}[htbp]
    \centering
    \subfloat[]{
        \label{fig:thresholdvsbias-a}
        \includegraphics[width=0.7\textwidth,page=2]{figures/2022-10-15_thresholds-vs-bias.pdf}
    }
    \hfill
    \subfloat[]{
        \label{fig:thresholdvsbias-b}
        \includegraphics[width=0.7\textwidth,page=1]{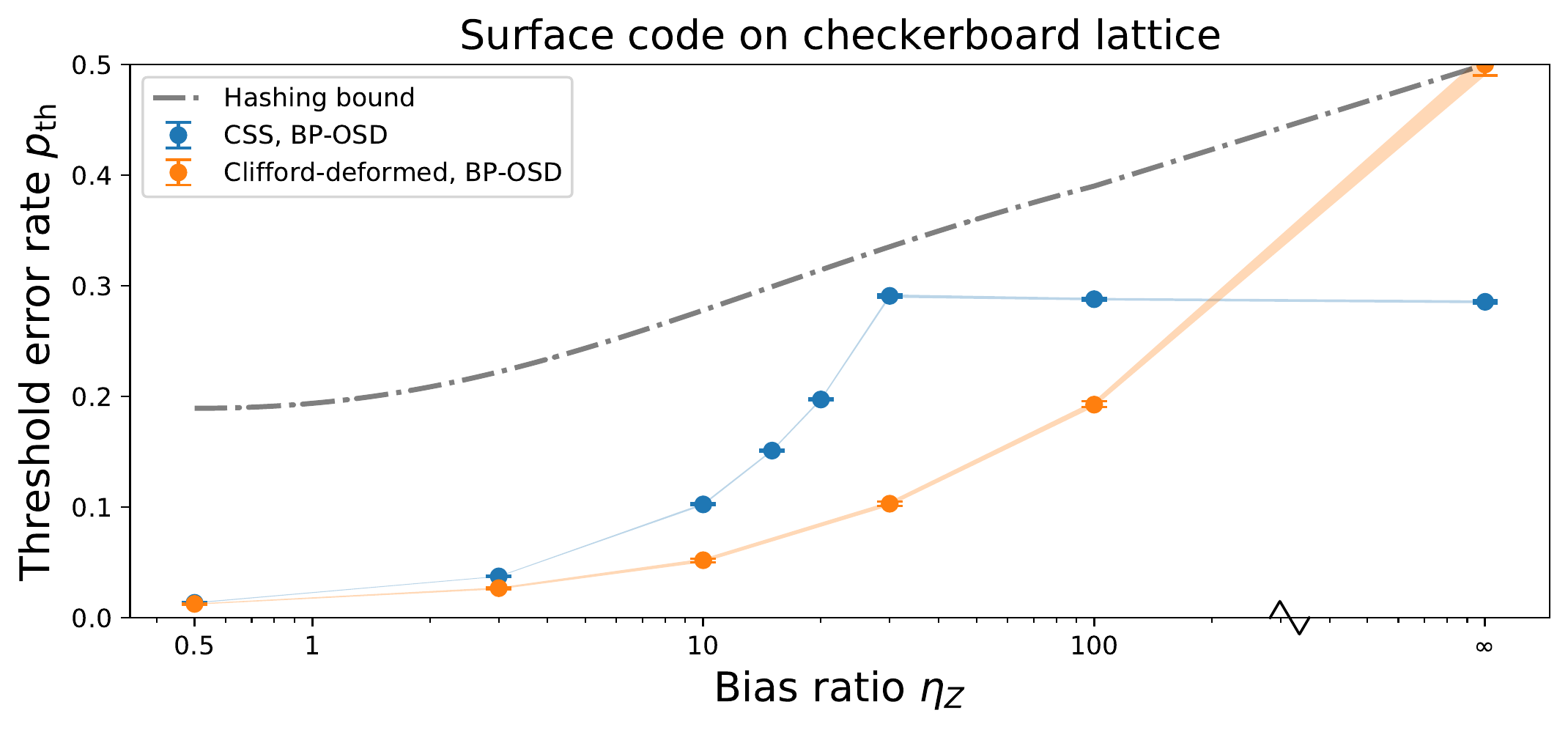}
    }
    \hfill
    \subfloat[]{
        \label{fig:thresholdvsbias-c}
        \includegraphics[width=0.7\textwidth,page=3]{figures/2022-10-15_thresholds-vs-bias.pdf}
    }
    \caption{Threshold error rate $p_{\textrm{th}}$ vs dephasing bias $\eta_Z$
    for some CSS and Clifford-deformed 3D topological codes.
    \textbf{(a)} BP-OSD and sweep-matching threshold error rates of the 3D CSS
    surface code and the Clifford-deformed surface code on a cubic lattice.
    \textbf{(b)} BP-OSD threshold error rates of the CSS and Clifford-deformed
    surface codes on a checkerboard lattice.
    \textbf{(c)} BP-OSD threshold error rates of the CSS and Clifford-deformed
    X-cube model.
    }%
    \label{fig:thresholdvsbias}
\end{figure}

\begin{table}[H]
\centering
\caption{Estimates and uncertainties of the threshold error rate
$p_{\textrm{th}}$ for CSS and Clifford-deformed surface codes on a cubic lattice using
different decoders under finite and infinite bias as plotted in
\cref{fig:thresholdvsbias-a}.}
\label{tab:toric}
\bgroup
\def\arraystretch{1.5}
\begin{tabular}{ccccc}
\toprule
 & \multicolumn{2}{c}{CSS} & \multicolumn{2}{c}{Deformed} \\
Bias $\eta_Z$ & BP-OSD & Sweep-matching & BP-OSD & Sweep-matching \\
\midrule
0.5 & $({5.95}\pm 0.03)\mathrm{\%}$ & ${4.0}^{+0.3}_{-0.5}\mathrm{\%}$ & ${5.99}^{+0.08}_{-0.1}\mathrm{\%}$ & ${4.42}^{+0.12}_{-0.32}\mathrm{\%}$ \\
1 & - & $({5.4}\pm 0.4)\mathrm{\%}$ & - & ${5.06}^{+0.06}_{-0.24}\mathrm{\%}$ \\
3 & ${12.25}^{+0.05}_{-0.06}\mathrm{\%}$ & ${11.48}^{+0.11}_{-0.19}\mathrm{\%}$ & $({7.94}\pm 0.03)\mathrm{\%}$ & $({6.7}\pm 0.6)\mathrm{\%}$ \\
10 & ${22.3}^{+0.04}_{-0.05}\mathrm{\%}$ & ${14.9}^{+0.5}_{-1.0}\mathrm{\%}$ & ${12.12}^{+0.12}_{-0.11}\mathrm{\%}$ & ${11.6}^{+0.3}_{-0.2}\mathrm{\%}$ \\
30 & $({21.7}\pm 0.04)\mathrm{\%}$ & ${14.6}^{+0.3}_{-0.7}\mathrm{\%}$ & ${17.76}^{+0.12}_{-0.1}\mathrm{\%}$ & $({19.3}\pm 0.3)\mathrm{\%}$ \\
100 & ${21.46}^{+0.02}_{-0.04}\mathrm{\%}$ & ${14.3}^{+0.4}_{-0.7}\mathrm{\%}$ & ${23.0}^{+0.6}_{-0.5}\mathrm{\%}$ & ${20.6}^{+0.5}_{-0.7}\mathrm{\%}$ \\
$\infty$ & ${21.37}^{+0.04}_{-0.03}\mathrm{\%}$ & ${14.59}^{+0.11}_{-0.18}\mathrm{\%}$ & $50\mathrm{\%}$ & ${20.7}^{+0.2}_{-0.3}\mathrm{\%}$ \\
\bottomrule
\end{tabular}
\egroup
\end{table}

\begin{table}[H]
\centering
\caption{Estimates and uncertainties of the threshold error rate
$p_{\textrm{th}}$ for CSS and Clifford-deformed surface codes on a checkerboard lattice
using a BP-OSD decoder under finite and infinite bias as plotted in
\cref{fig:thresholdvsbias-b}.}
\label{tab:rhombic}
\bgroup
\def\arraystretch{1.5}
\begin{tabular}{ccc}
\toprule
Bias $\eta_Z$ & CSS & Deformed \\
\midrule
0.5 & $({1.35}\pm 0.04)\mathrm{\%}$ & ${1.25}^{+0.03}_{-0.07}\mathrm{\%}$ \\
3 & $({3.74}\pm 0.03)\mathrm{\%}$ & ${2.67}^{+0.06}_{-0.1}\mathrm{\%}$ \\
10 & $({10.24}\pm 0.09)\mathrm{\%}$ & ${5.2}^{+0.14}_{-0.19}\mathrm{\%}$ \\
15 & $({15.1}\pm 0.08)\mathrm{\%}$ & - \\
20 & ${19.72}^{+0.09}_{-0.1}\mathrm{\%}$ & - \\
30 & ${29.09}^{+0.13}_{-0.16}\mathrm{\%}$ & $({10.3}\pm 0.2)\mathrm{\%}$ \\
100 & ${28.79}^{+0.12}_{-0.11}\mathrm{\%}$ & ${19.3}^{+0.3}_{-0.2}\mathrm{\%}$ \\
$\infty$ & $({28.55}\pm 0.13)\mathrm{\%}$ & $50.0\mathrm{\%}$ \\
\bottomrule
\end{tabular}
\egroup
\end{table}

\begin{table}
\centering
\caption{Estimates and uncertainties of the threshold error rate
$p_{\textrm{th}}$ for the CSS X-cube model and Clifford-deformed X-cube model using a BP-OSD decoder under
finite and infinite bias as plotted in \cref{fig:thresholdvsbias-c}.}
\label{tab:xcube}
\def\arraystretch{1.5}
\begin{tabular}{ccc}
\toprule
Bias $\eta_Z$ & CSS & Deformed \\
\midrule
0.5 & ${4.54}^{+0.05}_{-0.08}\mathrm{\%}$ & ${4.6}^{+0.3}_{-0.7}\mathrm{\%}$ \\
3 & ${10.64}^{+0.06}_{-0.07}\mathrm{\%}$ & ${5.9}^{+0.4}_{-0.8}\mathrm{\%}$ \\
10 & ${9.3}^{+0.2}_{-0.4}\mathrm{\%}$ & ${7.5}^{+0.4}_{-1.0}\mathrm{\%}$ \\
30 & ${8.9}^{+0.3}_{-0.6}\mathrm{\%}$ & $({10.7}\pm 0.4)\mathrm{\%}$ \\
100 & ${8.9}^{+0.2}_{-0.3}\mathrm{\%}$ & ${17.4}^{+0.3}_{-0.2}\mathrm{\%}$ \\
$\infty$ & $({9.25}\pm 0.08)\mathrm{\%}$ & ${17.2}^{+0.17}_{-0.12}\mathrm{\%}$ \\
\bottomrule
\end{tabular}
\end{table}

\section{Rotated layout and subthreshold scaling}\label{sec:rotsts}

\subsection{Rotated layout for the 3D surface code}

We now define a rotated layout for the 3D surface code.
The new lattice is obtained by rotating the coordinates about the vertical $z$
axis by $45^\circ$ such that the horizontal qubits
formerly living on $x$ and $y$ edges now live on vertices on
horizontal $xy$ planes,
while vertical qubits formerly living on $z$ edges now live on vertices floating
between horizontal $xy$ planes.
Nevertheless,
we continue refer to qubits living on $xy$ planes as \emph{horizontal} qubits
and qubits floating in between $xy$ planes as \emph{vertical} qubits.
Former vertex operators become octahedron stabilizers,
former vertical face stabilizers become diamond stabilizers,
and former horizontal face stabilizers become square stabilizers.
Such a rotated layout preserves the distances $(d_X,d_Z)$ for both Pauli $X$ and Pauli $Z$ logical operators,
while using roughly half the number of physical qubits compared to the regular layout.
This new lattice is illustrated in \cref{fig:rot3DSC-open,fig:rot3DSC-periodic}
for open and periodic boundary conditions respectively.

Note that for periodic boundary conditions with one odd horizontal dimension,
such as in \cref{fig:rot3DSC-periodic},
the horizontal planes are periodic in both directions,
but the vertical diamonds are not periodic in the odd-length direction,
resulting in a \emph{seam} across which the checkerboard pattern of the
horizontal planes of octahedron and horizontal square stabilizers are
incompatible,
and vertical qubits are not connected by diamond stabilizers.
To ensure that the octahedron and square stabilizers commute across this seam,
stabilizers touching the seam on only \emph{one} chosen side are modified by
introducing ``defects'' at every horizontal qubit on the seam,
where a Hadamard is applied to modify these stabilizer definitions
with $X\leftrightarrow Z$ at these defects,
as illustrated in \cref{fig:defect-1}.

\begin{figure}
\centering
\subfloat[]{
    \includegraphics[scale=0.85]{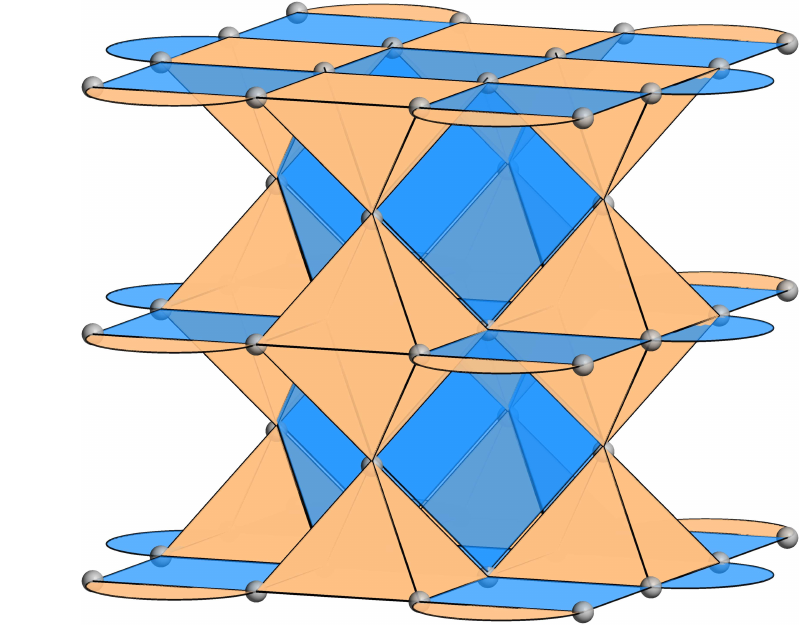}%
    \label{fig:rot3DSC-open}
}
\subfloat[]{
    \includegraphics[scale=0.85]{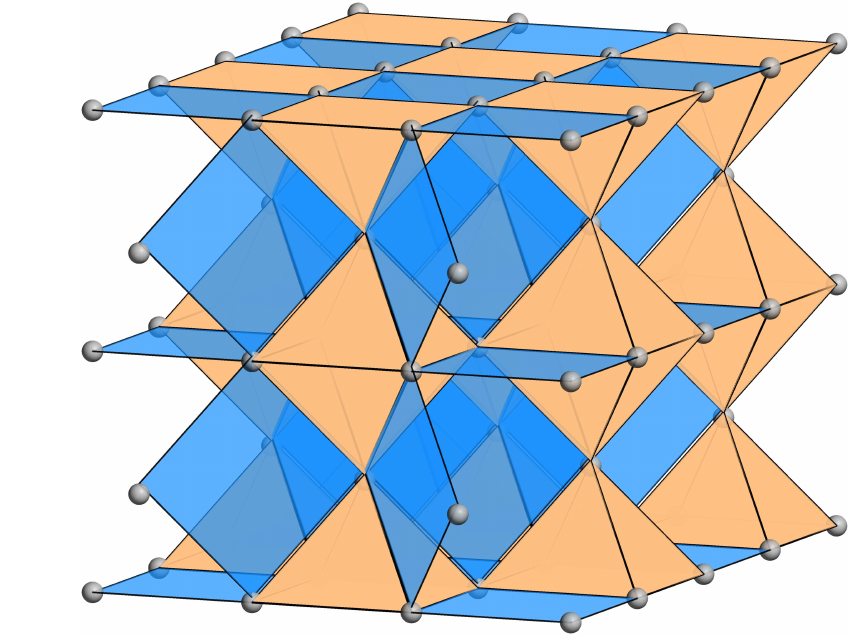}%
    \label{fig:rot3DSC-periodic}
}
\hfill
\subfloat[]{
    \includegraphics{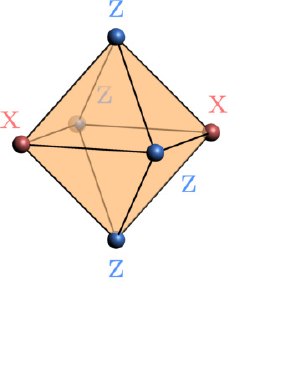}
    \includegraphics{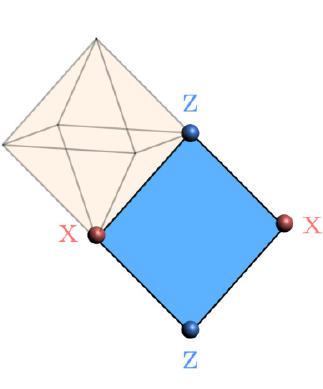}
    \includegraphics{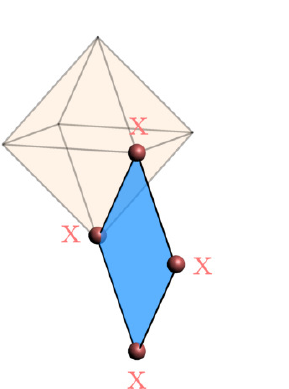}
    \includegraphics{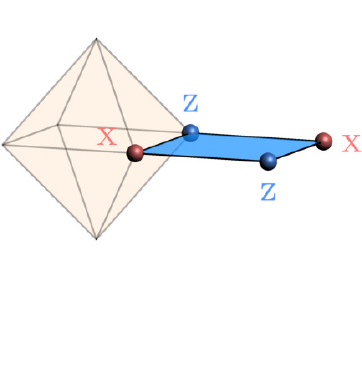}
    \label{fig:rot3DSC-stabilizers-def}
}
\caption{
    Clifford-deformed 3D surface code on a rotated layout with smooth boundaries
    top and bottom. 
    In this representation, qubits live on vertices, while stabilizers live on
    octahedra (orange) and faces (blue),
    which are either diamonds between $xy$ planes or squares on $xy$ planes.
    \textbf{(a)} $4\times 4\times 3$ rotated lattice with open boundaries.
    \textbf{(b)} $4\times 3\times 3$ rotated lattice with periodic
    boundaries.
    Note that although the horizontal $xy$ planes are always periodic,
    if one the of the dimensions along $x$ or $y$ is odd,
    then the diamonds in between $xy$ planes are not periodic along that
    direction,
    as illustrated here (blue diamonds are not periodic in the length-3
    direction).
    \textbf{(c)} Clifford-deformed stabilizer generators:
    (left to right) octahedron acting as $XZZX$ on horizontal
    qubits and $Z$ on vertical qubits above and below,
    diamond acting as $Z$ on horizontal qubits and $X$ on
    vertical qubits,
    diamond acting as $XXXX$,
    square acting as $XZZX$ on horizontal qubits.
} %
\label{fig:rot3DSC}
\end{figure}

To Clifford-deform the code, we apply a Hadamard only on
every second horizontal qubit on each $xy$ plane in a checkerboard manner,
while leaving vertical qubits untouched,
such that the resulting Clifford-deformed stabilizers are as shown in
\cref{fig:rot3DSC-stabilizers-def}.
In the case of an odd-length lattice,
even the stabilizers with defects on the seam
will be of this form after Clifford deformation.
Consequently,
all the octahedron and horizontal square stabilizers act as $XZZX$ when
restricted to horizontal planes,
as shown in \cref{fig:defect-2},
forming $L_z$ coupled layers of 2D \XZZX codes.%

\begin{figure}
    \subfloat[]{
    \includegraphics{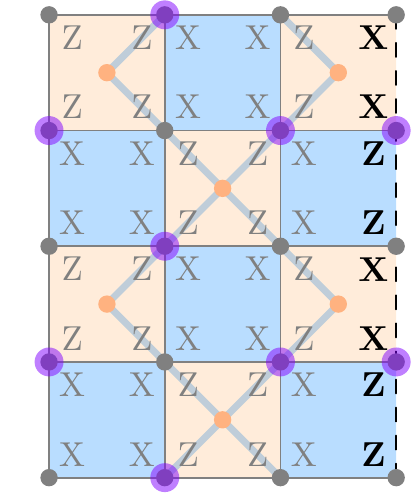}%
    \label{fig:defect-1}
    }
    \subfloat[]{
    \includegraphics{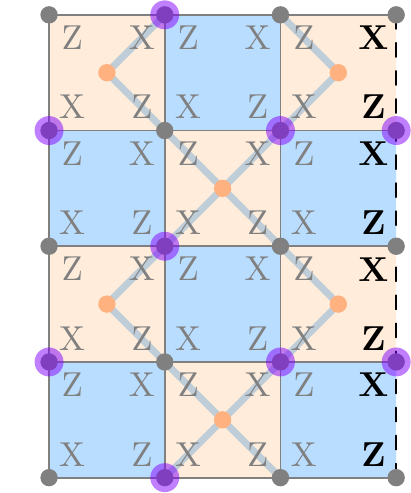}%
    \label{fig:defect-2}
    }
    \caption{Horizontal layer of 3D surface code on rotated
    layout with
    $\textit{odd}\times\textit{even}$
    horizontal dimensions and periodic boundary conditions with
    seam (dashed), like that of \cref{fig:rot3DSC-periodic}.
    The restricted action of octahedron stabilizers (orange)
    and horizontal square stabilizers (blue) on horizontal qubits (vertices)
    in the layer are labeled.
    The diamond stabilizers between adjacent layers,
    whose shadows are drawn as faint blue lines,
    are disconnected across the seam.
    \textbf{(a)} Original stabilizers.
    Stabilizer terms modified by defects on the seam are labeled in bold.
    Qubits where the Clifford deformation is applied as a Hadamard are
    highlighted in purple.
    \textbf{(b)} Clifford-deformed stabilizers.
    Note that all stabilizers restricted on the plane are of the form $XZZX$,
    including those on the seam.
    }
    \label{fig:defect}
\end{figure}

\subsection{Pure $Z$ logical operator}

The 2D \XZZX code on a rotated layout with periodic boundaries and coprime
dimensions has pure $Z$ logical operators supported on $O(L^2)={O}(n)$ physical qubits \cite{XZZX2021}.
The intuition behind this fact is that syndromes propagate on the diagonals, and when the dimensions of the lattice are coprime, 
strings of errors need to wrap around the whole torus in order to form a non-trivial loop. 
As a consequence, the logical error rate $\bar{p}$ for purely $Z$ biased noise scales as

\begin{align}
    \bar{p} \propto e^{-\alpha(p) d_Z} = e^{-\alpha(p) n}
\end{align}
when the physical error rate $p$ goes to zero, where $\alpha(p)$ is a polynomial in $p$ and $d_Z$ is the Z-distance of the code.

We establish a similar result for the Clifford-deformed 3D surface code, summarized in the following theorem:

\begin{restatable}[Lowest-weight $Z$-only logical]{theorem}{theoremsubthreshold}
\label{theorem:subthreshold-scaling}
Consider an $L\times (L+1)\times L_z$ Clifford-deformed 3D rotated surface code
with periodic boundary conditions.
If $L \equiv 1 \text{ or } 2 \mod 4$,
then the lowest-weight logical operator that consists of only
$I$ and $Z$ Pauli operators acts with $Z$ on all horizontal qubits.
\end{restatable}
This means that $Z$-distance $d_Z$ scales as $O\left( L^3\right)$, or in other words, the logical error rate for pure Z biased noise scales as
\begin{align}
    \bar{p} \propto e^{-\alpha(p) n}.
\end{align}

We show here that under pure $Z$ bias noise, the minimum distance of a rotated Clifford-deformed 3D surface code 
(with periodic boundary conditions) can scale as $O(L^3)$. Recall that our code has two types of qubits --
\emph{horizontal qubits} which live on horizontal planes,
and \emph{vertical qubits} which live on vertical edges.
We can establish the following theorem:
\begin{proof}
As we only consider $Z$ errors,
it suffices to work with classical parity-check operators
that detect $Z$ errors
rather than the full quantum stabilizers,
using the map $I\mapsto 0$, $X\mapsto 1$, $Y\mapsto 1$, $Z\mapsto 0$.
The horizontal parity-check operators are all horizontal squares of the form $0110$
while there are two types of vertical parity checks --
$1111$ and $0110$; see \cref{fig:rotated-classical-parity-checks}.
We divide our proof into 3 distinct parts.
In Part 1, we show that if a $Z$-logical has support in a horizontal qubit, then it has support in the whole horizontal layer containing this qubit.
In Part 2, we show that if a $Z$-logical has support in a horizontal layer and $L \equiv 1 \text{ or } 2 \mod 4$, then it has support in all horizontal layers.
Finally, we show in Part 3 that a $Z$-logical cannot have support uniquely in vertical qubits. It shows that if $L \equiv 1 \text{ or } 2 \mod 4$, the minimum-weight logical is the one with all horizontal qubits activated.

\begin{figure}[t]
\centering
\includegraphics[scale=0.9]{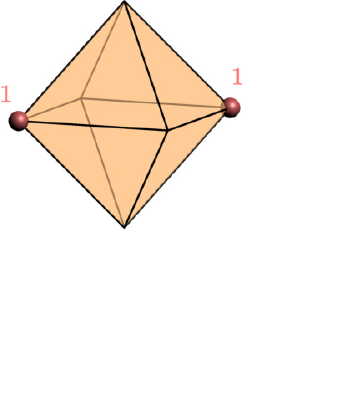}
\includegraphics[scale=0.9]{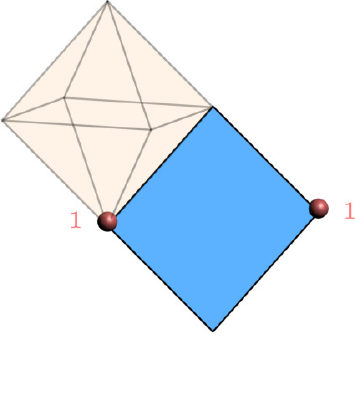}
\includegraphics[scale=0.9]{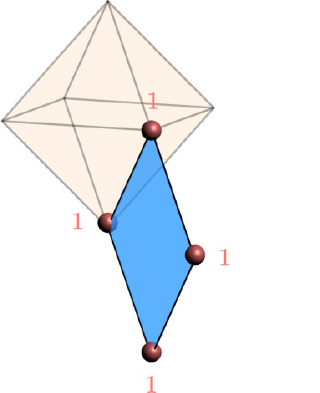}
\includegraphics[scale=0.9]{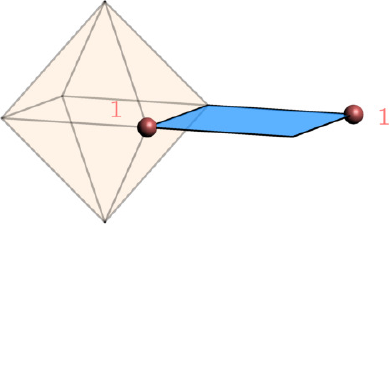}
\caption{%
    Stabilizers can be converted into binary parity checks for $Z$ errors.
}
\label{fig:rotated-classical-parity-checks}
\end{figure}

\begin{figure}[t]
    \centering

    \subfloat[]{
        \includegraphics[scale=0.85]{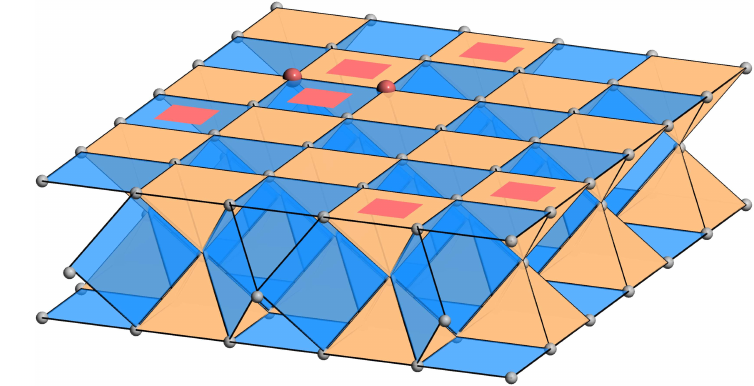}
        \label{fig:rotated-proof-1}
    }
    \subfloat[]{
        \includegraphics[scale=0.85]{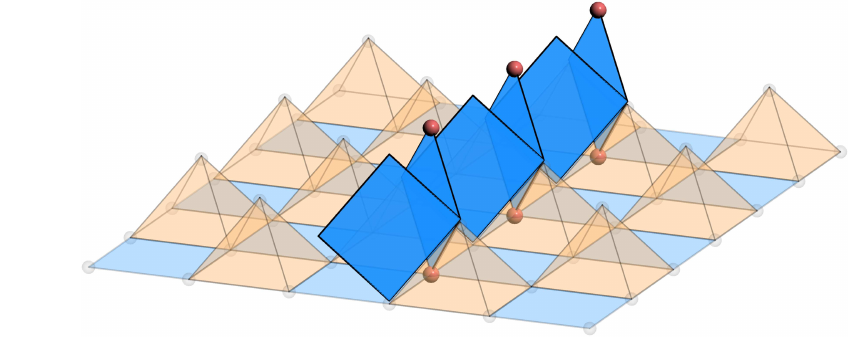}
        \label{fig:rotated-proof-2}
    }
    \caption{Proof that the pure Z logical of our tailored 3D rotated surface code is supported on all the qubits of the code.
        \textbf{(a)} A parity check between any two horizontal qubits (red dots) can be constructed by wrapping around the coprime plane along a diagonal (red squares).
        \textbf{(b)} Product of zig-zag parity-check operators along the $4k + 2$ periodic direction, which consists of $2k+1$ $0110$-type checks and $2k+1$ $1111$-type checks.
    }%
    \label{fig:rotated-proof}
\end{figure}

\begin{enumerate}[(1)]
\item
Let $\mathcal{L}$ be a pure-$Z$ logical with a horizontal qubit in its support. 
We show that $\mathcal{L}$ is supported on the whole horizontal layer containing this qubit.
This problem reduces to showing that there is always a weight-2 parity check between every pair of qubits in a chosen layer.
Indeed, if that is the case, it means that the parity of every pair of qubit must be $0$, which eliminate the possibility of layer
not entirely composed of $0$ or $1$.
To prove this, we notice that every diagonal line in the 2D coprime lattice forms a parity check with a $1$ at its boundary,
as illustrated in \cref{fig:rotated-proof-1}.
Since the lattice has coprime dimensions, there exists a diagonal line that go through all the qubits before looping to itself.
In particular, this line goes through every pair of qubits, which achieves our proof.

\item
If $L \equiv 1 \text{ or } 2 \mod 4$, it means that one of the direction contains $4k + 2$ cells, for some integer $k \geq 0$.
Consider a parity-check operator that consists of a product
of parity-check operators in a zig-zag fashion that wraps
through the periodic boundaries of the vertical qubit lattice, as shown in \cref{fig:rotated-proof-2}.
Such a chain would be made up of $2k + 1$ checks of type $1111$
and $2k + 1$ checks of type $0110$.
Their product is an operator that has support on $2k + 1$ horizontal qubits on the horizontal layer below and
$2k + 1$ horizontal qubits on the horizontal layer above in an identical manner.

Recall from Part 1 that there exists a parity operator that acts on
any two horizontal qubits on the same layer.
By composing this pair-wise parity check between every pair of qubits
on a layer,
the layers can be removed pairwise,
leaving only one remaining qubit on that layer.
The exact procedure can be applied to the other layer,
resulting in a parity-check operator that acts only on two qubits,
one in each layer.

This provides a constraint that the $Z$-only logical must be
the same between layers of horizontal qubits.
That is, either all horizontal qubits are $I$ or all horizontal qubits are $Z$.

\item
We show here that there cannot be a logical operator that has support only on vertical edges.
We first notice that there is only one logical qubit encoded in the Clifford-deformed 3D rotated surface code.
One pair of anti-commuting logical operators has a string of horizontal qubits of the form $XZXZ\cdots$,
and a membrane of horizontal qubits made of $Y$ operators.

Since both the string and the membrane logical operators do not act on vertical qubits,
they would commute with an only-vertical operator.
And since there is only one logical qubit, such an operator cannot be a logical operator.

\end{enumerate}
The operator that acts with $Z$ on all horizontal qubits is a valid logical
since it anticommutes with the $Y$ membrane logical and commutes with all the
stabilizers.
We showed that there cannot be a pure-$Z$ logical of lower weight that acts non-trivially
on a horizontal qubit, or that is supported only on vertical qubits.
Therefore, the minimum weight $Z$-only logical is one that acts as $Z$ on all horizontal qubits.
\end{proof}

Note that the proof fails if the even dimension is not $4k + 2$
for some integer $k$, as the constraint between layers (see Part 2) does not apply,
in which case the scaling becomes $O(L^2)$ instead.

\subsection{Robustness of the $Z$-weight scaling}

\begin{figure}[t]
    \centering
    \subfloat[]{%
        \includegraphics{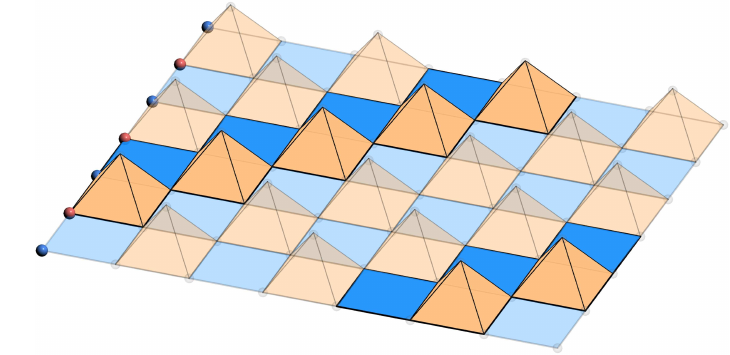}%
        \label{fig:rotated-fragile-membrane-1}
    }%
    \subfloat[]{%
        \includegraphics{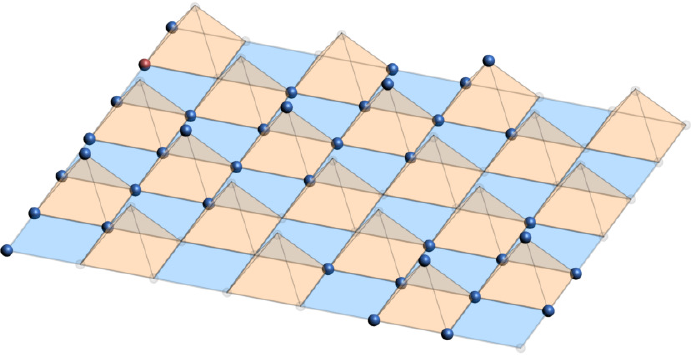}%
        \label{fig:rotated-fragile-membrane-2}
    }
    \caption{Example of membrane logical operator made of $O(1)$ $X$s and $\Theta(L^2)$ $Z$s.
        \textbf{(a)} The membrane can be obtained by starting from a string logical made of 
        alternating $X$s (red vertices) and $Z$s (blue vertices), and applying stabilizers (highlighted in the figure) 
        along the diagonal lines connecting pairs of $X$s. 
        This has the effect of annihilating all the connected pairs of $X$s,
        leaving only one unpaired $X$ in the logical operator,
        and a trail of $\Theta(L^2)$ $Z$s.
        \textbf{(b)} Resulting membrane logical operator.
    }%
    \label{fig:rotated-fragile-membrane}
\end{figure}

\begin{figure}[t]
    \centering
    \subfloat[]{
        \includegraphics{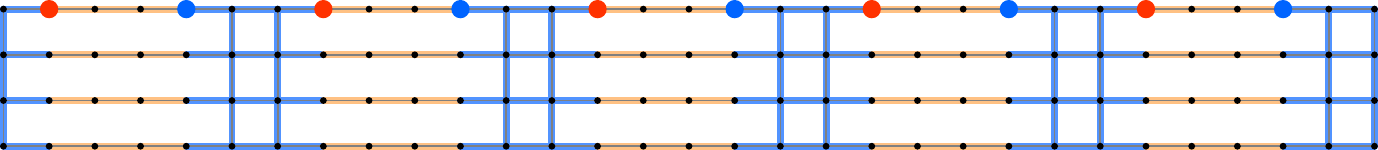}
        \label{fig:rotated-no-string-proof-1}
    }
    \\
    \subfloat[]{
        \includegraphics{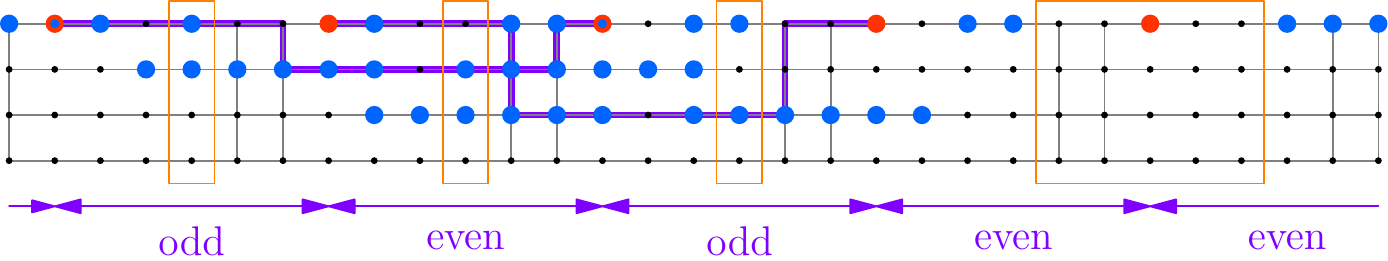}
        \label{fig:rotated-no-string-proof-2}
    }
    \caption{Proof that there is no string logical made of $O(1)$ Xs and $O(L)$ Zs.
        \textbf{(a)} All the horizontal qubits can be placed on a long diagonal slice of size $(L_x \cdot L_y, L_z)$.
        Considering only $X$s acting on horizontal qubits,
        all stabilizers become 2-body terms.
        The lattice shown is the diagonal slice of a 3D rotated surface code with dimensions $(10,3,4)$. 
        Horizontal qubits are represented as black vertices, octahedron stabilizers as orange edges and face stabilizers as blue edges.
        The leftmost and rightmost vertices are identified to represent periodic
        boundary conditions.
        The support of the $XZ$-string logical operator is represented by red ($X$) and blue ($Z$) vertices.
        \textbf{(b)}
        Example of a matching solution for two pairs of $X$s belonging to
        the $XZ$-string logical operator.
        The purple edges represent the stabilizers used in the matching solution. 
        This results in eliminating all but one of the $X$s (red),
        while introducing new $Z$s (blue) on horizontal qubits
        located $L_y+1$ vertices to the left and to the right of every horizontal
        stabilizer used in the matching solution.
        Every column of horizontal stabilizers either has an odd or even
        number used in the matching solution,
        which is its parity.
        All $2L_y - 1$ columns of horizontal stabilizers between any two
        neighboring $X$s have the same parity,
        as annotated in purple.
        Adjacent sections have different parities,
        except for the sections around unmatched $X$s,
        which are instead of equal parity
        (the two rightmost sections).
        As a result of this parity alternation,
        there are an odd number of $Z$s introduced in every column
        of horizontal qubits,
        except in columns of horizontal qubits at the midpoint between
        neighboring $X$s,
        and the columns of horizontal qubits around unmatched $X$s.
        Since there are only $O(L)$ such exceptions (boxed in orange),
        we deduce that $\Theta(L^2)$ $Z$s have been introduced in the process.
    }%
    \label{fig:no-string-proof}
\end{figure}

We note that, as in the 2D case~\cite{higgotFragile2022}, the above statement is not robust, 
in the sense that allowing for a single $X$ into our logical operator drops the effective scaling down from $O(L^3)$ to $O(L^2)$. 
This can be seen in \cref{fig:rotated-fragile-membrane},
where we present an example of a logical operator that has $O(1)$ Pauli $X$s and
$\Theta(L^2)$ Pauli $Z$s.

However, we show that there is no string logical operator that contains $O(L)$ $Z$s and $O(1)$ $X$s as follows:

\begin{theorem} \label{theorem:rotated-no-string}
    In a 3D rotated surface code
    whose dimensions satisfy the assumptions of
    \cref{theorem:subthreshold-scaling}, any logical operator containing $O(1)$ $Xs$ also contains $\Theta(L^2)$ $Zs$.
\end{theorem}

\begin{proof}
    To prove this theorem, we adopt the following strategy. 
    We first show that if such a logical operator exists, it must belong to the same coset as the logical string operator 
    that comes from the Clifford deformation of an $X$ string logical operator.
    This logical operator, which we call the $XZ$-string, by virtue of it having alternating $X$ and $Z$s along its length, is represented in \cref{fig:rotated-fragile-membrane-1}.
    We then derive all the transformations of this string, through the application of stabilizers, that 
    results in $O(1)$ $X$s on horizontal qubits. 
    For that, we show that this is equivalent to finding all the solutions of a matching problem on a 2D lattice,
    and prove that all such solutions necessarily result in creating $\Theta(L^2)$ $Z$s on horizontal qubits.

    Let us start by showing that our string logical operator must be logically equivalent to the $XZ$-string. 
    Since the $XZ$-string is free to move in 3D by application of stabilizer generators, any other string logical operator must commute with at least
    one of its instances by avoidance. Moreover, the 3D rotated surface code with the dimensions of \cref{theorem:subthreshold-scaling} 
    only encodes one qubit. 
    Therefore, any other string logical operator must either be trivial or logically equivalent to the $XZ$-string.

    We now show that any logical operator equivalent to the $XZ$-string with $O(1)$ $X$s on the horizontal qubits 
    also has $\Theta(L^2)$ $Z$s on the horizontal qubits.
    This statement, while restricted on the horizontal qubits, implies the stronger result that requiring $O(1)$ $X$s on
    both the horizontal and vertical qubits leads to the presence of $\Theta(L^2)$ $Z$s in the operator.
    Therefore, we ignore the vertical qubits in the rest of the proof.

    The goal is now to show that, by applying stabilizers on the $XZ$-string, we can eliminate all $X$s except $O(1)$ of them.
    To see how $X$s can be moved and eliminated, let us focus on the $X$ part of the stabilizers. 
    This corresponds to the parity-check operators of \cref{fig:rotated-classical-parity-checks}. 
    Since we choose to ignore the vertical qubits, we consider the restriction of these stabilizers to the
    the horizontal qubits. The final restricted operators are all weight-2.
    We can therefore represent all the horizontal qubits and stabilizers on a 2D lattice, constructed by taking a diagonal slice
    of the 3D rotated surface code lattice that is vertical and runs parallel to
    the line connecting a pair of $X$s on a horizontal square stabilizer. Since the code has coprime dimensions, there is only a single such diagonal slice
    on which all the horizontal qubits sit. This 2D lattice is represented in \cref{fig:rotated-no-string-proof-1}.

    The problem can now be formulated as a matching problem on this 2D lattice.
    Indeed, since the stabilizers have $X$-weight 2 restricted on horizontal qubits, the $X$s can only be annihilated in pairs, by applying a chain
    of stabilizers that connects the pair.
    Since we allow $O(1)$ $X$s to remain in the logical operator, the more precise problem is to match
    all but $O(1)$ $X$s. Any logical operators with $O(1)$ $X$s can then be seen as a different solution
    to this matching problem. 

    The next step is to count how many $Z$s are created for each matching solution.
    When applying a horizontal stabilizer,
    which is either an octahedron stabilizer or a horizontal square stabilizer,
    two $Z$s are introduced on the horizontal qubits,
    located $L_y+1$ vertices to the left and to the right of the stabilizer
    as viewed on the 2D diagonal slice.
    Note that additional $Z$s are also introduced on the vertical qubits for
    octahedron stabilizers.
    An example of a matching solution with its introduced $Z$s is shown in \cref{fig:rotated-no-string-proof-2}.
    To count them, we note that any matching solution has an alternation of even-parity and odd-parity sections,
    where a section is defined as the space between two original $X$s, and its parity is defined as the number of horizontal stabilizers
    applied on each column of the section, modulo two.
    Since we can choose $O(1)$ $X$s that are not matched, this alternating
    parity pattern
    breaks at these unmatched $X$s, where either two even-parity sections or two
    odd-parity sections follow one another.
    This can be seen on the rightmost sections of \cref{fig:rotated-no-string-proof-2}.
    However, since there are only $O(1)$ unmatched $X$s, the number of such breaks in alternation is also $O(1)$.

    We can then use this last fact to prove that the number of $Z$s is $\Theta(L^2)$.
    Indeed, the number of $Z$s on a given column of horizontal qubits is, by construction, equal to the number of stabilizers
    applied $L_y+1$ columns to the left and to the right.
    Those two columns of stabilizers are $2L_y - 1$ edges apart, and since the size of a section is $2L_y$ edges,
    they belong in different sections as long the column of $Z$s is not precisely in the middle of a section.
    Excluding these $O(L)$ columns, as well as the $O(L)$ columns located $L_y$
    edges to the left and to the right of a parity-alternation breaking point,
    of which there are only $O(1)$,
    we can see that the number of $Z$s applied to a given column is equal to the sum of the number of $Z$s in an odd-parity and in an even-parity section.
    Therefore, in $\Theta(L^2)$ columns, there is an odd number of $Z$s,
    where the number of $Z$s must be at least 1.
    The logical operator therefore contains $\Theta(L^2)$ $Z$s.
\end{proof}

This theorem shows that under high dephasing bias, with a low number of $X$ errors, the subthreshold error rate
scales as $O(e^{-\alpha L^2})$, similarly to the original 3D rotated surface code. 
However, the number of logical operators with a low number of $X$s is expected to be lower in the Clifford-deformed code,
and we therefore expect the coefficient $\alpha$ to be higher. 
This can be seen through the following heuristic argument.
In any transformation of the $XZ$-string considered in the proof of \cref{theorem:rotated-no-string}, 
applying a vertical stabilizer creates some $X$s on the vertical qubits. 
Those $X$s cannot be eliminated through the application of other stabilizers, 
so vertical stabilizers necessarily increase the number of $X$s in the operator.
In order to keep only one $X$ in our operator, the matching must therefore be performed on the horizontal plane.
But, for a fixed choice of which $X$ to keep, there are only two possible matching solutions confined to the plane.
By moving the remaining $X$ on the plane, or choosing different planes, we can deduce that the number of membrane operators
with a single $X$ scales as $\Theta(L^3)$.
This can be compared to the original 3D rotated surface code, where the number of such logical operators scales exponentially
in the system size.
This reasoning can be generalized to logical operators containing $O(1)$ $X$s. 
When more than one $X$ is present, each logical operator is characterized by the planes
in which the $X$s are supported, and the positions of these $X$s within these planes.
This is due to the confinement property discussed for the case of a single $X$.
The number of such choices still scales polynomially with the system size, and 
hence remains an exponential improvement compared to the original code.

\section{Discussion}\label{sec:discussion}

In this work, we presented Clifford deformations of many 3D topological codes with high quantum memory threshold error rates for biased Pauli noise. One important question following our study is whether it is always possible to design a Clifford deformation of a topological stabilizer code such that there exists a decoding strategy with 50\% threshold error rate at infinitely biased noise. On the basis of the wide range of examples we present in this work, we conjecture that this is true. We also presented a rotated layout of the surface code for which choosing appropriate dimensions and boundary conditions leads to a subthreshold scaling of $\exp{-O(n)}$, for infinitely biased noise. We showed that in the regime of large finite bias, which we model as the presence of $O(1)$ $X$ errors, this subthreshold scaling becomes $\exp{-O(L^2)}$.  It would be interesting to consider how such geometrical optimizations can improve the code performance for other 3D codes such as the 3D color code. 

Families of random Clifford-deformed surface codes in two dimensions have been shown to exhibit high threshold error rates and subthreshold scaling better than XZZX and XY surface codes~\cite{dua2022clifforddeformed}. The performance of the random codes at infinite bias can be intuitively explained via a mapping to percolation problems. One could consider random Clifford-deformed 3D surface codes and color codes for which we expect a similar mapping to percolation problems and a phase diagram containing a phase of 50\% threshold error rate analogous to the random Clifford-deformed surface codes in 2D. It would be also interesting to study how random Clifford deformations affect the memory performance of fracton codes at infinite bias, which have intrinsically rigid logical operators irrespective of the bias.

A natural next step is to extend the code capacity results in our work to the phenomenological fault-tolerant scenario as well as the more realistic circuit-level scenario. In the circuit-level scenario, it becomes important to use bias-preserving gates to maintain the performance advantage found for biased noise.

In three dimensions, surface codes have been defined on fractal lattices with Hausdorff dimension $D_H=2+\epsilon$~\cite{Zhu_fractal_2022,Dua_fractal_2022}. Our Clifford deformation of the 3D surface code naturally applies to such fractal surface codes that can be created by punching holes in the 3D surface code.

\paragraph{Code availability}
The source code for the numerical simulations of the quantum error correcting
codes, noise models and decoders described in this paper are publicly available
on a GitHub repository at
\href{https://github.com/panqec/panqec}{github.com/panqec/panqec}.
This is the source code repository for the PanQEC Python package
(pronounced ``pancake'').
The vision for PanQEC is to be a collection of
quantum error correcting codes, noise models and decoders, which are amenable to numerical
simulation and interactive three-dimensional visualization.
PanQEC's documentation is available at
\href{https://panqec.readthedocs.io/en/latest/}{panqec.readthedocs.io},
which also includes tutorials on usage.
Furthermore, an online demonstration of its 3D visualization capabilities are
available at
\href{https://gui.quantumcodes.io/}{gui.quantumcodes.io}.

\vspace{-2mm}
\paragraph{Acknowledgments}
We thank Benjamin Brown for showing how to prove a nontrivial part of Theorem 1 i.e.\ the decoding of the plaquette syndromes. 
We thank Steve Flammia for comments, especially for asking whether our rotated layout has the property proven in Theorem~\ref{theorem:rotated-no-string}. 
We thank Dan Browne, Michael Gullans, Oscar Higgott, Armanda Quintavalle, Joschka Roffe and George Umbrarescu for comments on the manuscript. 
AP is supported by the Engineering and Physical Sciences Research Council (EP/S021582/1).
EH was supported by the Perimeter Scholars International scholarship and the
Fulbright Future Scholarship.
EH acknowledges support from the National Science Foundation (QLCI grant
OMA-2120757)
Research at Perimeter Institute is supported in part by the Government of Canada through the Department of Innovation, Science and Economic Development Canada and by the Province of Ontario through the Ministry of Colleges and Universities.
CTC acknowledges support from the Swiss National Science Foundation through the Sinergia grant CRSII5-186364, and for the NCCRs QSIT and SwissMAP. AD is supported by the Simons Foundation through the collaboration
on Ultra-Quantum Matter (651438, AD) and by the Institute for Quantum
Information and Matter, an NSF Physics Frontiers Center (PHY-1733907).

\vspace{-2mm}
\paragraph{Correspondence:} adua@caltech.edu

\bibliographystyle{quantum}  
\bibliography{bib,chubbbib}

\begin{thebibliography}{10}

\bibitem{AliferisPreskill_bias2008}
Panos Aliferis and John Preskill.
\newblock ``Fault-tolerant quantum computation against biased noise''.
\newblock \href{https://dx.doi.org/10.1103/PhysRevA.78.052331}{Phys. Rev. A
  {\bf 78}, 052331}~(2008).

\bibitem{Ultrahigh2018}
David~K. Tuckett, Stephen~D. Bartlett, and Steven~T. Flammia.
\newblock ``Ultrahigh error threshold for surface codes with biased noise''.
\newblock \href{https://dx.doi.org/10.1103/physrevlett.120.050505}{Physical
  Review Letters {\bf 120}, 050505}~(2018).
\newblock  \href{http://arxiv.org/abs/1708.08474}{arXiv:1708.08474}.

\bibitem{aliferis2009fault}
Panos Aliferis, Frederico Brito, David~P DiVincenzo, John Preskill, Matthias
  Steffen, and Barbara~M Terhal.
\newblock ``Fault-tolerant computing with biased-noise superconducting qubits:
  a case study''.
\newblock New Journal of Physics {\bf 11}, 013061~(2009).
\newblock
  url:~\href{http://dx.doi.org/10.1088/1367-2630/11/1/013061}{http://dx.doi.org/10.1088/1367-2630/11/1/013061}.

\bibitem{Nigg_2014}
D.~Nigg, M.~Muller, E.~A. Martinez, P.~Schindler, M.~Hennrich, T.~Monz, M.~A.
  Martin-Delgado, and R.~Blatt.
\newblock ``Quantum computations on a topologically encoded qubit''.
\newblock \href{https://dx.doi.org/10.1126/science.1253742}{Science {\bf 345},
  302–305}~(2014).

\bibitem{burkard2021}
Guido Burkard, Thaddeus~D. Ladd, John~M. Nichol, Andrew Pan, and Jason~R.
  Petta.
\newblock ``Semiconductor spin qubits''~(2021).
\newblock  \href{http://arxiv.org/abs/2112.08863}{arXiv:2112.08863}.

\bibitem{puri2020bias}
Shruti Puri, Lucas St-Jean, Jonathan~A Gross, Alexander Grimm, Nicholas~E
  Frattini, Pavithran~S Iyer, Anirudh Krishna, Steven Touzard, Liang Jiang,
  Alexandre Blais, et~al.
\newblock ``Bias-preserving gates with stabilized cat qubits''.
\newblock \href{https://dx.doi.org/10.1126/sciadv.aay590}{Science advances {\bf
  6}, eaay5901}~(2020).

\bibitem{XZZX2021}
J.~Pablo Bonilla~Ataides, David~K. Tuckett, Stephen~D. Bartlett, Steven~T.
  Flammia, and Benjamin~J. Brown.
\newblock ``The {XZZX} surface code''.
\newblock \href{https://dx.doi.org/10.1038/s41467-021-22274-1}{Nature
  Communications {\bf 12}, 2172}~(2021).
\newblock  \href{http://arxiv.org/abs/2009.07851}{arXiv:2009.07851}.

\bibitem{dua2022clifforddeformed}
Arpit Dua, Aleksander Kubica, Liang Jiang, Steven~T. Flammia, and Michael~J.
  Gullans.
\newblock ``Clifford-deformed surface codes''~(2022)
  \href{http://arxiv.org/abs/2201.07802}{arXiv:2201.07802}.

\bibitem{tiurevCorrectingNonindependentNonidentically2022}
Konstantin Tiurev, Peter-Jan H.~S. Derks, Joschka Roffe, Jens Eisert, and
  Jan-Michael Reiner.
\newblock ``Correcting non-independent and non-identically distributed errors
  with surface codes''~(2022).
\newblock  \href{http://arxiv.org/abs/2208.02191}{arXiv:2208.02191}.

\bibitem{srivastava2021xyz}
Basudha Srivastava, Anton~Frisk Kockum, and Mats Granath.
\newblock ``The xyz$^2$ hexagonal stabilizer code''~(2021).
\newblock
  url:~\href{https://arxiv.org/abs/2112.06036}{arxiv.org/abs/2112.06036}.

\bibitem{miguel2022cellular}
Jonathan F~San Miguel, Dominic~J Williamson, and Benjamin~J Brown.
\newblock ``A cellular automaton decoder for a noise-bias tailored color
  code''~(2022).

\bibitem{bombinTopologicalComputationBraiding2007}
H.~Bomb{\'\i}n and M.~A. {Martin-Delgado}.
\newblock ``Topological {{Computation}} without {{Braiding}}''.
\newblock \href{https://dx.doi.org/10.1103/PhysRevLett.98.160502}{Phys. Rev.
  Lett. {\bf 98}, 160502}~(2007).

\bibitem{bombinGaugeColorCodes2015}
H{\'e}ctor Bomb{\'\i}n.
\newblock ``Gauge color codes: Optimal transversal gates and gauge fixing in
  topological stabilizer codes''.
\newblock \href{https://dx.doi.org/10.1088/1367-2630/17/8/083002}{New J. Phys.
  {\bf 17}, 083002}~(2015).

\bibitem{kubicaUniversalTransversalGates2015}
Aleksander Kubica and Michael~E. Beverland.
\newblock ``Universal transversal gates with color codes - a simplified
  approach''.
\newblock \href{https://dx.doi.org/10.1103/PhysRevA.91.032330}{Phys. Rev. A
  {\bf 91}, 032330}~(2015).

\bibitem{kubica2015unfolding}
Aleksander Kubica, Beni Yoshida, and Fernando Pastawski.
\newblock ``Unfolding the color code''.
\newblock \href{https://dx.doi.org/10.1088/1367-2630/17/8/083026}{New Journal
  of Physics {\bf 17}, 083026}~(2015).

\bibitem{Vasmer2019}
Michael Vasmer and Dan~E. Browne.
\newblock ``Three-dimensional surface codes: Transversal gates and
  fault-tolerant architectures''.
\newblock \href{https://dx.doi.org/10.1103/physreva.100.012312}{Physical Review
  A{\bf 100}}~(2019).
\newblock  \href{http://arxiv.org/abs/1801.04255}{arXiv:1801.04255}.

\bibitem{breuckmann2017a}
Nikolas~P. Breuckmann, Kasper Duivenvoorden, Dominik Michels, and Barbara~M.
  Terhal.
\newblock ``Local decoders for the {{2D}} and {{4D}} toric code''.
\newblock Quantum Information \& Computation {\bf 17}, 181--208~(2017).
\newblock  \href{http://arxiv.org/abs/1609.00510}{arXiv:1609.00510}.

\bibitem{duivenvoordenRenormalizationGroupDecoder2019}
Kasper Duivenvoorden, Nikolas~P. Breuckmann, and Barbara~M. Terhal.
\newblock ``Renormalization group decoder for a four-dimensional toric code''.
\newblock \href{https://dx.doi.org/10.1109/TIT.2018.2879937}{IEEE Trans.
  Inform. Theory {\bf 65}, 2545--2562}~(2019).

\bibitem{kubica2019cellular}
Aleksander Kubica and John Preskill.
\newblock ``Cellular-automaton decoders with provable thresholds for
  topological codes''.
\newblock \href{https://dx.doi.org/10.1103/PhysRevLett.123.020501}{Physical
  Review Letters {\bf 123}, 020501}~(2019).
\newblock  \href{http://arxiv.org/abs/1809.10145}{arXiv:1809.10145}.

\bibitem{vasmer2021sweep}
Michael Vasmer, Dan~E. Browne, and Aleksander Kubica.
\newblock ``Cellular automaton decoders for topological quantum codes with
  noisy measurements and beyond''.
\newblock \href{https://dx.doi.org/10.1038/s41598-021-81138-2}{Scientific
  Reports {\bf 11}, 2027}~(2021).
\newblock  \href{http://arxiv.org/abs/2004.07247}{arXiv:2004.07247}.

\bibitem{quintavalle2021single}
Armanda~O. Quintavalle, Michael Vasmer, Joschka Roffe, and Earl~T. Campbell.
\newblock ``Single-shot error correction of three-dimensional homological
  product codes''.
\newblock \href{https://dx.doi.org/10.1103/PRXQuantum.2.020340}{PRX Quantum
  {\bf 2}, 020340}~(2021).
\newblock  \href{http://arxiv.org/abs/2009.11790}{arXiv:2009.11790}.

\bibitem{higgottImprovedSingleshotDecoding2022}
Oscar Higgott and Nikolas~P. Breuckmann.
\newblock ``Improved single-shot decoding of higher dimensional hypergraph
  product codes''~(2022).
\newblock  \href{http://arxiv.org/abs/2206.03122}{arXiv:2206.03122}.

\bibitem{bombin2015a}
H{\'e}ctor Bomb{\'\i}n.
\newblock ``Single-{{Shot Fault-Tolerant Quantum Error Correction}}''.
\newblock \href{https://dx.doi.org/10.1103/PhysRevX.5.031043}{Physical Review X
  {\bf 5}, 031043}~(2015).
\newblock  \href{http://arxiv.org/abs/1404.5504}{arXiv:1404.5504}.

\bibitem{brownFaulttolerantErrorCorrection2016}
Benjamin~J. Brown, Naomi~H. Nickerson, and Dan~E. Browne.
\newblock ``Fault-tolerant error correction with the gauge color code''.
\newblock \href{https://dx.doi.org/10.1038/ncomms12302}{Nat Commun {\bf 7},
  12302}~(2016).

\bibitem{kubica2021single}
Aleksander Kubica and Michael Vasmer.
\newblock ``Single-shot quantum error correction with the three-dimensional
  subsystem toric code''~(2021).
\newblock  \href{http://arxiv.org/abs/2106.02621}{arXiv:2106.02621}.

\bibitem{bombinResilienceTimeCorrelatedNoise2016}
H{\'e}ctor Bomb{\'\i}n.
\newblock ``Resilience to {{Time-Correlated Noise}} in {{Quantum
  Computation}}''.
\newblock \href{https://dx.doi.org/10.1103/PhysRevX.6.041034}{Phys. Rev. X {\bf
  6}, 041034}~(2016).

\bibitem{brown2020parallelized}
Benjamin~J. Brown and Dominic~J. Williamson.
\newblock ``Parallelized quantum error correction with fracton topological
  codes''.
\newblock \href{https://dx.doi.org/10.1103/physrevresearch.2.013303}{Physical
  Review Research{\bf 2}}~(2020).

\bibitem{Zhu_fractal_2022}
Guanyu Zhu, Tomas Jochym-O'Connor, and Arpit Dua.
\newblock ``Topological order, quantum codes, and quantum computation on
  fractal geometries''.
\newblock \href{https://dx.doi.org/10.1103/PRXQuantum.3.030338}{PRX Quantum
  {\bf 3}, 030338}~(2022).

\bibitem{Dua_fractal_2022}
Arpit Dua, Tomas Jochym-O'Connor, and Guanyu Zhu.
\newblock ``Quantum error correction with fractal topological codes''~(2022)
  \href{http://arxiv.org/abs/2201.03568}{arXiv:2201.03568}.

\bibitem{cai2022looped}
Zhenyu Cai, Adam Siegel, and Simon Benjamin.
\newblock ``Looped pipelines enabling effective 3d qubit lattices in a strictly
  2d device''~(2022).
\newblock  \href{http://arxiv.org/abs/2203.13123}{arXiv:2203.13123}.

\bibitem{Buonacorsi_2019}
Brandon Buonacorsi, Zhenyu Cai, Eduardo~B Ramirez, Kyle~S Willick, Sean~M
  Walker, Jiahao Li, Benjamin~D Shaw, Xiaosi Xu, Simon~C Benjamin, and Jonathan
  Baugh.
\newblock ``Network architecture for a topological quantum computer in
  silicon''.
\newblock \href{https://dx.doi.org/10.1088/2058-9565/aaf3c4}{Quantum Science
  and Technology {\bf 4}, 025003}~(2019).
\newblock  \href{http://arxiv.org/abs/1807.09941}{arXiv:1807.09941}.

\bibitem{akhtar2022high}
M.~Akhtar, F.~Bonus, F.~R. Lebrun-Gallagher, N.~I. Johnson, M.~Siegele-Brown,
  S.~Hong, S.~J. Hile, S.~A. Kulmiya, S.~Weidt, and W.~K. Hensinger.
\newblock ``A high-fidelity quantum matter-link between ion-trap microchip
  modules''~(2022).
\newblock  \href{http://arxiv.org/abs/2203.14062}{arXiv:2203.14062}.

\bibitem{lukingroupcoldatoms}
Dolev Bluvstein, Harry Levine, Giulia Semeghini, Tout~T. Wang, Sepehr Ebadi,
  Marcin Kalinowski, Alexander Keesling, Nishad Maskara, Hannes Pichler, Markus
  Greiner, Vladan Vuleti{\'c}, and Mikhail~D. Lukin.
\newblock ``A quantum processor based on coherent transport of entangled atom
  arrays''.
\newblock \href{https://dx.doi.org/10.1038/s41586-022-04592-6}{Nature {\bf
  604}, 451--456}~(2022).
\newblock  \href{http://arxiv.org/abs/2112.03923}{arXiv:2112.03923}.

\bibitem{mallek2021fabrication}
Justin~L. Mallek, Donna-Ruth~W. Yost, Danna Rosenberg, Jonilyn~L. Yoder,
  Gregory Calusine, Matt Cook, Rabindra Das, Alexandra Day, Evan Golden,
  David~K. Kim, Jeffery Knecht, Bethany~M. Niedzielski, Mollie Schwartz, Arjan
  Sevi, Corey Stull, Wayne Woods, Andrew~J. Kerman, and William~D. Oliver.
\newblock ``Fabrication of superconducting through-silicon vias''~(2021).
\newblock  \href{http://arxiv.org/abs/2103.08536}{arXiv:2103.08536}.

\bibitem{2017_3DISC}
D.~Rosenberg, D.~Kim, R.~Das, D.~Yost, S.~Gustavsson, D.~Hover, P.~Krantz,
  A.~Melville, L.~Racz, G.~O. Samach, and et~al.
\newblock ``3d integrated superconducting qubits''.
\newblock \href{https://dx.doi.org/10.1038/s41534-017-0044-0}{npj Quantum
  Information {\bf 3}, 42}~(2017).
\newblock  \href{http://arxiv.org/abs/1706.04116}{arXiv:1706.04116}.

\bibitem{IBM3D}
Jerry Chow, Oliver Dial, and Jay Gambetta.
\newblock ``$\text{IBM Quantum}$ breaks the 100‑qubit processor barrier''.
\newblock
  \url{https://research.ibm.com/blog/127-qubit-quantum-processor-eagle}~(2021).

\bibitem{bartolucci2021}
Sara Bartolucci, Patrick Birchall, H\'{e}ctor Bomb\'{i}n, Hugo Cable, Chris
  Dawson, Mercedes {Gimeno-Segovia}, Eric Johnston, Konrad Kieling, Naomi
  Nickerson, Mihir Pant, Fernando Pastawski, Terry Rudolph, and Chris Sparrow.
\newblock ``Fusion-based quantum computation''~(2021).
\newblock  \href{http://arxiv.org/abs/2101.09310}{arXiv:2101.09310}.

\bibitem{bombin2021interleaving}
H\'{e}ctor Bomb\'{i}n, Isaac~H Kim, Daniel Litinski, Naomi Nickerson, Mihir
  Pant, Fernando Pastawski, Sam Roberts, and Terry Rudolph.
\newblock ``Interleaving: Modular architectures for fault-tolerant photonic
  quantum computing''~(2021).
\newblock  \href{http://arxiv.org/abs/2103.08612}{arXiv:2103.08612}.

\bibitem{Bourassa2021blueprintscalable}
J.~Eli Bourassa, Rafael~N. Alexander, Michael Vasmer, Ashlesha Patil, Ilan
  Tzitrin, Takaya Matsuura, Daiqin Su, Ben~Q. Baragiola, Saikat Guha, Guillaume
  Dauphinais, Krishna~K. Sabapathy, Nicolas~C. Menicucci, and Ish Dhand.
\newblock ``Blueprint for a {S}calable {P}hotonic {F}ault-{T}olerant {Q}uantum
  {C}omputer''.
\newblock \href{https://dx.doi.org/10.22331/q-2021-02-04-392}{{Quantum} {\bf
  5}, 392}~(2021).
\newblock  \href{http://arxiv.org/abs/2010.0290}{arXiv:2010.0290}.

\bibitem{Tzitrin2021}
Ilan Tzitrin, Takaya Matsuura, Rafael~N. Alexander, Guillaume Dauphinais,
  J.~Eli Bourassa, Krishna~K. Sabapathy, Nicolas~C. Menicucci, and Ish Dhand.
\newblock ``Fault-tolerant quantum computation with static linear optics''.
\newblock \href{https://dx.doi.org/10.1103/PRXQuantum.2.040353}{PRX Quantum
  {\bf 2}, 040353}~(2021).
\newblock  \href{http://arxiv.org/abs/2104.03241}{arXiv:2104.03241}.

\bibitem{bombin2DQuantumComputation2018}
H{\'e}ctor Bomb{\'\i}n.
\newblock ``{{2D}} quantum computation with {{3D}} topological codes''~(2018).

\bibitem{Dennis2002}
Eric Dennis, Alexei Kitaev, Andrew Landahl, and John Preskill.
\newblock ``{Topological quantum memory}''.
\newblock \href{https://dx.doi.org/10.1063/1.1499754}{Journal of Mathematical
  Physics {\bf 43}, 4452--4505}~(2002).
\newblock
  \href{http://arxiv.org/abs/quant-ph/0110143}{arXiv:quant-ph/0110143}.

\bibitem{bombinExactTopologicalQuantum2007}
H.~Bomb{\'\i}n and M.~A. {Martin-Delgado}.
\newblock ``Exact topological quantum order in {{D}} = 3 and beyond:
  {{Branyons}} and brane-net condensates''.
\newblock \href{https://dx.doi.org/10.1103/PhysRevB.75.075103}{Phys. Rev. B
  {\bf 75}, 075103}~(2007).

\bibitem{Vijay_2016}
Sagar Vijay, Jeongwan Haah, and Liang Fu.
\newblock ``Fracton topological order, generalized lattice gauge theory, and
  duality''.
\newblock \href{https://dx.doi.org/10.1103/physrevb.94.235157}{Physical Review
  B{\bf 94}}~(2016).

\bibitem{Yoshida_2013}
Beni Yoshida.
\newblock ``Exotic topological order in fractal spin liquids''.
\newblock \href{https://dx.doi.org/10.1103/physrevb.88.125122}{Physical Review
  B{\bf 88}}~(2013).

\bibitem{Haah_2011}
Jeongwan Haah.
\newblock ``Local stabilizer codes in three dimensions without string logical
  operators''.
\newblock \href{https://dx.doi.org/10.1103/physreva.83.042330}{Physical Review
  A{\bf 83}}~(2011).

\bibitem{brown2022conservation}
Benjamin~J. Brown.
\newblock ``Conservation laws and quantum error correction: towards a
  generalised matching decoder''~(2022).
\newblock  \href{http://arxiv.org/abs/2207.06428}{arXiv:2207.06428}.

\bibitem{panteleev2019degenerate}
Pavel Panteleev and Gleb Kalachev.
\newblock ``Degenerate quantum {LDPC} codes with good finite length
  performance''.
\newblock \href{https://dx.doi.org/10.22331/q-2021-11-22-585}{{Quantum} {\bf
  5}, 585}~(2021).
\newblock  \href{http://arxiv.org/abs/1904.02703}{arXiv:1904.02703}.

\bibitem{roffe2020decoding}
Joschka Roffe, David~R. White, Simon Burton, and Earl~T. Campbell.
\newblock ``Decoding across the quantum low-density parity-check code
  landscape''.
\newblock \href{https://dx.doi.org/10.1103/PhysRevResearch.2.043423}{Physical
  Review Research {\bf 2}, 043423}~(2020).
\newblock  \href{http://arxiv.org/abs/2005.07016}{arXiv:2005.07016}.

\bibitem{Tailoring2019}
David~K. Tuckett, Andrew~S. Darmawan, Christopher~T. Chubb, Sergey Bravyi,
  Stephen~D. Bartlett, and Steven~T. Flammia.
\newblock ``Tailoring surface codes for highly biased noise''.
\newblock \href{https://dx.doi.org/10.1103/physrevx.9.041031}{Physical Review X
  {\bf 9}, 041031}~(2019).
\newblock  \href{http://arxiv.org/abs/1812.08186}{arXiv:1812.08186}.

\bibitem{bealeQuantumErrorCorrection2018}
Stefanie~J. Beale, Joel~J. Wallman, Mauricio Guti{\'e}rrez, Kenneth~R. Brown,
  and Raymond Laflamme.
\newblock ``Quantum {{Error Correction Decoheres Noise}}''.
\newblock \href{https://dx.doi.org/10.1103/PhysRevLett.121.190501}{Phys. Rev.
  Lett. {\bf 121}, 190501}~(2018).

\bibitem{FernGeneralizedPerformance2006}
J.~Fern, J.~Kempe, S.N. Simic, and S.~Sastry.
\newblock ``Generalized performance of concatenated quantum codes—a dynamical
  systems approach''.
\newblock \href{https://dx.doi.org/10.1109/TAC.2006.871942}{IEEE Transactions
  on Automatic Control {\bf 51}, 448--459}~(2006).

\bibitem{GreenbaumCoherent2017}
Daniel Greenbaum and Zachary Dutton.
\newblock ``Modeling coherent errors in quantum error correction''.
\newblock \href{https://dx.doi.org/10.1088/2058-9565/aa9a06}{Quantum Science
  and Technology {\bf 3}, 015007}~(2017).

\bibitem{HuangCoherent2019}
Eric Huang, Andrew~C. Doherty, and Steven Flammia.
\newblock ``Performance of quantum error correction with coherent errors''.
\newblock \href{https://dx.doi.org/10.1103/PhysRevA.99.022313}{Phys. Rev. A
  {\bf 99}, 022313}~(2019).

\bibitem{Bravyi_2018}
Sergey Bravyi, Matthias Englbrecht, Robert König, and Nolan Peard.
\newblock ``Correcting coherent errors with surface codes''.
\newblock \href{https://dx.doi.org/10.1038/s41534-018-0106-y}{npj Quantum
  Information{\bf 4}}~(2018).

\bibitem{BrownCCZ}
Benjamin~J. Brown.
\newblock ``A fault-tolerant non-{C}lifford gate for the surface code in two
  dimensions''.
\newblock \href{https://dx.doi.org/10.1126/sciadv.aay4929}{Science Advances
  {\bf 6}, eaay4929}~(2020).
\newblock  \href{http://arxiv.org/abs/1903.11634}{arXiv:1903.11634}.

\bibitem{vasmerMorphingQuantumCodes2022}
Michael Vasmer and Aleksander Kubica.
\newblock ``Morphing {{Quantum Codes}}''.
\newblock \href{https://dx.doi.org/10.1103/PRXQuantum.3.030319}{PRX Quantum
  {\bf 3}, 030319}~(2022).

\bibitem{shirley2017fracton}
Wilbur Shirley, Kevin Slagle, Zhenghan Wang, and Xie Chen.
\newblock ``{Fracton Models on General Three-Dimensional Manifolds}''.
\newblock \href{https://dx.doi.org/10.1103/PhysRevX.8.031051}{Physical Review
  X{\bf 8}}~(2018).
\newblock  \href{http://arxiv.org/abs/1712.05892}{arXiv:1712.05892}.

\bibitem{vijay2016fracton}
Sagar Vijay, Jeongwan Haah, and Liang Fu.
\newblock ``Fracton topological order, generalized lattice gauge theory, and
  duality''.
\newblock \href{https://dx.doi.org/10.1103/PhysRevB.94.235157}{Physical Review
  B {\bf 94}, 235157}~(2016).
\newblock  \href{http://arxiv.org/abs/1603.04442}{arXiv:1603.04442}.

\bibitem{song_optimal_xcube}
Hao Song, Janik Schönmeier-Kromer, Ke~Liu, Oscar Viyuela, Lode Pollet, and
  M.~A. Martin-Delgado.
\newblock ``Optimal thresholds for fracton codes and random spin models with
  subsystem symmetry''~(2021).

\bibitem{doi:10.1080/14786435.2011.609152}
Claudio Castelnovo and Claudio Chamon.
\newblock ``{Topological quantum glassiness}''.
\newblock \href{https://dx.doi.org/10.1080/14786435.2011.609152}{Philosophical
  Magazine {\bf 92}, 304--323}~(2012).
\newblock  \href{http://arxiv.org/abs/1108.2051}{arXiv:1108.2051}.

\bibitem{yoshida2013exotic}
Beni Yoshida.
\newblock ``Exotic topological order in fractal spin liquids''.
\newblock \href{https://dx.doi.org/10.1103/PhysRevB.88.125122}{Phys. Rev. B
  {\bf 88}, 125122}~(2013).
\newblock  \href{http://arxiv.org/abs/1302.6248}{arXiv:1302.6248}.

\bibitem{Dua_2019}
Arpit Dua, Isaac~H. Kim, Meng Cheng, and Dominic~J. Williamson.
\newblock ``Sorting topological stabilizer models in three dimensions''.
\newblock \href{https://dx.doi.org/10.1103/physrevb.100.155137}{Physical Review
  B{\bf 100}}~(2019).

\bibitem{Dua_2020}
Arpit Dua, Pratyush Sarkar, Dominic~J. Williamson, and Meng Cheng.
\newblock ``Bifurcating entanglement-renormalization group flows of fracton
  stabilizer models''.
\newblock \href{https://dx.doi.org/10.1103/physrevresearch.2.033021}{Physical
  Review Research{\bf 2}}~(2020).

\bibitem{Takeda_2005}
Koujin Takeda, Tomohiro Sasamoto, and Hidetoshi Nishimori.
\newblock ``Exact location of the multicritical point for finite-dimensional
  spin glasses: a conjecture''.
\newblock \href{https://dx.doi.org/10.1088/0305-4470/38/17/004}{Journal of
  Physics A: Mathematical and General {\bf 38}, 3751--3774}~(2005).

\bibitem{nixon2021correcting}
Georgia~M Nixon and Benjamin~J Brown.
\newblock ``Correcting spanning errors with a fractal code''.
\newblock \href{https://dx.doi.org/10.1109/TIT.2021.3068359}{IEEE Transactions
  on Information Theory {\bf 67}, 4504--4516}~(2021).

\bibitem{higgotFragile2022}
Oscar Higgott, Thomas~C. Bohdanowicz, Aleksander Kubica, Steven~T. Flammia, and
  Earl~T. Campbell.
\newblock ``Fragile boundaries of tailored surface codes and improved decoding
  of circuit-level noise''~(2022).
\newblock  \href{http://arxiv.org/abs/2203.04948}{arXiv:2203.04948}.

\bibitem{fossorier1995soft}
Marc P.~C. Fossorier and Shu Lin.
\newblock ``Soft-decision decoding of linear block codes based on ordered
  statistics''.
\newblock \href{https://dx.doi.org/10.1109/18.412683}{IEEE Transactions on
  Information Theory {\bf 41}, 1379--1396}~(1995).

\bibitem{fossorier2001iterative}
Marc P.~C. Fossorier.
\newblock ``Iterative reliability-based decoding of low-density parity check
  codes''.
\newblock \href{https://dx.doi.org/10.1109/49.924874}{IEEE Journal on selected
  Areas in Communications {\bf 19}, 908--917}~(2001).

\bibitem{mackay2003information}
David J.~C. MacKay.
\newblock ``Information theory, inference and learning algorithms''.
\newblock Cambridge University Press. ~(2003).
\newblock
  url:~\href{http://www.inference.org.uk/mackay/itila/book.html}{http://www.inference.org.uk/mackay/itila/book.html}.

\bibitem{Raveendran2021trappingsetsof}
Nithin Raveendran and Bane Vasi{\'{c}}.
\newblock ``Trapping {S}ets of {Q}uantum {LDPC} {C}odes''.
\newblock \href{https://dx.doi.org/10.22331/q-2021-10-14-562}{{Quantum} {\bf
  5}, 562}~(2021).

\bibitem{poulin2006optimal}
David Poulin.
\newblock ``Optimal and efficient decoding of concatenated quantum block
  codes''.
\newblock \href{https://dx.doi.org/10.1103/PhysRevA.74.052333}{Physical Review
  A {\bf 74}, 052333}~(2006).
\newblock
  \href{http://arxiv.org/abs/quant-ph/0606126}{arXiv:quant-ph/0606126}.

\bibitem{poulin2008iterative}
David Poulin and Yeojin Chung.
\newblock ``On the iterative decoding of sparse quantum codes''.
\newblock Quantum Information \& Computation {\bf 8}, 987–1000~(2008).
\newblock  \href{http://arxiv.org/abs/0801.1241}{arXiv:0801.1241}.

\bibitem{babar2015fifteen}
Zunaira Babar, Panagiotis Botsinis, Dimitrios Alanis, Soon~Xin Ng, and Lajos
  Hanzo.
\newblock ``Fifteen years of quantum {LDPC} coding and improved decoding
  strategies''.
\newblock \href{https://dx.doi.org/10.1109/ACCESS.2015.2503267}{IEEE Access
  {\bf 3}, 2492--2519}~(2015).

\bibitem{rigby2019modified}
Alex Rigby, J.~C. Olivier, and Peter Jarvis.
\newblock ``Modified belief propagation decoders for quantum low-density
  parity-check codes''.
\newblock \href{https://dx.doi.org/10.1103/PhysRevA.100.012330}{Physical Review
  A {\bf 100}, 012330}~(2019).
\newblock  \href{http://arxiv.org/abs/1903.07404}{arXiv:1903.07404}.

\bibitem{kuo2021exploiting}
Kao-Yueh Kuo and Ching-Yi Lai.
\newblock ``Exploiting degeneracy in belief propagation decoding of quantum
  codes''~(2021).
\newblock  \href{http://arxiv.org/abs/2104.13659}{arXiv:2104.13659}.

\bibitem{wang2012enhanced}
Yun-Jiang Wang, Barry~C. Sanders, Bao-Ming Bai, and Xin-Mei Wang.
\newblock ``Enhanced feedback iterative decoding of sparse quantum codes''.
\newblock \href{https://dx.doi.org/10.1109/TIT.2011.2169534}{IEEE transactions
  on information theory {\bf 58}, 1231--1241}~(2012).
\newblock  \href{http://arxiv.org/abs/0912.4546}{arXiv:0912.4546}.

\bibitem{kuo2020refined}
Kao-Yueh Kuo and Ching-Yi Lai.
\newblock ``Refined belief-propagation decoding of quantum codes with scalar
  messages''.
\newblock In 2020 IEEE Globecom Workshops.
\newblock \href{https://dx.doi.org/10.1109/GCWkshps50303.2020.9367482}{Pages
  1--6}.
\newblock ~(2020).
\newblock  \href{http://arxiv.org/abs/2102.07122}{arXiv:2102.07122}.

\bibitem{liu2019neural}
Ye-Hua Liu and David Poulin.
\newblock ``Neural belief-propagation decoders for quantum error-correcting
  codes''.
\newblock \href{https://dx.doi.org/10.1103/PhysRevLett.122.200501}{Physical
  Review Letters {\bf 122}, 200501}~(2019).
\newblock  \href{http://arxiv.org/abs/1811.07835}{arXiv:1811.07835}.

\bibitem{roffe2021bposd}
Joschka Roffe.
\newblock ``{BP+OSD}: A decoder for quantum {LDPC} codes''.
\newblock \url{https://pypi.org/project/bposd/}~(2021).

\bibitem{wang2010}
D.~S. Wang, A.~G. Fowler, A.~M. Stephens, and L.~C.~L. Hollenberg.
\newblock ``Threshold error rates for the toric and planar codes''.
\newblock Quantum Information \& Computation {\bf 10}, 456--469~(2010).
\newblock  \href{http://arxiv.org/abs/0905.0531}{arXiv:0905.0531}.

\bibitem{fowler2015}
Austin~G. Fowler.
\newblock ``Minimum weight perfect matching of fault-tolerant topological
  quantum error correction in average {{O}}(1) parallel time''.
\newblock Quantum Information \& Computation {\bf 15}, 145--158~(2015).
\newblock  \href{http://arxiv.org/abs/1307.1740}{arXiv:1307.1740}.

\bibitem{higgott2021pymatching}
Oscar Higgott.
\newblock ``Pymatching: A python package for decoding quantum codes with
  minimum-weight perfect matching''.
\newblock \href{https://dx.doi.org/10.1145/3505637}{ACM Transactions on Quantum
  Computing{\bf 3}}~(2022).

\bibitem{lemon}
Bal\'{a}zs Dezs, Alp\'{a}r J\"{u}ttner, and P\'{e}ter Kov\'{a}cs.
\newblock ``Lemon - an open source c++ graph template library''.
\newblock \href{https://dx.doi.org/10.1016/j.entcs.2011.06.003}{Electron. Notes
  Theor. Comput. Sci. {\bf 264}, 23–45}~(2011).

\bibitem{ChubbFlammia2018}
Christopher~T. Chubb and Steven~T. Flammia.
\newblock ``Statistical mechanical models for quantum codes with correlated
  noise''.
\newblock \href{https://dx.doi.org/10.4171/AIHPD/105}{Annales de l’Institut
  Henri Poincaré D {\bf 8}, 269--321}~(2018).
\newblock  \href{http://arxiv.org/abs/1809.10704}{arXiv:1809.10704}.

\bibitem{Chubb2021}
Christopher~T. Chubb.
\newblock ``General tensor network decoding of 2d pauli codes''~(2021)
  \href{http://arxiv.org/abs/2101.04125}{arXiv:2101.04125}.

\end{thebibliography}

\clearpage
\appendix

\section{Proof of 50\% threshold for the 3D surface code on the checkerboard lattice}
\label{sec:rhombic-thresholdproof-appendix}

We presented a Clifford deformation of the checkerboard lattice surface code which consists of applying a Hadamard operation on half of the vertical qubits, 
in a three-dimensional checkerboard manner (see \cref{fig:rhombic-checkerboard-deformation}). For this Clifford-deformed code, cube stabilizer generators are violated under $Z$ errors on the Clifford-deformed edges, while the triangle stabilizer generators are violated by $Z$ errors on the remaining edges. Using this fact, one can show that the products of cubes and the products of triangles along the diagonals of the code are effectively identity for pure $Z$ errors.
The presence of these linear symmetries gives rise to the following theorem:

\begin{theorem}
The Clifford-deformed checkerboard lattice surface code has a $50\%$ threshold error rate under pure $Z$ noise
\end{theorem}

\begin{proof}

The Clifford deformation we consider for the checkerboard lattice surface code consists of applying a Hadamard on half of the vertical qubits, 
in a 3D checkerboard fashion (see \cref{fig:rhombic-checkerboard-deformation}).
In this new code, under pure $Z$ noise, cube stabilizers can only be excited by errors acting on the Clifford-deformed qubits, 
while triangle stabilizers can be excited by any of the remaining qubits.
We now show that this code has a $50\%$ threshold error rate under pure $Z$ noise.

We start by decoding the cube stabilizers. 
Due to the Clifford deformation, these stabilizers are now effectively weight-2 at infinite bias, 
and are involved in some linear symmetries represented in \cref{fig:rhombic-symmetries-a}. 
We can therefore decode the cubes by performing matching along each symmetry line.

We then tackle the triangle stabilizers. 
To decode them, we can perform matching along the two linear symmetries represented 
in \cref{fig:rhombic-symmetries-b,fig:rhombic-symmetries-c}.
One symmetry allows us to decode all the remaining vertical qubits, and the other all the horizontal qubits.

Since all the steps involve decoding a polynomial number of repetition codes (in the lattice size $L$), and the
probability of success is lower-bounded by the probability of correctly decoding all these repetition codes,
it shows that this decoding strategy leads to $50\%$ threshold error rate.

\begin{figure}[t]
    \centering
    \subfloat[]{
        \includegraphics{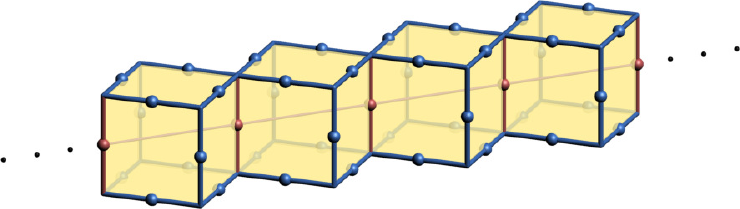}
        \label{fig:rhombic-symmetries-a}
    }
    \hfill
    \subfloat[]{
        \includegraphics{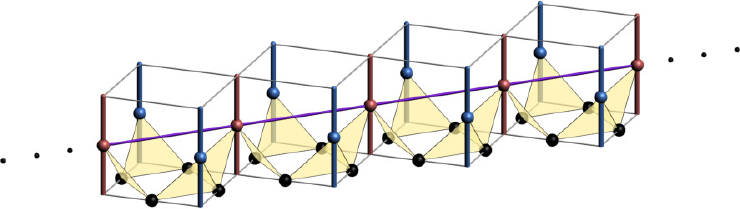}
        \label{fig:rhombic-symmetries-b}
    }
    \subfloat[]{
        \includegraphics{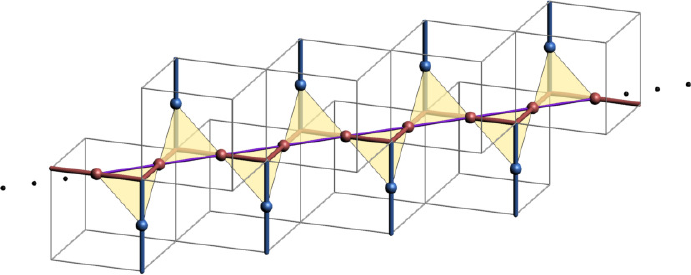}
        \label{fig:rhombic-symmetries-c}
    }
    \caption{Linear symmetries of the checkerboard lattice surface code.
    \textbf{(a)} Symmetry of the cube stabilizers. Due to the Clifford deformation, the cube stabilizers become weight-2 and are supported on
    two diagonally-opposite edges (red). Multiplying cubes diagonally effectively gives the identity. 
    Matching along this line allows us to decode all the Clifford-deformed vertical qubits.
    \textbf{(b)} Symmetry of the triangle stabilizers involving vertical qubits only. 
    The product of four triangle stabilizers within a cube gives a four-body stabilizer supported on the four vertical qubits of the cube
    (blue and red). Due to the Clifford deformation, this becomes a two-body term, supported on two diagonally opposite qubits (blue).
    Multiplying these weight-2 stabilizers on a diagonal line (purple) gives the identity.
    Matching along all such lines allows to decode all the undeformed vertical qubits.
    \textbf{(c)} Symmetry of the triangle stabilizers involving horizontal qubits only. 
    The triangles containing a Clifford-deformed qubit become two-body terms after applying the Clifford deformation. 
    Multiplying them along a line (purple) gives the identity.
    All the other horizontal qubits of the cube are involved in a similar symmetry. 
    Matching along all these symmetries allows to decode all the horizontal qubits of the code.
    }%
    \label{fig:rhombic-symmetries}
\end{figure}
\end{proof}


\section{Decoders}
\label{sec:decoders-appendix}

\subsection{BP-OSD}\label{sec:bposd-appendix}

\begin{figure}
    \centering
    \includegraphics[width=0.8\textwidth]{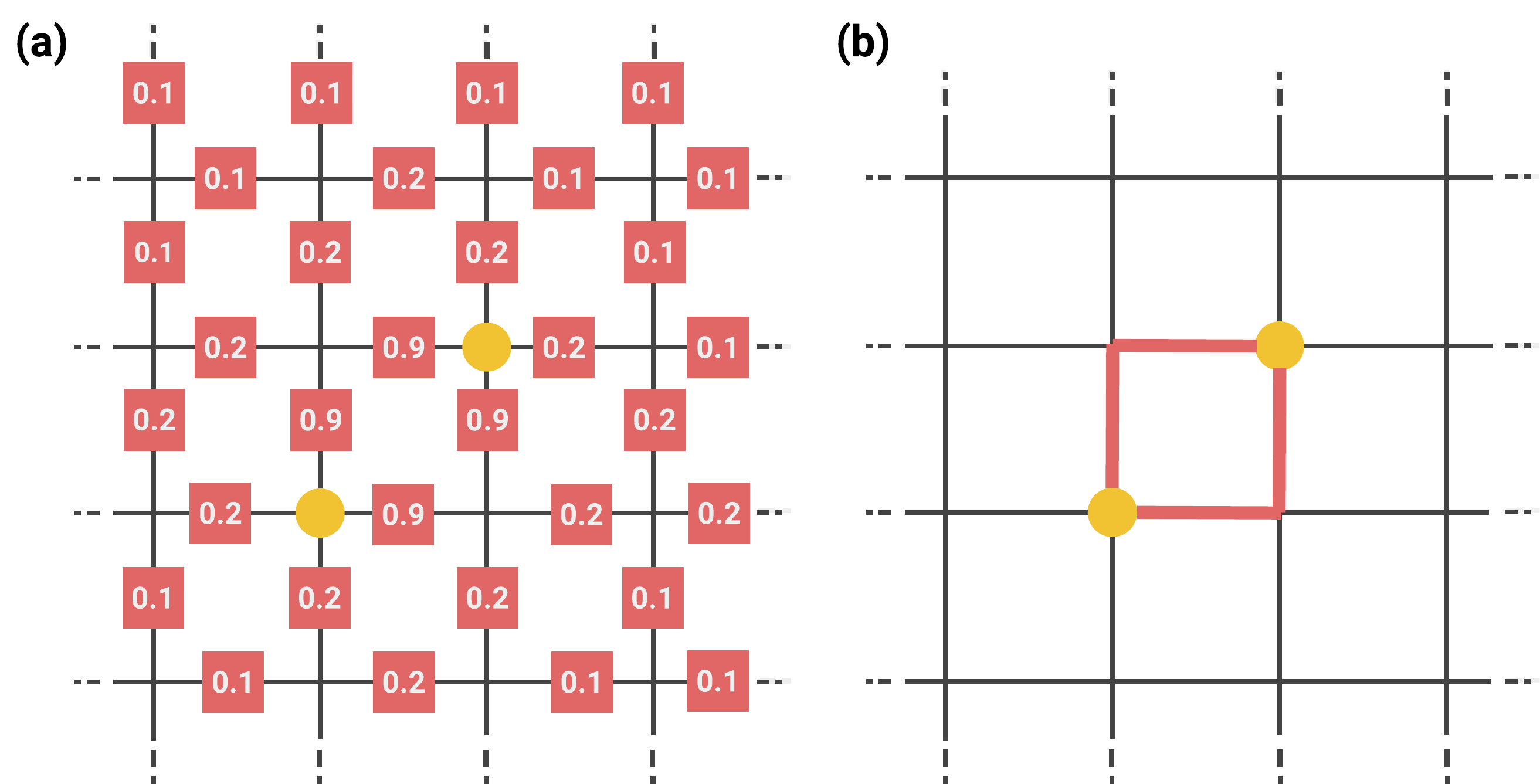}
    \caption{Illustration of the BP-OSD decoder on a 2D surface code. 
    \textbf{(a)} The belief propagation algorithm takes as input a syndrome $\vb{s}$ (yellow dots) and an error model, 
    and computes a probability $P(e_i|\vb{s})$ of error on each individual qubit (red squares). 
    When the solution is degenerate (e.g.\ two chains of errors have minimal weight), a high probability ($0.9$ shown in figure) 
    is assigned to all the solutions.
    \textbf{(b)} Split-belief phenomenon. In the basic version of belief propagation, a correction operator (red edges) is applied to 
    all qubits $i$ having $P(e_i|\vb{s}) > 0.5$. 
    In degenerate cases, it can produce an invalid correction with some defects remaining. 
    The OSD algorithm consists of solving the parity-check equation $\vb{H}\vb{e}=\vb{s}$ for the most probable set of errors,
    as found by belief propagation, guaranteeing that the final corrected state lives in the codespace.}
    \label{fig:bp-osd-explanation}
\end{figure}

The belief propagation with ordered statistics decoder (BP-OSD) is a generic decoder for quantum LDPC codes. 
Based on a classical technique to improve the iterative decoding of linear codes \cite{fossorier1995soft, fossorier2001iterative}, 
it was introduced to the quantum domain by Panteleev and Kalachev \cite{panteleev2019degenerate} and has been shown
to have high performance on a large class of LDPC codes, including topological codes \cite{roffe2020decoding, quintavalle2021single}. 
BP-OSD is built from two components: the belief propagation (BP) algorithm, which estimates the probability for each qubit to have an error,
and the ordered statistics decoder (OSD), which takes these probabilities as input and proposes a correction that fits the syndrome. 
It is particularly well-adapted to the decoding of Clifford-deformed codes under biased noise, as it naturally takes into account 
the non-uniform probability of errors along the different axes in the Clifford-deformed noise model.

\paragraph{Belief propagation}
The belief propagation decoder is one of the most commonly used decoders for classical LDPC codes \cite{mackay2003information}. 
It is an inference algorithm that computes an approximation of the probabilities $P(e_i|\bm{s})$ that an error has occurred 
on each bit $i$ given a syndrome $\bm{s}$. 
A correction operator is then applied to all the bits $i$ such that $P(e_i|\bm{s}) > 0.5$. 
While computing this marginal probability involves in principle summing over an exponential number of terms, 
belief propagation exploits the fact that for LDPC codes, this sum can be factored into a small number of terms. 
It then uses an algorithm called the \emph{product-sum algorithm} (or its variant the \emph{min-sum algorithm}) 
to calculate this sum, by iteratively passing messages between parity checks and data bits.
Belief propagation can be shown to converge to the exact marginal distribution when the Tanner graph is a tree.
For more general Tanner graphs that can contain loops, it is used as a heuristic algorithm to approximate the distribution, 
and is sometimes called \emph{loopy belief propagation}. While the approximation is often acceptable when the \emph{girth}
\footnote{The girth of a graph is the size of its shortest cycle.} of the graph is large, the presence of short-cycles tends to be detrimental to the performance
of BP~\cite{mackay2003information, Raveendran2021trappingsetsof}.

Several methods have been proposed in the literature to generalize belief propagation to quantum codes 
\cite{poulin2006optimal, poulin2008iterative, babar2015fifteen, rigby2019modified, roffe2020decoding, kuo2021exploiting, wang2012enhanced}. 
For instance, one can decode $X$ and $Z$ errors separately using the classical version of BP.
The potential correlations between $X$ and $Z$ errors can be taken into account by first decoding $X$ errors, 
adjusting the channel probabilities based on the correction, and decoding $Z$ errors with this adjusted probability, 
as proposed in Ref. \cite{rigby2019modified}.
It has also been proposed to send vector instead of scalar messages, to compute the probability $P(e_i=W|\bm{s})$ 
that a Pauli error $W \in \{I, X, Y, Z\}$ has occurred on each qubit $i$ \cite{poulin2008iterative}. 
However, this results in an increase in complexity compared to the original BP algorithm, 
and sometimes reduced performance due to the presence of shorter cycles in the whole Tanner graph compared to the $X$ and $Z$ ones.
A simplified message-passing rule was proposed in Ref. \cite{kuo2020refined} to reduce this complexity while guaranteeing the same output as 
the original version, but the presence of short cycles is still hindering its performance.
In this work, we chose to decode $X$ and $Z$ errors separately.

Apart from the presence of short cycles in quantum Tanner graphs, a major problem with BP decoding of quantum codes 
is the \emph{degeneracy problem}, also called \emph{split-belief} phenomenon \cite{poulin2008iterative}. 
Indeed, in quantum codes, a syndrome can often be generated by several equally likely combinations of errors. 
By symmetry, the BP algorithm outputs the same probability for all these errors, and if they are all higher than $0.5$, 
it applies a correction operator to all these equally likely errors, resulting in an invalid correction that does not fit the syndrome.
An illustration of a split-belief problem is shown in \cref{fig:bp-osd-explanation}.

To mitigate the degeneracy problem, several solutions have been proposed in the literature, such as breaking the degeneracy 
with random noise \cite{poulin2008iterative}, adjusting the error probabilities when the decoder fails \cite{wang2012enhanced}, 
using a neural network to learn the BP procedure, with a loss function tailored to avoid degeneracies \cite{liu2019neural}, 
using previous messages in the message-passing update rule \cite{kuo2021exploiting}, 
or complementing the BP decoder with a second decoder such as the ordered statistics decoder (OSD) \cite{roffe2020decoding}. 
Since OSD has recently been shown to outperform other methods for many different codes \cite{panteleev2019degenerate}, 
we are using this solution in our work.

\paragraph{Ordered statistics decoding}
For any classical linear code with a parity-check matrix $\vb{H}$, the following equation, called the \emph{syndrome equation}, holds:
\begin{align}
    \vb{H} \vb{e} = \vb{s}
\end{align}
Since many errors can correspond to a given syndrome, $\vb{H}$ is not directly invertible. 
The idea of OSD is to only solve the system for the most-likely errors, as given by the BP algorithm. 
More precisely, we sort the columns of $\vb{H}$ by increasing probability of error, and eliminate them one-by-one in that order until the system
is full-rank. We then solve the reduced system to find a set of errors that respect the syndrome equation. 
The remaining qubits can either be set to have no error, or be searched-over for a better correction using some heuristics
\cite{panteleev2019degenerate, roffe2020decoding}.

In the quantum setting, a similar syndrome equation holds, where the parity-check matrix and the error vector can either be 
written over the field $GF(4)$ or in a symplectic form. 
For CSS codes, we can also consider $X$ and $Z$ errors separately and write a syndrome equation for each of them:
\begin{aligns}
    \vb{H_Z} \vb{e_X} &= \vb{s_Z} \\
    \vb{H_X} \vb{e_Z} &= \vb{s_X}
\end{aligns}
A classical OSD algorithm can then be applied separately for each equation.

In this work, we use BP-OSD through the Python library \texttt{bposd} developed by J. Roffe \cite{roffe2021bposd}. 
In particular, we use the min-sum algorithm for BP and the \emph{combination sweep strategy} (to order $50$) described 
in Ref. \cite{roffe2020decoding} to search over the remaining errors in OSD.

\paragraph{Limitations of BP-OSD} While OSD turns the output of BP into a valid correction, the algorithm still suffers
from the main drawbacks of loopy belief propagation discussed, such as short cycles and error degeneracies.
For instance, the effect of short cycles can be seen when decoding long strings on the 2D surface code. 
When the size of a string is higher than $8$ (the girth of the surface code when considering $X$ and $Z$ errors separately),
short cycles tend to deteriorate the messages passed between the two distant defects. 
This phenomenon, called \emph{bounded information spread}, has been documented in the literature
for the 2D surface and color codes \cite{higgottImprovedSingleshotDecoding2022}. 
As a result, decoding topological codes of large sizes is often harder for BP-OSD than for small sizes. 
This can be observed in some threshold error rate plots where the apparent threshold error rate seems to shrink 
when increasing the system size. Therefore, finite-size approximations of the BP-OSD threshold error rate might not reflect the true
threshold error rate, obtained when taking the system size to infinity. 
Examples of this finite-size effect on 3D codes are given in \cref{sec:finbiasthresholds}.

Apart from short cycles, error degeneracies can also have a negative effect on BP-OSD.
In general, split-belief problems appear around stabilizers with even weight.
Indeed, any error supported on half of an even-weight stabilizer gives the same syndrome after the
application of the stabilizer, while the new error has the same weight.
This argument has been used to explain why BP-OSD performs poorly on the 2D surface code,
while performing well on the 3D surface code,
observing that the smallest split-belief appears for weight-2 errors on the 2D surface code, 
but on weight-3 errors on the loop sector of the 3D surface code \cite{higgottImprovedSingleshotDecoding2022}.
We call the size of the smallest error that causes a split-belief the \emph{split-belief number} of the code.
It can be calculated by taking the smallest even-weight stabilizer and dividing by two.
Examples of split-belief numbers for different 3D codes are given in \cref{tab:girth-and-split-belief}.

\subsection{Sweep-matching decoder\label{sec:sweepmatch-appendix}}
The sweep-matching decoder on the CSS 3D surface code uses minimum weight
perfect matching (MWPM)~\cite{Dennis2002,wang2010,fowler2015} between vertices to correct for point-like syndromes
and the Sweep decoder~\cite{kubica2019cellular,vasmer2021sweep} over face syndromes to correct for string-like syndromes.
The two sectors are decoded independently
and the procedure can be generalized for the Clifford-deformed code by using the
Clifford-deformed stabilizer generators instead.

Implementations of MWPM are easily applicable in 3D for the Clifford-deformed code
and for biased noise by adjusting the graph weights of the match to match
the known noise parameters.
In this work, the Python library
\texttt{PyMatching}~\cite{higgott2021pymatching,lemon} is used for fast MWPM.
\vspace{5mm}
\begin{figure}[ht]
    \begin{center}
        \includegraphics[scale=1.2]{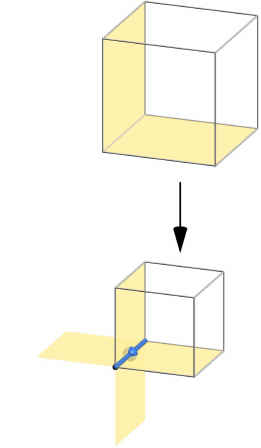}
        \includegraphics[scale=1.2]{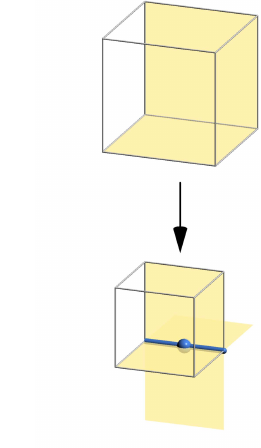}
        \includegraphics[scale=1.2]{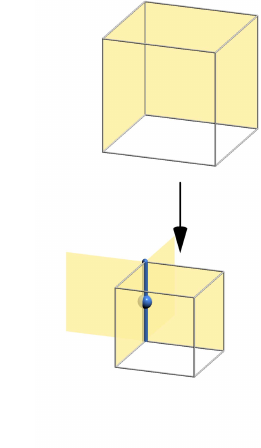}
        \includegraphics[scale=1.2]{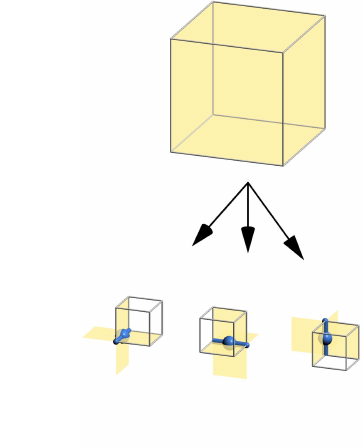}
    \end{center}
    \caption{The greedy sweep rule as applied to a single cell.
    The top figures enumerate the four possible
    non-trivial initial syndrome configurations on the 3 faces of a cell
    adjoining the vertex that is furthest from the sweep direction,
    where faces with nontrivial syndromes are shaded yellow.
    The corresponding corrections are shown in the figures below,
    where the $Z$ edge operator to be applied as a correction is shown in blue
    and the corresponding syndromes to be flipped and updated are shaded yellow.
    In the right-most initial configuration where non-trivial syndromes are on
    all 3 faces,
    the correction to apply is chosen randomly out of the 3 possible
    corrections.
    This rule is greedily applied to all cells at once and repeated
    $T_{\textrm{max}}$ times.}%
    \label{fig:sweep}
\end{figure}
The sweep decoder is a local cellular automaton decoder,
meaning that it is an iterative algorithm where,
at each step, a correction operator is computed locally according to the current syndrome using a cellular automaton rule.
To be of use, such a decoder should be able to eliminate every syndrome after a number of steps that is polynomial in the size of the code.

The sweep decoder is based on a cellular automaton called the sweep rule \cite{kubica2019cellular}.
We now briefly review the sweep rule in the special case of the simple cubic lattice.
We start by choosing a spatial direction, defined by a 3D vector $\va{v}$, called the sweep direction, with the only condition that it is not parallel to an edge of the lattice.
In practice we choose the sweep direction from one of eight possibilities $(\pm 1, \pm 1, \pm 1)$.
As illustrated in \cref{fig:sweep}, we then apply the following rule at each iteration, simultaneously for all the vertices:
\begin{enumerate}
    \item Find the three oriented lattice edges, $\va{e}_1$, $\va{e}_2$ and $\va{e}_3$, pointing away from the vertex and in the same direction as $\va{v}$, i.e., such that $\va{e}_i \cdot \va{v} > 0$.
    Each pair of edges corresponds to a face of the lattice.
    \item If two of these faces are excited, then apply a $Z$ operator to the intersecting edge. If all three faces are excited, then apply a $Z$ operator to a random edge among $\va{e}_1$, $\va{e}_2$, and $\va{e}_3$. Otherwise, do nothing.
\end{enumerate}
In the sweep decoder, we apply the sweep rule $T_{\mathrm{max}} = O(L)$ times, where $L$ is the linear lattice size.
For lattices with boundaries we run the decoder multiple times using different sweep directions, as described in~\cite{vasmer2021sweep}.
The sweep decoder can fail in two ways: either if the product of the original error and the operators applied by the sweep rule is a non-trivial logical operator, or if the syndrome is non-trivial after $T_{\mathrm{max}}$ applications of the rule.

We implemented the sweep-matching decoder and simulated its performance for 3D surface codes defined on cubic lattices, with and without boundaries. The code is available online\footnote{See the repository at \url{https://github.com/panqec/panqec} for the implementation of the sweep-matching decoder.}.

\section{Numerical simulation details}\label{sec:numerics-appendix}
To compute the threshold error rate of the different codes, we simulate up to
$n_{\textrm{trials}}=10,000$
trials for each physical error rate $p$ and for each lattice size $L$.
Values of $p$ were taken in intervals between $0$ and $0.5$,
with a maximum step size of $0.01$,
while lattice sizes were chosen to be greater than $L=9$ and up to $L=21$,
using at least 3 values of $L$ for each $p$.
We extract the threshold error rate from crossover plots
using a common finite-size scaling regression analysis~\cite{Dennis2002,ChubbFlammia2018,Chubb2021}.
The logical error simulation data is fitted to
the following ansatz for the physical error rate $p_L(p, L)$ as a function
of the physical error rate $p$ and system size $L$.
\begin{align}
    p_L &= A + Bx + Cx^2,\label{eqn:ansatz}\\
    x &= (p-p_{\text{th}}) L^{1/\nu},
\end{align}
where $p_{\text{th}}$ is the threshold error rate we seek
to evaluate, $\nu$ is a critical exponent and, $A,B,C$ are coefficients of the
quadratic ansatz,
all of which are free parameters to be determined by fitting to the data.
Here $x$ is termed the \emph{rescaled physical error rate},
which is zero at the phase transition $p=p_{\textrm{th}}$.
That $p_L$ is a quadratic function of $x$ is only expected to be a valid
approximation near this phase transition for $x$,
so only data points with physical error rates close to the phase transition were
used for the fitting.

For each given physical error rate $p$ and system size $L$,
suppose that
$n_{\textrm{fail}}$ trials out of $n_{\textrm{trials}}$ trials
result in a logical error after running the numerical simulations of sampling
the noise model, syndrome extraction and decoding.
The logical error rate can then be estimated by
$p_L=n_{\textrm{fail}}/n_{\textrm{trials}}$.

Using these estimated logical error rates,
we run an optimization procedure to obtain the set of free parameters
$(p_{\text{th}}, \nu, A, B, C)$ that fits the data the best,
as measured by minimizing the mean-squared error.

Uncertainties for the threshold error rate estimate $p_{\textrm{th}}$ are
calculated using the following bootstrap resampling method.

The first step is to obtain a distribution to sample for estimates of the
logical error rate $p_L$ for each $(p, L)$ to account for the finite number of
trials $n_{\textrm{trials}}$.
Starting from a uniform prior distribution before taking into account the number
of trials and failures,
the posterior distribution for $p_L$ is a Beta distribution with
\begin{align}
    p_L\sim\textrm{Beta}\left(n_{\textrm{trials}} - n_{\textrm{fail}} + 1,
    n_{\textrm{fail}} + 1\right),
\end{align}
where $\textrm{Beta}(a, b)$ is a probability distribution with support over the
interval $[0,1]$ and probability density function
\begin{align}
    f(x; a, b)
    = \frac{%
    \Gamma\left( a + b \right)
    x^{a - 1} {\left( 1 - x \right)}^{b - 1}%
    }{%
    \Gamma\left(a\right) \Gamma\left(b\right)
    }
\end{align}
for real parameters $a,b$.
Here, $\Gamma$ is the Gamma function with the property that $\Gamma(n)=(n-1)!$.
Thus the logical error rates $p_L(p, L)$ for each $p$ and $L$ can be sampled
independently from these posterior distributions to produce a resampled set of
$p_L$ values to use for fitting.
The second step is to reflect the uncertainty that arises from the choice of
data points,
which may be done by resampling with replacement the set of $(p,L)$ pairs
to use.

With these two sources of uncertainty accounted for by resampling
$n_{\textrm{bs}}=100$ times,
the least-squares fit of $p_L(p,L)$ can be done repeatedly on the resampled
data to produce a set of $n_{\textrm{bs}}$ best-fit parameters
$\left\{(p_{\text{th}}, \nu, A, B, C) \right\}$ for each resampling.
This collection of best-fit parameters can be used to produce error bars on the
threshold error rate $p_{\textrm{th}}$ by taking the $1\sigma$-bounds of the
resampled threshold error rate estimates to be interpreted as a credible
interval.
Note that these uncertainty bounds need not be symmetrical in upper and lower
directions,
as seen in \cref{tab:toric,tab:rhombic,tab:xcube}.
An example of this fitting is shown in \cref{fig:collapse},
with the bootstrapped uncertainty estimates shown pink and the best-fit value of
$p_{\textrm{th}}$ marked with the red dashed vertical line.
The validity of the ansatz may be vindicated by visual inspection of the
so-called data collapse plot of the logical error rate $p_L$ over the rescaled
physical error rate $x=(p-p_{\text{th}}) L^{1/\nu}$
where all data points collapse onto the fit line per the quadratic ansatz
in \cref{eqn:ansatz} to within reasonable bounds,
as quantified by the $1\sigma$ envelope above and below the fit line.
This ensures that the data points chosen for the finite-size scaling were chosen
sufficiently close to the critical point such that the ansatz is valid.
\begin{figure}[t]
    \centering
    \includegraphics[width=0.7\textwidth,page=52]{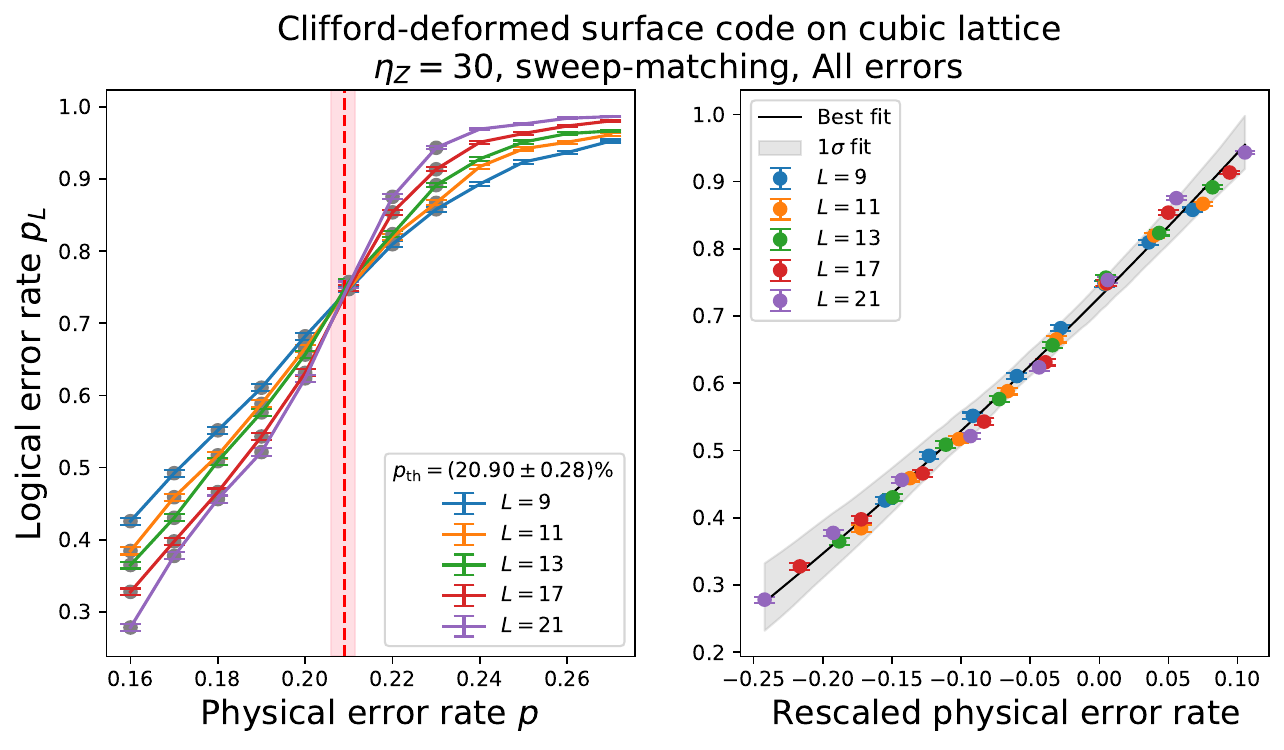}
    \caption{(left) Total logical error rate $p_L$ vs the physical error rate
    $p$ for $Z$-biased noise with $\eta=30$ decoded with the sweep-matching decoder.
    (right) Data collapse onto the quadratic finite-size-scaling ansatz with
    bootstrapped error bars and $1\sigma$ fit bounds.}
    \label{fig:collapse}
\end{figure}

To verify the reliability of the estimated threshold error rates,
the average $X$ and $Z$ logical failure rate over every logical qubit is used
as the logical error rate and subjected to the same analysis to extract corresponding
threshold error rates.
This is important since, due to finite-size effects,
the apparent threshold error rate as determined by the total logical error rate
may be higher than the threshold error rates determined by the logical $X$ and
$Z$ error rates.
For the case of the X-cube model where the number of logical
operators increases with the code distance,
the logical $X$ error rate is determined by taking the average logical $X$
error rate over all logical qubits. 
The logical $Z$ error rate is calculated analogously.

We use this procedure to compute the threshold error rate of both the CSS and
Clifford-deformed codes, for bias ratios
$\eta_Z \in \{0.5, 1, 3, 10, 30, 100, \infty \}$,
sampling more where interesting features are to be elucidated.
Representative examples of crossover plots of the logical error rate over
physical error rate along with the ansatz-fitting for both $X$ and $Z$ logical
errors are separately given in
\cref{fig:collapse2}.

The above procedure produces best-fit estimates and credible intervals for
the threshold error rate $p_{\textrm{th}}$ with respect to the total logical
error rate, logical $X$ errors and logical $Z$ errors,
which may differ significantly.
To be conservative, the reported threshold error rate is the minimum of these estimates,
as determined by which $1\sigma$ (68\% equal-tail) credible interval has
the lowest lower bound.

\begin{figure}[t]
    \centering
    \subfloat[]{
        \includegraphics[width=0.7\textwidth,page=53]{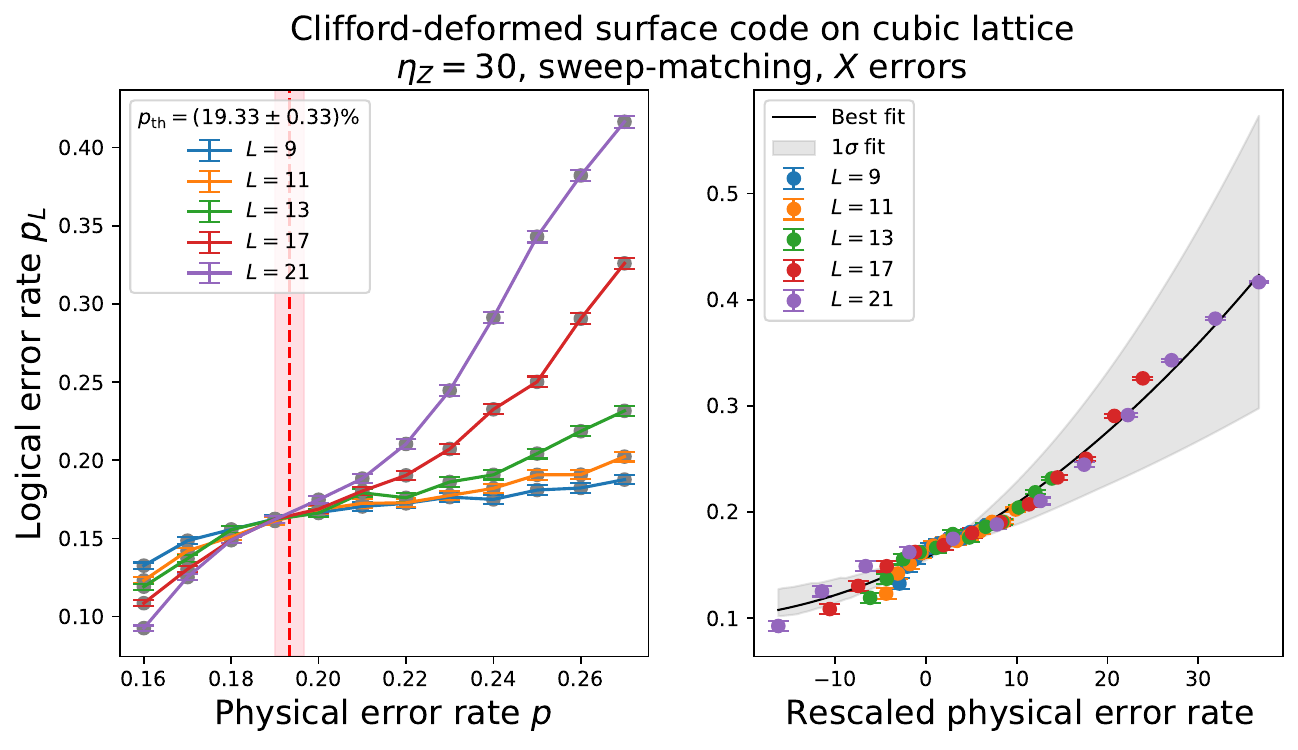}%
        \label{fig:collapse2-a}
    }\\
    \subfloat[]{
        \includegraphics[width=0.7\textwidth,page=54]{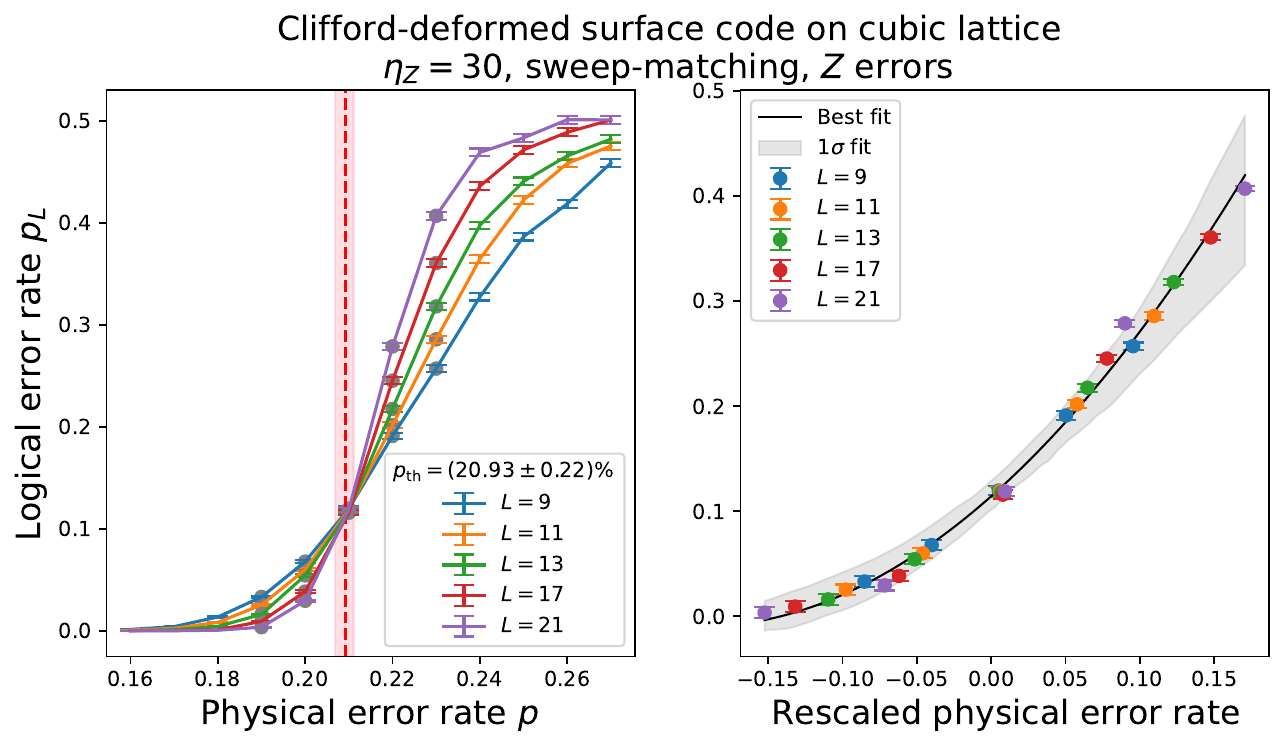}%
        \label{fig:collapse2-b}
    }
    \caption{Examples of data collapse plots in terms of logical $X$ and $Z$
    errors for the 3D surface code under $Z$-biased noise with $\eta_Z=30$ and
    decoding by sweep-matching. Note that these are the same parameters as those in
    \cref{fig:collapse}, but the minimum threshold error rate is slightly lower
    than that estimated using the total error rate in \cref{fig:collapse} and
    is thus a more conservative estimate.}%
    \label{fig:collapse2}
\end{figure}

\end{document}